\documentclass[a4paper,10pt]{article}
\usepackage{graphicx,amsmath,amssymb,upgreek,wasysym}
\usepackage{subcaption}
\usepackage{paralist}
\usepackage{color}
\usepackage{tikz,pgfkeys}
\usepackage{tkz-graph}
\usepackage{amsthm}
\usetikzlibrary{arrows,shapes}
\usepackage{charter,eulervm}
\usepackage{setspace}
\usepackage{multirow}
\usepackage[pdftex,breaklinks,colorlinks,
linkcolor=blue,
citecolor=blue,
urlcolor=blue]{hyperref}

\newtheorem{theorem}{Theorem}[section]
\newtheorem{lemma}[theorem]{Lemma}
\newtheorem{corollary}[theorem]{Corollary}
\newtheorem{proposition}[theorem]{Proposition}

\newtheorem{reduction}{Reduction}
\newtheorem{claim}{Claim}
\newenvironment{definition}{$\;$\newline \noindent {\bf
    Definition}$\;$}{$\;$\newline}
\def\boxit#1{\vbox{\hrule\hbox{\vrule\kern4pt
  \vbox{\kern1pt#1\kern1pt}
\kern2pt\vrule}\hrule}}

\newcommand{\fis}{forbidden induced subgraph}
\newcommand{\fiss}{forbidden induced subgraphs}
\newcommand{\mfis}{minimal forbidden induced subgraph}
\newcommand{\mfiss}{minimal forbidden induced subgraphs}
\newcommand{\mfset}{minimal forbidden set}
\newcommand{\mfsets}{minimal forbidden sets}
\newcommand{\smodule}{module}

\begin{document}

\title{\bf Interval Deletion is Fixed-Parameter
  Tractable\thanks{Supported in part by the European Research Council
    (ERC) under grant 280152 and the Hungarian Scientific Research
    Fund (OTKA) under grant NK105645.}}

\author{{\sc Yixin Cao}\thanks{Institute for Computer Science and
    Control, Hungarian Academy of Sciences (MTA SZTAKI), Budapest,
    Hungary.  Email: \href{mailto:yixin@sztaki.hu}{\tt
      yixin@sztaki.hu}, \href{mailto:dmarx@cs.bme.hu}{\tt
      dmarx@cs.bme.hu}.}  \and {\sc D\'aniel Marx}$^{\dag}$}

\date{}
\maketitle

\begin{abstract}
  We study the minimum \emph{interval deletion} problem, which asks
  for the removal of a set of at most $k$ vertices to make a graph of
  $n$ vertices into an interval graph.  We present a parameterized
  algorithm of runtime $10^k \cdot n^{O(1)}$ for this problem,
  that is, we show the problem is fixed-parameter tractable.

\end{abstract}

\section{Introduction} 
A graph is an interval graph if its vertices can be assigned to
intervals of the real line such that there is an edge between two
vertices if and only if their corresponding intervals intersect.
Interval graphs are the natural models for DNA chains in biology and
many other applications, among which the most cited ones include jobs
scheduling in industrial engineering
\cite{bar-noy-01-resource-allocation} and seriation in archeology
\cite{kendall-69-seriation-1}.  Motivated by pure contemplation of
combinatorics and practical problems of biology respectively,
Haj{\'o}s \cite{hajos-57-interval-graphs} and Benzer
\cite{benzer-59-topology-genetic-structure} independently initiated
the study of interval graphs.

Interval graphs are a proper subset of chordal graphs.  After more
than half century of intensive investigation, the properties and the
recognition of interval and chordal graphs are well understood
\cite{booth-76-test-c1p}.  More generally, many NP-hard problems
(coloring, maximum independent set, etc.) are known to be
polynomial-time solvable when restricted to interval or chordal
graphs.  Therefore, one would like to generalize these results to
graphs that do not belong to these classes, but close to them in the
sense that they have only a few ``erroneous''/``missing'' edges or
vertices.  As a first step in understanding such generalizations, one
would like to know how far the given graph is from the class and to
find the erroneous/missing elements.  This leads us naturally to the
area of graph modification problems, where given a graph $G$, the task
is to apply a minimum number of operations on $G$ to make it a member
of some prescribed graph class ${\cal F}$.  Depending on the
operations we allow, we can consider, e.g., completion
(edge-addition), edge-deletion, and vertex-deletion versions of these
problems.  Let us point out that, when $\cal F$ is hereditary, the
vertex deletion version can be considered as the most robust variant,
which in some sense encompasses both edge addition and edge deletion:
if $G$ can be made a member of ${\cal F}$ by $k_1$ edge additions and
$k_2$ edge deletions, then it can be also made a member of ${\cal F}$
by deleting at most $k_1+k_2$ vertices (e.g., by deleting one endpoint
of each added/deleted edge).

Unfortunately, most of these graph modification problems are
computationally hard: for example, a classical result of Lewis and
Yannakakis \cite{lewis-80-node-deletion-np} shows that the vertex
deletion problem is NP-hard for {\em every} nontrivial and hereditary
class ${\cal F}$, and according to Lund and Yannakakis
\cite{lund-93-approximation-maximum-subgraph}, they are also MAX
SNP-hard.  Therefore, early work of Kaplan et
al.~\cite{kaplan-99-chordal-completion} and
Cai~\cite{cai-96-hereditary-graph-modification} focused on the
fixed-parameter tractability of graph modification problems.  Recall
that a problem, parameterized by $k$, is {\em fixed-parameter
  tractable (FPT)} if there is an algorithm with runtime $f(k)\cdot
n^{O(1)}$, where $f$ is a computable function depending only on $k$
\cite{downey-fellows-13}.  In the special case when the desired graph
class ${\cal F}$ can be characterized by a finite number of forbidden
(induced) subgraphs, then fixed-parameter tractability of such a
problem follows from a basic bounded search tree algorithm
\cite{cai-96-hereditary-graph-modification}.  However, many important
graph classes, such as forests, bipartite graphs, and chordal graphs
have minimal obstructions of arbitrarily large size (cycles, odd
cycles, and holes, respectively).  It is much more challenging to
obtain fixed-parameter tractability results for such classes, see
results on, e.g., bipartite graphs
\cite{reed-04-odd-cycle-transversals,
  kawarabayashi-10-linear-odd-cyles-transversal}, planar graphs
\cite{marx-12-planar-deletion,
  kawarabayashi-09-linear-planar-deletion}, acyclic graphs
\cite{cao-10-ufvs,chen-08-dfvs}, and minor-closed classes
\cite{adler-08-excluded-minors,fomin-12-f-deletion}.

For interval graphs, the fixed-parameter tractability of the
completion problem was raised as an open question by Kaplan et
al.~\cite{kaplan-99-chordal-completion} in 1994, to which a positive
answer with a $k^{2k}\cdot n^{O(1)}$-time algorithm was given by
Villanger et al.~\cite{villanger-09-interval-completion} in 2007.  In
this paper, we answer the complementary question on vertex deletion:
\begin{theorem}[{\bf Main result}]\label{thm:main-alg}
  There is a $10^k\cdot n^{O(1)}$-time algorithm for deciding whether
  or not there is a set of at most $k$ vertices whose deletion makes
  an $n$-vertex graph $G$ an interval graph.
\end{theorem}

\textbf{Related work.}  Let us put our result into context. Interval
graphs form a subclass of chordal graphs, which are graphs containing
no induced cycle of length greater than 3 (also called {\em holes}).
In other words, the minimal obstruction for being a chordal graph
might be holes of arbitrary length, hence infinitely many of them.
Even so, {\sc chordal completion} (to make a graph chordal by the
addition of at most $k$ edges) can still be solved by a bounded search
tree algorithm by observing that a large hole immediately implies a
negative answer to the problem
\cite{kaplan-99-chordal-completion,cai-96-hereditary-graph-modification}.
No such simple argument works for {\sc chordal deletion} (to make the
graph chordal by removing at most $k$ edges/vertices) and its
fixed-parameter tractability was procured by a completely different
and much more complicated approach \cite{marx-10-chordal-deletion}.

It is known that a graph is an interval graph if and only if it is
chordal and does not contain a structure called ``asteroidal triple''
(AT for short), i.e., three vertices such that each pair of them is
connected by a path avoiding neighbors of the third one
\cite{lekkerkerker-62-interval-graphs}.  Therefore, in the graph
modification problems related to interval graphs, one has to destroy
not only all holes, but all ATs as well.  The algorithm of Villanger
et al.~\cite{villanger-09-interval-completion} for the {\sc interval
  completion} problem first destroys all holes by the same bounded
search tree technique as in {\sc chordal completion}.  This step is
followed by a delicate analysis of the ATs and a complicated branching
step to break them in the resulting chordal graph.

A subclass of interval graphs that received attention is the class of
unit interval graphs: graphs that can be represented by intervals of
unit length. Interestingly, this class coincides with proper interval
graphs, which are those graphs that have a representation with no
interval containing another one. It is known that unit interval graphs
can be characterized as not having holes and three other specific
forbidden subgraphs, thus graph modification problems related to unit
interval graphs \cite{kaplan-99-chordal-completion,villanger-13-pivd}
are very different from those related to interval graphs, where the
minimal obstructions include an infinite family of ATs.

\textbf{Our techniques.}  Even though both {\sc chordal deletion} and
{\sc interval completion} seem related to {\sc interval deletion}, our
algorithm is completely different from the published algorithms for
these two problems.  The algorithm of
Marx~\cite{marx-10-chordal-deletion} for {\sc chordal deletion} is
based on iterative compression, identifying irrelevant vertices in
large cliques, and the use of Courcelle's Theorem on a bounded
treewidth graph; none of these techniques appears in the present
paper.

Villanger et al.~\cite{villanger-09-interval-completion} used a simple
bounded search tree algorithm to try every minimal way of completing
all the holes; therefore, one can assume that the input graph is
chordal.  ATs in a chordal graph are known to have the property of
being {\em shallow}, and in a minimal witness of an AT, every vertex
of the triple is {\em simplicial}.  This means that the algorithm of
\cite{villanger-09-interval-completion} can focus on completing such
ATs (see also \cite{cao-13-interval-completion}).  On the other hand,
there is no similar upper bound known on the number of minimal ways of
breaking all holes by removing vertices, and it is unlikely to exist.
Therefore, in a sense, {\sc interval deletion} is inherently harder
than {\sc interval completion}: in the former problem, we have to deal
with two types of forbidden structures, holes and shallow ATs, while
in the second problem, only shallow ATs concern us.  Indeed, we spend
significant effort in the present paper to make the graph chordal; the
main part of the proof is understanding how holes interact and what
the minimal ways of breaking them are.

 The main technical idea to handle holes is developing a reduction
 rule based on the modular decomposition of the graph and analyzing
 the structural properties of reduced graphs.  It turns out that the
 holes remaining in a reduced graph interact in a very special way
 (each hole is fully contained in the closed neighborhood of any other
 hole).  This property allows us to prove that the number of minimal
 ways of breaking the holes is polynomially bounded, and thus a simple
 branching step can reduce the problem to the case when the graph is
 chordal.  As another consequence of our reduction rule, we can prove
 that this chordal graph already has a structure close to interval
 graphs (it has a clique tree that is a caterpillar).  We can show
 that in such a chordal graph, ATs interact in a
 well-behaved way and we can find a set of 10 vertices such that there
 always exists a minimum solution that contains at least one of these
 10 vertices.  Therefore, we can complete our algorithm by branching
 on the deletion of one of these vertices.

 \textbf{Motivation.}  The motivation for the graph modification
 problem studied in this paper is twofold: theoretical and coming from
 applications.  Many classical graph-theoretic problems can be
 formulated as graph deletion to special graph classes.  For instance,
 \textsc{vertex cover}, \textsc{feedback vertex set}, \textsc{cluster
   vertex deletion}, and \textsc{odd cycle transversal} can be viewed
 as vertex deletion problems where the class ${\cal F}$ is the class
 of all empty graphs, forests, cluster graphs (i.e., disjoint union of
 cliques), and bipartite graphs, respectively.  Thus, the study of
 graph modification problems related to important graph classes can be
 seen as a natural extension of the study of classical combinatorial
 problems. In light of the importance of interval graphs, it is not
 surprising that there are natural combinatorial problems that can be
 formulated as, or computationally reduced to \textsc{interval
   deletion}, and then our algorithm for \textsc{interval deletion}
 can be applied.  For instance, Narayanaswamy and Subashini
 \cite{narayanaswamy-13-d-cos-r} recently used Theorem~\ref{thm:main-alg}
 as a subroutine to solve the maximum \textsc{consecutive ones
   sub-matrix} problem and the minimum \textsc{convex bipartite
   deletion} problem.

 As a historical coincidence, interval graph modification problems are
 motivated not only from the aforementioned theoretical studies, but
 because they have wide applications.  One central problem in molecular
 biology is to reconstruct the relative positions of clones along the
 target DNA based on their pairwise overlap information obtained via
 experimental methods.  These data are naturally formulated as a
 graph, where each clone is a vertex, and two clones are adjacent
 if{f} they overlap.  The graph should be an interval graph provided
 the relations are \emph{perfect}, and the problem is then equivalent
 to the construction of its interval model, which can be done in
 linear time.  However, real data are always inconsistent and
 contaminated by a few but crucial errors, which have to be detected
 and fixed.  In particular, on the detection of false-positive errors
 that correspond to false edges, Goldberg et
 al.~\cite{goldberg-95-interval-edge-deletion} proposed the
 \textsc{interval edge deletion} problem (to make the graph an
 interval graph by the deletion of at most $k$ edges) and showed its
 NP-hardness.  This problem is equivalent to the {maximum
   \emph{spanning} interval subgraph}, and is not known to be FPT or
 not.  Moreover, false-negative errors are also possible, which
 significantly complicates the situation.

 In this regard, we turn to the clones (vertices) involved in
 erroneous relations (edges) instead of the relations themselves, and
 try to identify them based on a similar assumption.  More
 specifically, we study the \textsc{interval (vertex) deletion}
 problem, which is equivalent to finding the {maximum \emph{induced}
   interval subgraph}.  Conceptually, this formulation is capable of
 dealing with both false-negatives and false-positives.
 Computationally, the number of clones involved in mis-observed
 relations is never larger, and believed to be significantly smaller,
 than the number of erroneous relations.  It might thus provide better
 assistance to biologists by revealing more meaningful information in
 less time, as proclaimed by Karp \cite{karp-93-mapping-genomes}:
\begin{quote} 
  {\em Thus, optimization methods should be viewed not as vehicles for
    solving a problem, but for proposing a plausible hypothesis to be
    confirmed or disconfirmed by further experiments. The search for
    the correct solution of a reconstruction problem must inevitably
    be an iterative process involving a close interaction between
    experimentation and computation.}
\end{quote}

In a seriation problem of archeology, overlap information of a
collection of artifacts is given, and we are asked to put them in
chronological order.  Again we cannot expect the data to be consistent
and have to deal with errors first.  In particular, the famous
\emph{Berge mystery story} \cite{golumbic-2004-perfect-graphs} is
essentially a seriation problem with false overlap information given
by a \emph{cheater}, and can be viewed as {\sc interval deletion} with
$k=1$.

\section{Outline}\label{sec:results}
The purpose of this section is to describe the main steps of our
algorithm at a high level.  We say that a set $Q\subset V(G)$ is an
\emph{interval deletion set} to a graph $G$ if $G - Q$ is an interval
graph.  An interval deletion set $Q$ is {\em minimum} if there is no
interval deletion set strictly smaller than $|Q|$, and it is {\em
  minimal} if no proper subset $Q'\subset Q$ is an interval deletion
set.  A set $X$ of vertices is called a {\em \mfset}\ if $X$ does not
induce an interval graph but every proper subset $X'\subset X$ does;
the subgraph $G[X]$ is called a {\em \mfis}.  Clearly, set $Q$ is an
interval deletion set if and only if it intersects every \mfset. Our
goal is to find an interval deletion set of size at most $k$. For
technical reasons, it will be convenient to define the problem as
follows: \medskip

\fbox{\parbox{0.85\linewidth}{
    {\sc interval deletion}: Given a graph $G$ and an integer
    parameter $k$, return
    \begin{itemize}
    \item if an interval deletion set of size $\le k$ exists, a {\em
        minimum} interval deletion set $Q\subset V(G)$;
    \item if no interval deletion set of size
        $\le k$ exists, ``NO.''
    \end{itemize}}
}

\paragraph{PHASE 1: Preprocessing.}  The first phase of the algorithm
applies two reduction rules exhaustively.  They either simplify the
instance or branch into a constant number of instances with strictly
smaller parameter value. The first reduction rule is straightforward:
we destroy every forbidden set of size at most $10$.

\begin{reduction}\label{rule:small-forbidden-subgraph}
  {\bf [Small forbidden sets]} Given an instance $(G,k)$ and a \mfset\
  $X$ of no more than $10$ vertices, we branch into $|X|$ instances,
  $(G-v, k-1)$ for each $v \in X$.
\end{reduction}

A graph on which Reduction~\ref{rule:small-forbidden-subgraph} cannot
be applied is called \emph{prereduced}.  

The second reduction rule is less obvious and more involved. Recall
that a subset $M$ of vertices forms a \emph{module} if every vertex in
$M$ has the same neighbors outside $M$
\cite{gallai-67-transitive-orientation}.  A module $M$ of $G$ is {\em
  nontrivial} if $1<|M|<|V(G)|$.  We observe (see
Section~\ref{sec:modul-decomp}) that a \mfset\ $X$ of at least 5
vertices is either fully contained in a module $M$ or contains at most
one vertex of $M$.  Moreover, if $X \cap M = \{x\}$, then replacing
$x$ by any other vertex $x'\in M\setminus\{x\}$ in $X$ results in
another \mfset.  This permits us to branch on modules, as described in
the following reduction rule.

\begin{reduction}\label{rule:non-interval-modules}{\bf [Main]}
  Let $I = (G,k)$ be an instance where the graph $G$ is prereduced,
  and a nontrivial \smodule\ $M$ that does not induce a clique.
  \begin{enumerate}
  \item If every \mfset\ is contained in $M$, then return the instance
    $(G[M],k)$.
  \item If no \mfset\ is contained in $M$, then return the instance
    $(G_M,k)$, where $G_M$ is obtained from $G$ by inserting edges to
    make $G[M]$ a clique.
  \item Otherwise, we solve three instances: $I_1 = (G-M,k-|M|)$, $I_2
    = (G[M],k-1)$, and $I_3=(G',k-1)$, where $G'$ is obtained from $G$
    by adding a clique $M'$ of $(k+1)$ vertices, connecting every pair
    of vertices $u\in M'$ and $v\in N(M)$, and deleting $M$; letting
    $Q_1$, $Q_2$, and $Q_3$ be the solutions of these instances
    respectively, we return either $Q_1 \cup M$ or $Q_2\cup Q_3$
    (``NO'' when $|Q_2\cup Q_3| > k$), whichever is smaller.
  \end{enumerate}
\end{reduction}
That is, in the third case we branch into two directions: the solution
is obtained either as the union of $M$ and the solution of $I_1$, or
as the union of solutions of $I_2$ and $I_3$. The two branches
correspond to the two cases where the solution fully contains $M$ or
only a minimum interval deletion set to $G[M]$ (i.e., $Q_2$),
respectively.  Note that in the second branch, it can be shown that
$Q_3$ is disjoint from $M'$; hence $Q_2\cup Q_3$ is indeed a subset of
$V(G)$.  Moreover, we have to clarify what the behavior of the
reduction is if one or more of $Q_1$, $Q_2$, and $Q_3$ are ``NO.''  If
$Q_2$ or $Q_3$ is ``NO,'' then we define $Q_2\cup Q_3$ to be ``NO'' as
well.  If one of $Q_1$ and $Q_2\cup Q_3$ is ``NO,'' we return the
other one; if both of them are ``NO,'' we return ``NO'' as well.

A graph on which neither reduction rule applies is called
\emph{reduced}; in such a graph, every nontrivial module induces a
clique.  In Section~\ref{sec:reduct-rules-branch}, we prove the
correctness of the reductions rules and that it can be checked in
polynomial time if a reduction rule is applicable.  Hence after
exhaustive application of the reductions, we may assume that the graph
is reduced.

The reductions are followed by a comprehensive study on reduced graphs
that yields two crucial combinatorial statements.  The first statement
is on an AT $\{x,y,z\}$ that are witnessed by a \mfis\ $W$ different
from a hole.  We say that $x$ is the \emph{shallow terminal} if $W -
N[x]$ is an induced path.  We prove the shallow terminal $x$ is
simplicial, i.e., $N(x)$ induces a clique.

\begin{theorem}{\bf [Shallow terminals]}
  \label{thm:reduced-instances-ap}
  {All shallow terminals in a reduced graph are simplicial.}
\end{theorem}

We say that two holes are \emph{congenial} to each other if each
vertex of one hole is a neighbor of the other hole.  It turns out that
the holes are pairwise congenial in a reduced graph.

\begin{theorem}{\bf [Congenial holes]}
  \label{thm:reduced-instances-b}
  {All holes in a reduced graph are congenial to each other.}
\end{theorem}

We point out that circular-arc graphs form an important example of
graphs where the holes are pairwise congenial.  Indeed, all holes of a
reduced graph induce a circular-arc graph, but such a proof will not
be given in this paper, as it is unnecessary for our purpose here.
One may refer to \cite{villanger-13-pivd} on more intuition.

\paragraph{PHASE 2: Breaking holes.}  A consequence of
Theorem~\ref{thm:reduced-instances-b} is that if a vertex $v$ is in a
hole, then $N[v]$ intersects every hole and thus makes a \emph{hole
  cover}.  Intuitively, this suggests that a minimal hole cover has to
be very local in a certain sense. Indeed, by relating minimal hole
covers in the reduced graph to minimal separators in the subgraph $G -
N[v]$, we are able to establish a quadratic bound on the number of
minimal hole covers, and more importantly, a cubic time algorithm to
construct them.

\begin{theorem}{\bf [Hole covers]}
  \label{thm:alg-reduced-with-holes}
  Every reduced graph of $n$ vertices contains at most $n^2$ minimal
  hole covers, and they can be enumerated in ${O}(n^3)$ time.
\end{theorem}

Any {interval deletion} set must be a hole cover, and thus contains a
minimal hole cover.  This allows us to branch into at most $n^2$
instances, in each of which the input graph is chordal.  Note that
this branching step is applied only once; hence only a polynomial
factor will be induced in the running time.

\paragraph{PHASE 3: Breaking ATs.}
As all the holes have been broken, the graph is already chordal at the
onset of the third phase.  It should be noted that, however, the graph
might not be reduced, as new nontrivial non-clique modules can be
introduced with the deletion of a hole cover in Phase~2.  In
principle, we could rerun the reductions of Phase~1 to obtain a
reduced instance, but there is no need to do so at this point.  The
properties that we need in this phase are that graph is prereduced,
chordal, and every shallow terminal is simplicial
(Theorem~\ref{thm:reduced-instances-ap}).  We give a name to such graphs
and compare it with previously defined notions here.

\begin{itemize}
\item A graph is \emph{prereduced} if
  Reduction~\ref{rule:small-forbidden-subgraph} does not apply.
\item A prereduced graph is \emph{reduced} if
  Reduction~\ref{rule:non-interval-modules} does not apply.
\item A prereduced graph is \emph{nice} if it is chordal and every
  shallow terminal in it is simplicial.
\end{itemize}

While both reduced graphs and nice graphs are prereduced, they are
incomparable to each other.  As only vertex deletions are applied
after Phase 1, in the remainder of this algorithm the graph is an
induced subgraph of that in a previous step.  In other words, once a
hereditary property is obtained after Phase 1, it remains true
thereafter.  It is easy to verify that the three defining properties
of nice graphs are all hereditary.  On the one hand, after the end of
Phase~1, a reduced graph is prereduced by definition, and according to
Theorem~\ref{thm:reduced-instances-ap}, every shallow terminal in it is
simplicial.  On the other hand, Phase~2 destroys all holes and the
chordal property is obtained.  Therefore, the graph becomes nice after
Phase~2 and will remain nice till the end of our algorithm.

The removal of all simplicial vertices from a nice graph breaks all
ATs (Theorem~\ref{thm:reduced-instances-ap}), thereby yielding an
interval graph.  This implies that a nice graph has a very special
structure: It has a clique tree decomposition where the tree is a
caterpillar, i.e., a path with degree-1 vertices attached to it. In
other words, all vertices other than the shallow terminals can be
arranged in a linear way, which greatly simplifies the examination of
interactions between ATs.  As a consequence, we can select an AT that
is minimal in a certain sense, and single out $10$ vertices such that
there must exist a minimum interval deletion set destroying this AT
with one of these $10$ vertices. We can therefore safely branch on
removing one of these 10 vertices.
\begin{theorem}{\bf [Nice graphs]}
  \label{thm:alg-reduced-and-chordal}
  There is a $10^k\cdot n^{O(1)}$ time algorithm for \textsc{interval
    deletion} on nice graphs.
\end{theorem}

Putting together these steps, the fixed-parameter tractability of
  \textsc{interval deletion} follows (see Figure~\ref{fig:alg}).

\begin{figure*}[t]
\setbox4=\vbox{\hsize28pc \noindent\strut
\begin{quote}
  \vspace*{-5mm} \footnotesize {\bf Algorithm
    Interval-Deletion($G,k$)}
  \\
  {\sc input}: a non-interval graph $G$ and a positive integer $k$.
  \\
  {\sc output}: a minimum interval deletion set $Q\subset V(G)$ of
  size $\leq k$ or ``NO.''

  1 \hspace*{3mm} Reduction~\ref{rule:small-forbidden-subgraph}: Let
  $U$ be a \mfset\ of at most 10 vertices;
  \\
  \hspace*{9mm} {\bf branch} on deleting one vertex of $U$;
  \\
  {\em $\setminus\!\!\setminus$ the graph will then be prereduced and
    remains so hereafter;}
  \\
  2 \hspace*{3mm} Reduction~\ref{rule:non-interval-modules}: Let $M$ be
  a nontrivial module of $G$ not inducing a clique;
  \\
  2.1 \hspace*{5mm} {\bf if} all \mfsets\ of $G$ are contained in $M$
  {\bf then}
  \\
  \hspace*{14mm} {\bf return} Interval-Deletion$(G[M],k)$;
  \\
  2.2 \hspace*{5mm} {\bf else if} no \mfset\ is contained in $M$ {\bf
    then}
  \\
  \hspace*{14mm} {\bf return} Interval-Deletion$(G_M,k)$, where edges
  are inserted to make $G[M]$ a clique;
  \\
  2.3 \hspace*{5mm} {\bf else branch} into three instances $I_1$,
  $I_2$, $I_3$;
  \\
  {\em $\setminus\!\!\setminus$ now the graph is reduced;}
  \\
  3 \hspace*{3mm} use the algorithm of
  Theorem~\ref{thm:alg-reduced-with-holes} to enumerate the at most
  $n^2$ minimal hole covers of $G$;
  \\
  {\em $\setminus\!\!\setminus$ the graph will then be nice and
    remains so hereafter;}
  \\
  4 \hspace*{3mm} {\bf for each} minimal hole cover $HC$ {\bf do}
  \\
  \hspace*{10mm} use the algorithm of
  Theorem~\ref{thm:alg-reduced-and-chordal} to solve {\bf ($G-HC,
    k-|HC|$)};
  \\
  5 \hspace*{3mm} {\bf return} the smallest solution obtained, or
  ``NO'' if all solutions are ``NO.''

\end{quote} \vspace*{-6mm} \strut} $$\boxit{\box4}$$
\vspace*{-9mm}
\caption{Outline of algorithm for \textsc{interval deletion}}
\label{fig:alg}
\end{figure*}
 
\paragraph{Theorem 1.1 (restated).}
  There is a $10^k\cdot n^{O(1)}$ time algorithm for deciding whether
  or not there is a set of at most $k$ vertices whose deletion makes
  an $n$-vertex graph $G$ an interval graph.
\begin{proof}
  The algorithm described in Figure~\ref{fig:alg} solves the problem
  by making recursive calls to itself, or calling the algorithm of
  Theorem~\ref{thm:alg-reduced-and-chordal} ${O}(n^2)$ times. In the
  former case, at most 10 recursive calls are made, all with parameter
  value at most $k-1$. In the latter case, the running time is
  $10^k\cdot n^{O(1)}$. It follows that the total running time of the
  algorithm is $10^k\cdot n^{O(1)}$.
\end{proof}

The paper is organized as follows.  Section~\ref{sec:pre} sets the
definitions and recalls some basic facts.
Section~\ref{sec:reduct-rules-branch} presents the details of the
first phase.  The next four sections are devoted to the proofs of
Theorems~\ref{thm:reduced-instances-ap}--\ref{thm:alg-reduced-and-chordal}.
Sections~\ref{sec:shallow} and \ref{sec:holes-in-modules} put shallow
terminals and congenial holes under thorough examination, and prove
Theorems~\ref{thm:reduced-instances-ap} and
\ref{thm:reduced-instances-b}, respectively.  Section~\ref{sec:holes}
fully characterizes minimal hole covers in reduced graphs and proves
Theorem~\ref{thm:alg-reduced-with-holes}.
Section~\ref{sec:caterp-decomp} presents the algorithm that destroys
ATs in nice graphs and proves
Theorem~\ref{thm:alg-reduced-and-chordal}.  Section~\ref{sec:remark}
closes this paper by some possible improvement and new directions.

\section{Preliminaries}\label{sec:pre}

All graphs discussed in this paper shall always be undirected and
simple.  A graph $G$ is given by its vertex set $V(G)$ and edge set
$E(G)$.  If a pair of vertices $v_1$ and $v_2$ is connected by an
edge, they are \emph{adjacent} to each other, and denoted by $v_1 \sim
v_2$, otherwise \emph{nonadjacent} and denoted by $v_1 \not\sim v_2$.
By $v \sim X$ we mean $v$ is adjacent to at least one vertex of the
set $X$.  Two vertex sets $X$ and $Y$ are completely connected if $x
\sim y$ for each pair of $x \in X$ and $y \in Y$.  A graph is {\it
  complete} if each pair of vertices are adjacent.  A \emph{clique} in
a graph is a subgraph that is complete, and a clique is maximal if it
is not contained in another clique.  A vertex is \emph{simplicial} if
its neighbors induce a clique.  A neighbor of a vertex is another
vertex that is adjacent to it, and the set of \emph{neighborhood} of a
vertex $v$ is denoted by $N(v)$.  The \emph{closed neighborhood} of
$v$ is defined as $N[v] = N(v) \cup \{v\}$.  This is generalized to a
vertex set $U$, whose closed neighborhood and neighborhood are defined
to be $N[U] = \bigcup_{v \in U} N[v]$ and $N(U) = N[U] \backslash U$.
The notation $N_U(v)$ ($N_U[v]$) stands for the neighbors of $v$ in
the set $U$, i.e., $N_U(v) = N(v) \cap U$ ($N_U[v] = N[v] \cap U$),
regardless of whether $v \in U$ or not.  The subgraph of a graph $G$
induced by a subset of vertices $U$ is denoted by $G[U]$, and $G - U$
is used as a shorthand for the subgraph induced by $V(G) \setminus U$.

A sequence of \emph{distinct} vertices $(v_0 v_1 \dots v_\ell)$ such
that $v_i \sim v_{i+1}$ for each $0 \leq i < \ell$ is called a
\emph{$v_0$-$v_\ell$ path}, whose \emph{length} is defined to be
$\ell$.  Vertices $v_0$ and $v_\ell$ are the \emph{ends} of the path,
while others, $\{v_1,\dots,v_{\ell-1}\}$, are called \emph{inner
  vertices}.  If the ends are distinct and adjacent, i.e., $\ell > 1$
and $v_0 \sim v_\ell$, then $(v_0 v_1 \dots v_\ell v_0)$ is called a
\emph{cycle}, whose \emph{length} is defined to be $\ell + 1$.  As an
abuse of notation, by $u \in P$ (resp. $u \in C$) we mean that the
vertex $u$ appears in the path $P$ (resp. cycle $C$), i.e., we use $P$
or $C$ as the set of vertices in the path (resp.  cycle).  A
\emph{chord} in a path or cycle is an edge between two non-consecutive
vertices in the path or cycle.  It is worth noting that the edge
$v_0v_\ell$, if exists, is a chord in the path $(v_0 v_1 \dots
v_{\ell})$, but not in the cycle $(v_0 v_1 \dots v_{\ell} v_0)$.  It
is easy to verify that no shortest path can contain a chord, so
between each pair of vertices of a connected graph there is a
chordless path.  A chordless cycle of length $\ell$, where $\ell\ge
4$, is called an \emph{($\ell$-)hole}.  A graph is \emph{chordal} if
it contains no hole, in other words, any cycle of length at least $4$
contains a chord.

Chordal graphs admit several important and related characterizations.
A set $S$ of vertices \emph{separates} $x$ and $y$, and is called an
\emph{$x$-$y$ separator} if there is no $x$-$y$ path in the subgraph
$G -S$, and \emph{minimal $x$-$y$ separator} if no proper subset of
$S$ separates $x$ and $y$.  For any pair of vertices $x$ and $y$, a
minimal $x$-$y$ separator is also called a \emph{minimal separator}.
A graph is chordal if and only if each minimal separator in it induces
a clique \cite{dirac-61-chordal-graphs}.  A nontrivial chordal graph
contains at least two simplicial vertices, and there is at least one
simplicial vertex in each component after the removal of any
separator.

A tree ${\cal T}$ whose nodes are the maximal cliques of a graph $G$
is a \emph{(maximal) clique tree} of $G$ if it satisfies the following
conditions: any pair of adjacent nodes $K_i$ and $K_j$ defines a
minimal separator that is $K_i \cap K_j$; for any vertex $x \in V$,
the maximal cliques containing $x$ correspond to a subtree of ${\cal
  T}$.  A graph is chordal if and only if it has such a clique tree.
A clique tree of a graph $G$ will be denoted by ${\cal T}(G)$, or
${\cal T}$ when the graph $G$ is clear from the context.  Without
distinguishing the node in a clique tree and the maximal clique in the
graph $G$ corresponding to it, we use $K$ to denote both.  A set of
vertices is a minimal separator of $G$ if and only if it is the
intersection of $K_i$ and $K_j$ for some edge $K_iK_j$ in ${\cal T}$
\cite{buneman-1974-rigid-circuit-graphs}.  This separator, $K_i\cap
K_j$, is a minimal $x$-$y$ separator for any pair of vertices $x \in
K_i \setminus K_j$ and $y \in K_j \setminus K_i$.

As interval graphs are chordal, all aforementioned properties also
apply to interval graphs.  Moreover, by the following characterization
of Fulkerson and Gross, each interval graph has a clique tree that is
a path.
\begin{theorem}[\cite{fulkerson-65-interval-graphs}]
  \label{thm:linear-tree}
  A graph $G$ is an interval graph if and only if the maximal cliques
  of $G$ can be linearly ordered such that, for each vertex $v$, the
  maximal cliques containing $v$ occur consecutively.
\end{theorem}

For a comprehensive treatment and for references to the extensive
literature on chordal graphs and interval graphs, one may refer to the
monograph of Golumbic~\cite{golumbic-2004-perfect-graphs} and the
survey of Brandst\"{a}dt et al.~\cite{brandstadt-99-graph-classes}.

\section{Reduction rules and branching}
\label{sec:reduct-rules-branch}
This section discusses the reduction rules described in
Section~\ref{sec:results} in more details.

\subsection{Forbidden induced subgraphs}
\label{sec:forbidden-subgraphs}

Three vertices form an \emph{asteroidal triple}, AT for short, if each
pair of them is connected by a path that avoids the neighborhood of
the third one.  We use \emph{asteroidal witness} (AW) to refer to a
minimal induced subgraph that is not a hole and contains an AT but
none of its proper induced subgraphs does.  It should be easy to check
that an AW contains precisely one AT, and its vertices are the union
of these three defining paths for this triple; the three defining
vertices will be called \emph{terminals} of this AW.  It can be
observed from Figure~\ref{fig:at} that the three terminals are the
only simplicial vertices of this AW and they are nonadjacent to each
other.  Lekkerkerker and Boland \cite{lekkerkerker-62-interval-graphs}
observed that a graph is an interval graph if and only if it is
chordal and contains no AW.  Not stopping here, they rolled up their
sleeves and got their hands dirty by checking each possible \fis.
Their effort brought the following less beautiful but more useful
characterization, here a minimal non-interval graph refers to a graph
whose every proper induced subgraph is an interval graph but itself is
not.

\begin{theorem}[\cite{lekkerkerker-62-interval-graphs}]
  A minimal non-interval graph is either a hole or an AW depicted in
  Figure~\ref{fig:at}.
\end{theorem}
\tikzstyle{corner}  = [fill=blue,inner sep=3pt]
\tikzstyle{special} = [fill=black,circle,inner sep=2pt]
\tikzstyle{vertex}  = [fill=black,circle,inner sep=2pt]
\tikzstyle{edge}    = [draw,thick,-]
\tikzstyle{at edge} = [draw,ultra thick,-,red]
\begin{figure*}[t]
  \centering
  \vspace*{-6mm}
  \begin{subfigure}[b]{0.2\textwidth}
    \centering
    \begin{tikzpicture}[scale=.2]
    \node [corner,label=above:$t_1$] (s) at (0,6.44) {};
    \node [corner,label=above:$t_2$] (a) at (-7, 0) {};
    \node [vertex] (a1) at (-4, 0) {};
    \node [special,label=45:$c$] (v) at (0, 0) {};
    \node [vertex] (b1) at (4, 0) {};
    \node [corner,label=above:$t_3$] (b) at (7, 0) {};
    \node [vertex] (c) at (0,3.5) {};
    \draw[] (a) -- (a1) -- (v) -- (b1) -- (b);
    \draw[] (v) -- (c) -- (s);
    \end{tikzpicture}
    \caption{long claw}
    \label{fig:long-claw}
  \end{subfigure}%
  \quad
  \begin{subfigure}[b]{0.2\textwidth}
    \centering
    \begin{tikzpicture}[scale=.2]
    \node [corner,label=above:$t_1$] (s) at (0,2.8) {};
    \node [corner,label=above:$t_2$] (a) at (-7, 0) {};
    \node [vertex] (a1) at (-4, 0) {};
    \node [special,label=45:$c$] (v) at (0, 0) {};
    \node [vertex] (b1) at (4, 0) {};
    \node [corner,label=above:$t_3$] (b) at (7, 0) {};
    \node [vertex] (c) at (0,-3.8) {};
    \draw[] (a) -- (a1) -- (v) -- (b1) -- (b) -- (c) -- (a);
    \draw[] (b1) -- (c) -- (a1);
    \draw[] (v) -- (c) -- (s);
    \end{tikzpicture}
    \caption{whipping top}
    \label{fig:top}
  \end{subfigure}%
  \quad
  \begin{subfigure}[b]{0.2\textwidth}
    \centering
    \begin{tikzpicture}[scale=.2]
    \node [corner,label=above:$t_1$] (s) at (0,6.44) {};
    \node [corner,label=above:$t_2$] (a) at (-7, 0) {};
    \node [vertex] (a1) at (-4, 0) {};
    \node [vertex] (b1) at (4, 0) {};
    \node [corner,label=above:$t_3$] (b) at (7, 0) {};
    \node [vertex] (c) at (0,3.5) {};
    \draw[] (a) -- (a1) -- (b1) -- (b);
    \draw[] (c) -- (s);
    \draw[] (a1) -- (c) -- (b1);
    \end{tikzpicture}
    \caption{net}
    \label{fig:net}
  \end{subfigure}%
  \quad
  \begin{subfigure}[b]{0.2\textwidth}
    \centering
    \begin{tikzpicture}[scale=.2]
    \node [corner,label=above:$t_1$] (s) at (0,6.44) {};
    \node [corner,label=above:$t_2$] (a) at (-7, 0) {};
    \node [vertex] (a1) at (0, 0) {};
    \node [corner,label=above:$t_3$] (b) at (7, 0) {};
    \node [vertex] (c1) at (-3,3) {};
    \node [vertex] (c2) at (3,3) {};
    \draw[] (a) -- (a1) -- (b) -- (c2) -- (s) -- (c1) -- (a);
    \draw[] (c1) -- (c2) -- (a1) -- (c1);
    \end{tikzpicture}
    \caption{tent}
    \label{fig:tent}
  \end{subfigure}%

  \begin{subfigure}[b]{0.4\textwidth}
    \centering
    \begin{tikzpicture}[scale=.45]
    \node [corner,label=right:$s$] (s) at (0,4.44) {};
    \node [corner,label=above:$l$,label=below:$b_{0}$] (a) at (-7, 0) {};
    \node [special,label=below:$b_{1}$] (a1) at (-5, 0) {};
    \node [vertex,label=below:$b_{2}$] (a2) at (-3, 0) {};
    \node [vertex,label=below:$b_{i}$] (bi) at (0, 0) {};
    \node [vertex,label=below:$b_{d-1}$] (b2) at (3, 0) {};
    \node [special,label=below:$b_d$] (b1) at (5, 0) {};
    \node [corner,label=above:$r$,label=below:$b_{d+1}$] (b) at (7, 0) {};
    \node [special,label=3:$c$] (c) at (0,2.7) {};
    \draw[] (a) -- (a1) -- (a2) (b2) -- (b1) -- (b);
    \draw[] (bi) -- (c) -- (s);
    \draw[] (a1) -- (c) -- (b1);
    \draw[] (a2) -- (c) -- (b2);
    \draw[dashed] (a2) -- (b2);
    \end{tikzpicture}
    \caption{$\dag$-AW $(s:c:l,B,r)\quad (d =|B| \geq 3)$}
    \label{fig:dag}
  \end{subfigure}%
  \qquad   \quad
  \begin{subfigure}[b]{0.4\textwidth}
    \centering
    \begin{tikzpicture}[scale=.45]
    \node [corner,label=right:$s$] (s) at (0,4.44) {};
    \node [corner,label=above:$l$,label=below:$b_{0}$] (a) at (-7, 0) {};
    \node [special,label=below:$b_{1}$] (a1) at (-5, 0) {};
    \node [vertex,label=below:$b_{2}$] (a2) at (-3, 0) {};
    \node [vertex,label=below:$b_{i}$] (bi) at (0, 0) {};
    \node [vertex,label=below:$b_{d-1}$] (b2) at (3, 0) {};
    \node [special,label=below:$b_d$] (b1) at (5, 0) {};
    \node [corner,label=above:$r$,label=below:$b_{d+1}$] (b) at (7, 0) {};
    \node [special,label=177:$c_1$] (c1) at (-1,2.7) {};
    \node [special,label=3:$c_2$] (c2) at (1,2.7) {};
    \draw[] (a) -- (a1) -- (a2) (b2) -- (b1) -- (b);
    \draw[] (bi) -- (c1) -- (s) -- (c2) -- (bi);
    \draw[] (a) -- (c1) -- (c2) -- (b);
    \draw[] (a1) -- (c1) -- (b1) -- (c2) -- (a1);
    \draw[] (a2) -- (c1) -- (b2) -- (c2) -- (a2);
    \draw[dashed] (a2) -- (b2);
    \end{tikzpicture}
    \caption{$\ddag$-AW $(s:c_1,c_2:l,B,r)\quad (d = |B| \geq 2)$}
    \label{fig:ddag}
  \end{subfigure}%
  \caption{Minimal asteroidal witnesses in a chordal graph (terminals
    are marked as squares).}
  \label{fig:at}
\end{figure*}
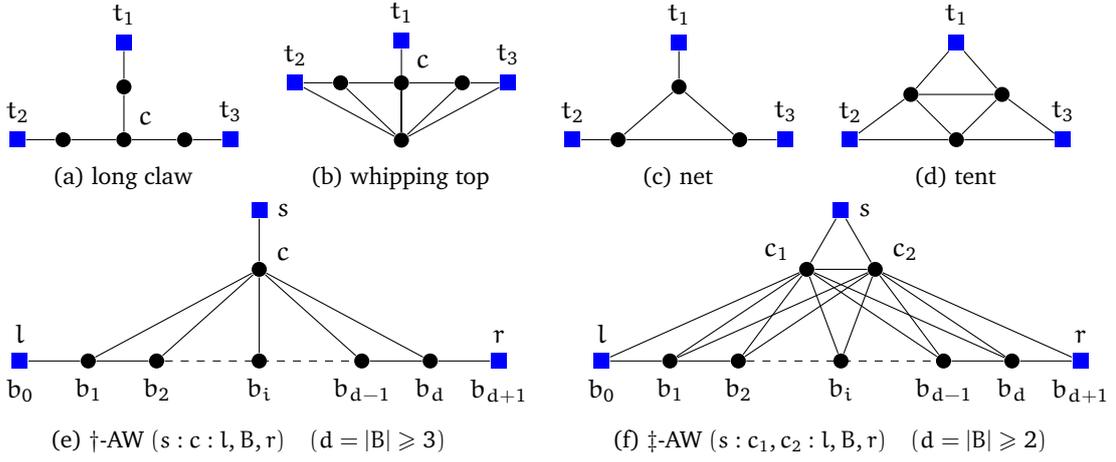

Some remarks are in order.  First, it is easy to verify that a hole of
six or more vertices witnesses an AT (specifically, any three
nonadjacent vertices from it) and is minimal, but following
convention, we only refer to it as a hole, while reserve the term AW
for graphs listed in Figure~\ref{fig:at}.  Second, the set of AWs
depicted in Figure~\ref{fig:at} is not a literal copy of the original
list in \cite{lekkerkerker-62-interval-graphs}, which contains neither
net nor tent; they are viewed as $\dag$-AW with $d=2$ and $\ddag$-AW
with $d=1$, respectively.  We single them out for the convenience of
later presentation.  To avoid ambiguities, in this paper we explicitly
require a $\dag$-AW (resp., $\ddag$-AW) to contain at least $7$
(resp., $8$) vertices.  Third, each of the four subgraphs in the first
row of Figure~\ref{fig:at} consists of a constant number, $6$ or $7$,
of vertices, and thus can be easily located and disposed of by
standard enumeration.  For the purpose of the current paper, we are
mainly concerned with the two kinds of AWs in the second row, whose
sizes are unbounded.  A $\dag$- or $\ddag$-AW $W$ contains a unique
terminal $s$, called the \emph{shallow terminal}, such that $W - N[s]$
is an induced path.  The neighbor(s) of the shallow terminal are the
\emph{center(s)}.  The other two terminals are called \emph{base
  terminals}, and other vertices are called \emph{base vertices}.  The
whole set of base vertices is called the \emph{base}.  We use
$(s:c:l,B,r)$ (resp., $(s:c_1,c_2:l,B,r)$) to denote the $\dag$-AW
(resp., $\ddag$-AW) with shallow terminal $s$, center $c$ (resp.,
centers $c_1$ and $c_2$), base terminals $l$ and $r$, and base $B =
\{b_1,\dots,b_d\}$.  For the sake of notational convenience, we will
also use $b_0$ and $b_{d+1}$ to refer to the base terminals $l$ and
$r$, respectively, even though they are not part of the base $B$.  The
center(s) and base vertices are called \emph{non-terminal vertices}.

Clearly, Reduction~\ref{rule:small-forbidden-subgraph} can be applied
in polynomial time: we can find a minimal forbidden set of size at
most 10 in polynomial time, e.g., by complete enumeration.  There are
ways to improve this, but optimizing the exponent is not the focus of
this paper.  After the exhaustive application of
Reduction~\ref{rule:small-forbidden-subgraph}, the graph is
\emph{prereduced}.  By definition, any AW in a prereduced graph
contains at least $11$ vertices, which rules out long claws, whipping
tops, nets, and tents.  Furthermore, the base of a $\dag$-AW (resp.,
$\ddag$-AW) in a prereduced graph contains at least $7$ (resp., $6$)
vertices.
  
The purpose of the following proposition and a detailed proof is
twofold.  These special structures arise frequently in this paper, and
we do not want to repeat the same argument again and again.  The proof
is exemplary in the sense that, by and large, most proofs of this
paper exploit a similar contradictory arguments: They explicitly
construct a \fis, either a small AW or a short hole, assuming the
property under discussion does not hold; because all graphs discussed
henceforth are prereduced, such a contradiction will suffice to prove
the desired property.

\begin{proposition}\label{lem:five-path}
  Let $P = (v_0\dots v_p)$ be a chordless path of length $p$ in a
  prereduced graph, and $u$ be adjacent to every inner vertex of $P$.
  \begin{itemize}
  \item[(1)] If $p \geq 4$ and $u$ is also adjacent to $v_0$ and
    $v_p$, then $N[v_{\ell}] \subseteq N[u]$ for every $2 \leq {\ell}
    \leq p-2$.
  \item[(2)] If $p \geq 3$ and $u$ is also adjacent to $v_0$ and
    $v_p$, then $N[v_{\ell}]\cap N[v_{\ell+1}] \subseteq N[u]$ for
    every $1 \leq {\ell} \leq p-2$.
  \item[(3)] If $p \geq 4$, then $N[v_{\ell}]\setminus ( N(v_{1})\cup
    N(v_{p - 1}) ) \subseteq N[u]$ for every $2 \leq {\ell} \leq p-2$.
  \end{itemize}
\end{proposition}
\begin{proof}
  Suppose to the contrary of statement (1), there is a vertex $x \in
  N[v_{\ell}] \setminus N[u]$, then we show the existence of a short
  hole or small AW in $G$, thus contradicting the assumption that $G$
  is prereduced.  Note that $x\not\sim v_i$ for any $i \leq \ell - 2$
  or $i \geq \ell + 2$, as otherwise, there is a $4$-hole $(u v_{i} x
  v_{{\ell}} u)$ (here $v_i \not\sim v_{\ell}$ because $P$ is
  chordless).  There is
  \begin{inparaitem}
    \item a $4$-hole $(u v_{{\ell} -1} x v_{\ell + 1} u)$ when $x$ is
      also adjacent to both $v_{\ell - 1}$ and $v_{\ell + 1}$;
    \item a tent $\{u,v_{\ell - 1},v_{\ell},x,v_{\ell + 1},v_{\ell +
        2}\}$ when $x$ is adjacent to $v_{\ell + 1}$ but not $v_{\ell
        - 1}$;
    \item a tent $\{u,v_{\ell - 2},v_{\ell - 1},x,v_{\ell},v_{\ell +
        1}\}$ when $x$ is adjacent to $v_{\ell - 1}$ but not $v_{\ell
        + 1}$; or
    \item a whipping top $\{x, u, v_{\ell - 2}, v_{\ell - 1},
      v_{\ell}, v_{\ell + 1}, v_{\ell + 2}\}$ otherwise ($x$ is only
      adjacent to $v_{\ell}$ in the path).
  \end{inparaitem}

  Suppose, for contradiction to statement (2), $x \in N[v_{\ell}]\cap
  N[v_{\ell+1}] \setminus N[u]$.  If $x$ is adjacent to $v_{\ell-1}$
  or $v_{\ell+2}$, then there is a $4$-hole; otherwise, there is a
  tent $\{u,v_{\ell - 1},v_{\ell},x,v_{\ell + 1},v_{\ell + 2}\}$.

  Statement (3) will follow from statement (1) if $u$ is also adjacent
  to $v_0$ and $v_p$; hence we assume otherwise, and without loss of
  generality, $u\not\sim v_0$.  Suppose to the contrary of statement
  (3), there is a vertex $x\in N[v_{\ell}]\setminus ( N(v_{1})\cup
  N(v_{p - 1}) \cup N[u])$.  If $v_2$ is the only inner vertex of $P$
  that is adjacent to $x$, then there is
  \begin{inparaitem}
    \item a $4$-hole $(x v_0 v_1 v_2 x)$ when $x\sim v_0$;
    \item a $4$-hole $(x v_4 v_3 v_2 x)$ when $x\sim v_4$;
    \item a net $\{v_0,v_1, x, v_2, u, v_4\}$ when $x\not\sim v_0,v_4$
      and $u\sim v_4$; or
    \item a $\dag$-AW $(x:v_2:v_0,v_1 u v_3, v_4)$ when $x\not\sim
      v_0,v_4$ and $u\not\sim v_4$.  
  \end{inparaitem} 
  A symmetric argument proves the case when $u\not\sim v_p$ and
  $v_{p-2}$ is the only inner vertex of $P$ that is adjacent to $x$.
  Other cases follow from statements (1) and (2).
\end{proof}

Let $X$ be a nonempty set of vertices.  A vertex $v$ is a \emph{common
  neighbor} of $X$ if it is adjacent to every vertex $x \in X$.  We
denote by $\widehat N (X)$ the set of all {common neighbors} of $X$.
It is easy to verify that in a prereduced graph, at least one of $X$
and $\widehat N (X)$ induces a clique, as otherwise two nonadjacent
vertices in $\widehat N (X)$, together with two nonadjacent vertices
in $X$, will induce a $4$-hole.  In particular, we have the following
proposition.

\begin{proposition}\label{lem:common-neighbors-is-clique}
  Let $X$ be a set of vertices of a prereduced graph that induces
  either a hole, an AW, or a path of length at least $2$.  Then
  $\widehat N (X)$ induces a clique.
\end{proposition}

\subsection{Modular decomposition}
\label{sec:modul-decomp}

A subset $M$ of vertices forms a \emph{module} of $G$ if all vertices
in $M$ have the same neighborhood outside $M$.  In other words, for
any pair of vertices $u,v \in M$ and vertex $x \not\in M$, $u \sim x$
if and only if $v \sim x$.  The set $V(G)$ and all singleton vertex
sets are modules, called \emph{trivial}.
A brief inspection shows that no graph in Figure~\ref{fig:at} has any
nontrivial modules and this is true also for holes of length greater
than 4:
\begin{proposition}\label{lem:module-exchangeable}
  Let $M$ be a module, and $X$ be a \mfset.  If $|X|>4$, then either
  $X \subseteq M$, or $|M \cap X| \leq 1$.
\end{proposition}

Indeed, the only \mfis\ of no more than $4$ vertices is a $4$-hole, of
which the pair of nonadjacent vertices might belong to a module.  This
observation allows us to prove the following statement, which is the
main combinatorial reason behind the correctness of the branching in
Reduction~\ref{rule:non-interval-modules}.
\begin{theorem}
  \label{thm:separable-modules}
  Let $G$ be a graph that contains no $4$-hole and $M$ be a module of
  $G$.  A minimum interval deletion set to $G$ contains either all
  vertices of $M$, or only a minimum interval deletion set to $G[M]$.
\end{theorem}
\begin{proof}
  Let $Q$ be a minimum interval deletion set to $G$ such that $M
  \not\subseteq Q$; otherwise we are already done.  To show that $Q_M
  = Q \cap M$ is precisely a minimum interval deletion set to $G[M]$,
  it suffices to show that for any minimum interval deletion set
  $Q'_M$ to $G[M]$, the set $Q' = (Q \setminus Q_M) \cup Q'_M$ is an
  interval deletion set to $G$: Trivially $Q_M$ is an interval
  deletion set to $G[M]$; if it is not minimum, then $|Q_M| > |Q'_M|$,
  and $|Q| > |Q'|$, which contradicts the fact that $Q$ is minimum.

  Suppose the contrary and $X$ is a \mfset\ in $G- Q'$.  By
  construction, $Q'_M$ intersects every \mfset\ in $G[M]$, while $Q
  \setminus Q_M$ intersects every \mfset\ in $G - M$. Thus $X$
  intersects both $M$ and $V(G)\setminus M$.  On the other hand,
  $|X|>4$ as the graph is $4$-hole free.  According to
  Proposition~\ref{lem:module-exchangeable}, $X\cap M$ contains
  exactly one vertex; let it be $x$.  Let $x'$ be a vertex in $M
  \setminus Q$, which is nonempty by the assumption $M \not\subseteq
  Q$, and let $X' = X \setminus \{x\} \cup \{x'\}$; it is immaterial
  whether $x'=x$ or not.  The set $X'$ is disjoint from $Q$, and by
  definition of modules, $G[X']$ and $G[X]$ are isomorphic.  In other
  words, $X'$ is a \mfset\ in $G - Q$, which is impossible.
  Therefore, $Q'$ is a interval deletion set to $G$ and this finishes
  this proof.
\end{proof}

To apply Reduction~\ref{rule:non-interval-modules}, we have to first
find a nontrivial module that is not a clique. For this purpose, we do
not need to compute a modular decomposition tree of the graph.  The
simple algorithm described in Figure~\ref{fig:alg-module} is
sufficient.
\begin{lemma}
  \label{lem:find-non-clique-module}
  We can find in polynomial time a nontrivial module $M$ such that
  $G[M]$ is not a clique, or report no such a module exists.
\end{lemma}
\begin{proof}
  The algorithm described in Figure~\ref{fig:alg-module} finds such a
  module in a greedy manner.  It starts from a pair of nonadjacent
  vertices $u$ and $v$, and generates the module by adding vertices.
  Note that each vertex in the set $X$ defined at step~2.1 is a
  witness for the fact that $M$ is not a module, in other words, $M$
  is a module only if $X=\emptyset$.  When a nonempty vertex set $M$
  is returned at step~2.2, from the algorithm we can derive that $X =
  \emptyset$ and $M \neq V(G)$; hence $M$ must be a nontrivial module.
  Now it remains to show that as long as there is a nontrivial
  non-clique module $U$ in the graph, the algorithm is guaranteed to
  return a nonempty set (not necessarily $U$ itself).  As $U$ does not
  induce a clique, it contains a pair of nonadjacent vertices $u$ and
  $v$, which shall be considered in some iteration of the for-loop.
  In this iteration, initially $M\subseteq U$, and by induction we are
  able to show that no vertex of $V(G)\setminus U$ can be included in
  $X$ during this iteration; hence $M\subseteq U$ will remain an
  invariant.  As a consequence, a subset $M$ that satisfies
  $\{u,v\}\subseteq M\subseteq U$ is returned.
\end{proof}

\begin{figure}[t]
\setbox4=\vbox{\hsize28pc \noindent\strut
\begin{quote}
  \vspace*{-5mm} \footnotesize

  \hspace*{-.5mm} {\bf for each} pair of nonadjacent vertices
  $u$ and $v$ {\bf do}\\
  1 \hspace*{5mm} $M = \{u,v\}$;\\
  2 \hspace*{5mm} {\bf while} $M\neq V(G)$ {\bf do}\\
  2.1 \hspace*{7mm} $X = \{x\not\in M : 0 < |N_M(x)| < |M|\}$;\\
  2.2 \hspace*{7mm} {\bf if} $X = \emptyset$ {\bf then return} $M$;\\
  2.3 \hspace*{7mm} {\bf else} $M = M \cup X$;\\
  \hspace*{-0.5mm} {\bf return} $\emptyset$.
  $\qquad\setminus\!\!\setminus$  there is no such a module

\end{quote} \vspace*{-6mm} \strut} $$\boxit{\box4}$$
\vspace*{-9mm}
\caption{Algorithm Find-Module}
\label{fig:alg-module}
\end{figure}

Indeed, one can easily verify that the module found as above is the
inclusive-wise minimal one containing both $u$ and $v$.  We are now
ready to explain the application of
Reduction~\ref{rule:non-interval-modules} and prove its correctness.

\begin{lemma}
  \label{lem:rule-module-ok}
  Reduction~\ref{rule:non-interval-modules} is correct and it can be
  checked in polynomial time whether
  Reduction~\ref{rule:non-interval-modules} (and which case of it) is
  applicable.
\end{lemma}
\begin{proof}
  The correctness of the reduction is clear in case~1: removing the
  vertices of $V(G)\setminus M$ does not make the problem any easier,
  as these vertices do not participate in any \mfset.
  
  In case~2, the correctness of the reduction follows from the fact
  that $G$ and $G_M$ have the same set of \mfsets.  Note that a clique
  is an interval graph, and more importantly, the insertion of edges
  to make $M$ a clique neither breaks the modularity of $M$ nor
  introduces any new 4-hole; thus
  Proposition~\ref{lem:module-exchangeable} is applicable to $G_M$.
  As $M$ induces an interval graph in both $G$ and $G_M$, if $X$ is a
  \mfset\ of $G$ or $G_M$, then
  Proposition~\ref{lem:module-exchangeable} implies that $X$ contains
  at most one vertex of $M$.  In other words, the insertion of edges
  has no effect on any \mfset, which means that $Q$ is an interval
  deletion set to $G$ if and only if it is an interval deletion set to
  $G_M$.
  
  The correctness of case~3 can be argued using
  Theorem~\ref{thm:separable-modules}, which states the two
  possibilities of any interval deletion set to $G$ with respect to
  $M$.  In particular, the two branches of case~3 correspond to these
  two cases.  The first branch is straightforward: we simply remove
  all vertices of $M$ from the graph and solve the instance $I_1=(G -
  M,k-|M|)$.  It is the second branch (where we assume $M
  \not\subseteq Q$) that needs more explanation.  Recall that by
  construction of $I_3$, the set $M'$ is a module of $G'$ and induces
  an interval graph.  It is clear that either solution $Q_2$ or $Q_3$
  being ``NO'' will rule out the existence of an interval deletion set
  of $G$ that does not fully contain $M$.  Hence we may assume $Q_2$
  and $Q_3$ are minimum interval deletion sets of $I_2$ and $I_3$,
  respectively; and $Q=Q_2\cup Q_3$.  Note that both $|Q_2|$ and
  $|Q_3|$ are upper bounded by $k - 1$.

  \begin{claim}\label{claim:solution}
    Set $Q$ is an interval deletion set of $G$.
  \end{claim}
  \begin{proof}
    According to Theorem~\ref{thm:separable-modules}, if $Q_3$ intersects
    $M'$, which is a module of $G'$, then it must contain all $(k+1)$
    vertices in $M'$,\footnote{Indeed, $\min(k+1,|N(M)|)$ vertices
      will suffice for our bookkeeping purpose, and an alternative way
      to this is to add only one vertex but mark it as ``forbidden.''}
    i.e., $|Q_3| > k$; a contradiction.  Therefore, $Q_3\cap M' =
    \emptyset$, which means $Q \subset V(G)$.  Suppose that there is a
    \mfset\ $X$ of $G$ disjoint from $Q$.  It cannot be fully
    contained in $M$, as $Q_2\subseteq Q$ is an interval deletion set
    of $G[M]$.  Then by Proposition~\ref{lem:module-exchangeable}, $X$
    contains exactly one vertex $x$ of $M$ and $X'=X\setminus \{x\}
    \cup \{x'\}$ is also a minimal forbidden set of $G'$ for any
    $x'\in M'$.  Since $Q_3$ is an interval deletion set of $G'$ disjoint
    from $M'$, it has to contain a vertex of $X'\setminus \{x'\} = X
    \setminus \{x\}$; a contradiction.
    \renewcommand{\qedsymbol}{$\lrcorner$}
  \end{proof}

  \begin{claim}\label{claim:minimum}
    Set $Q$ is not larger than the smallest interval deletion set $Q'$
    satisfying $M\not\subseteq Q'$.
  \end{claim}
  \begin{proof}
    Suppose that $Q'$ is an interval deletion set of $G$ of size at
    most $k$ with $M\not\subseteq Q'$; let $Q'_2=Q'\cap M$ and
    $Q'_3=Q'\setminus M$.  We claim that $Q'_2$ and $Q'_3$ are
    interval deletion sets of $I_2$ and $I_3$, respectively.  First,
    we argue that $Q'_2$ and $Q'_3$ are not empty; hence both of them
    have sizes at most $k-1$.  The assumption that $G[M]$ is not an
    interval graph implies $Q'_2\neq \emptyset$.  By assumption,
    $M\not\subseteq Q'$, thus there is a vertex $x\in M\setminus Q'$.
    Now $Q'_3=\emptyset$ would imply that $G - (M\setminus \{x\})$ is
    an interval graph, that is, there is no \mfset\ containing only
    one vertex of $M$, and it follows that we should have been in Case
    1. Since $|Q'_2|\le k-1$, it is clear that $Q'_2$ is a solution of
    instance $I_2=(G[M],k-1)$. The only way $Q'_3$ is not a solution
    of $I_3$ is that there is a \mfset\ $X$ containing a vertex of the
    $(k+1)$-clique introduced to replace $M$. As this $(k+1)$-clique
    is a module, Proposition~\ref{lem:module-exchangeable} implies
    that $X$ contains exactly one vertex $y$ of this clique. But in
    this case $X'=X\setminus \{y\} \cup \{x\}$ (where $x$ is a vertex
    of $M\setminus Q'$) is a \mfset\ disjoint from $Q'$, a
    contradiction.  Thus $|Q|\le |Q'|$ follows from the fact that both
    $Q_2$ and $Q_3$ are minimum.
    \renewcommand{\qedsymbol}{$\lrcorner$}
  \end{proof}

  As a consequence of Claim~\ref{claim:minimum}, if $|Q| > k$, then
  there cannot be an interval deletion set of size no more than $k$
  that does not fully include $M$.  This finishes the proof of the
  correctness of Reduction~\ref{rule:non-interval-modules}.

  On the applicability of Reduction~\ref{rule:non-interval-modules},
  we first use Lemma~\ref{lem:find-non-clique-module} to find a
  nontrivial module that does not induce a clique.  If such a module
  $M$ is found, then Reduction~\ref{rule:non-interval-modules} is
  applicable, and it remains to figure out which case should apply by
  checking the conditions in order.  To check whether case~1 holds, we
  need to check if there is a \mfset\ $X$ not contained in $M$. By
  Proposition~\ref{lem:module-exchangeable}, such an $X$, if exists,
  contains at most one vertex $x$ from $M$; and $x$ can be replaced by
  any other vertex of $M$.  Therefore, it suffices to pick any vertex
  $x\in M$, and test in linear time whether $G - (M\setminus \{x\})$ is an
  interval graph.  If it is not an interval graph, then there is a
  \mfset\ $X$ not contained in $M$ (as it contains at most one vertex
  of $M$).  Otherwise, $G- (M\setminus \{x\})$ is an interval graph for
  every $x\in M$, and there is no such $X$; hence case~1 holds.  To
  check whether case~2 holds, observe that the condition ``there is no
  \mfset\ contained in $M$'' is equivalent to saying that $G[M]$ is an
  interval graph, which can be checked in linear time. In all
  remaining cases, we are in case~3.
\end{proof}

\section{Shallow terminals}
\label{sec:shallow}

This section proves Theorem~\ref{thm:reduced-instances-ap} by showing
that each shallow terminal is contained in a \smodule\ whose
neighborhood induces a clique.  This module either is trivial
(consisting of only this shallow terminal), or induces a clique (after
the application of Reduction~\ref{rule:non-interval-modules}).
Therefore, this shallow terminal is always simplicial.  Recall that an
AW in a prereduced graph $G$ has to be a $\dag$- or $\ddag$-AW.  Let
us start from a thorough scrutiny of neighbors of its shallow
terminal, which, by definition, is disjoint from the base and base
terminals.

\begin{lemma}
  \label{lem:common-neighbor-of-base}
  Let $W$ be an AW in a prereduced graph.  Every common neighbor $x$
  of the base $B$ is adjacent to the shallow terminal $s$.
\end{lemma}
\begin{proof}
  The center(s) of $W$ are also common neighbors of $B$, and hence
  according to Proposition~\ref{lem:common-neighbors-is-clique}, they are
  adjacent to $x$.  Suppose, for contradiction, $x\in \widehat
  N(B)\setminus N(s)$.  If $W$ is a $\dag$-AW, then there is (see the
  first row of Figure~\ref{fig:common-neighbor})
    \begin{inparaitem}
    \item a whipping top $\{s,c,l,b_1,x,b_d,r\}$ centered at $c$ when $x
      \sim l,r$;
    \item a net $\{s,c,l,b_1,r,x\}$ when $x \sim r$ but $x \not\sim l$
      (similarly for $x \sim l$ but $x \not\sim r$); or
    \item a $\dag$-AW $(s:c:l, b_1 x b_d, r)$ when $x \not\sim l,r$.
    \end{inparaitem}
    If $W$ is a $\ddag$-AW, then there is (see the second row of
    Figure~\ref{fig:common-neighbor})
    \begin{inparaitem}
    \item  a tent $\{x,c_1,b_1,s,b_d,c_2\}$ when $x
      \sim l,r$;
    \item a $\ddag$-AW $(s:c_1,c_2:l,b_1 x,r)$ when $x \sim r$ but
      $x \not\sim l$ (similarly for $x \sim l$ but $x \not\sim r$); or
    \item a $\ddag$-AW $(s:c_1,c_2:l,b_1 x b_d,r)$ when $x \not\sim
      l,r$.
    \end{inparaitem}
    As none of these structures can exist in a prereduced graph,
    this lemma is proved.
\end{proof}
\begin{figure*}[t]
  \centering
  \begin{subfigure}[b]{0.3\textwidth}
    \centering
    \begin{tikzpicture}[scale=.35]
      \node [corner,label=right:$s$] (s) at (0,6.44) {};
      \node [corner,label=below:$l$] (a) at (-7, 0) {};
      \node [vertex,label=below:$b_1$] (a1) at (-5, 0) {};
      \node [vertex] (a2) at (-3, 0) {};
      \node [vertex] (b2) at (3, 0) {};
      \node [vertex,label=below:$b_d$] (b1) at (5, 0) {};
      \node [corner,label=below:$r$] (b) at (7, 0) {};
      \node [special,label=right:$c$] (c) at (0,4) {};
      \node [vertex,label=below:$x$] (u) at (0,-1.5) {};
      \draw[] (a1) -- (a2) -- (c) -- (b2) -- (b1);
      \draw[dashed] (a2) -- (b2);
      
      \draw[] (a2) -- (u) -- (b2);
      \draw[at edge] (a) -- (u) -- (b1) -- (b) -- (u) -- (a1) -- (a);
      \draw[at edge] (a1) -- (c) -- (b1);
      \draw[at edge] (s) -- (c) -- (u);
    \end{tikzpicture}
    \caption{$\dag$-AW, $x\sim l,r$}
  \end{subfigure}
  \qquad
  \begin{subfigure}[b]{0.3\textwidth}
    \centering
    \begin{tikzpicture}[scale=.35]
      \node [corner,label=right:$s$] (s) at (0,6.44) {};
      \node [corner,label=below:$l$] (a) at (-7, 0) {};
      \node [vertex,label=below:$b_1$] (a1) at (-5, 0) {};
      \node [vertex] (a2) at (-3, 0) {};
      \node [vertex] (b2) at (3, 0) {};
      \node [vertex,label=below:$b_d$] (b1) at (5, 0) {};
      \node [corner,label=below:$r$] (b) at (7, 0) {};
      \node [vertex,label=right:$c$] (c) at (0,4) {};
      \node [vertex,label=below:$x$] (u) at (0,-1.5) {};
      \draw[] (a1) -- (a2) -- (c) -- (b2) -- (b1) -- (b);
      \draw[dashed] (a2) -- (b2);
      \draw[] (u) -- (b1) -- (c);
      \draw[] (a2) -- (u) -- (b2);

      \draw[at edge] (b) -- (u) -- (a1) -- (a);
      \draw[at edge] (a1) -- (c);
      \draw[at edge] (s) -- (c) -- (u);
    \end{tikzpicture}
    \caption{$\dag$-AW, $x\not\sim l$, and $x\sim r$}
  \end{subfigure}
  \qquad
  \begin{subfigure}[b]{0.3\textwidth}
    \centering
    \begin{tikzpicture}[scale=.35]
      \node [corner,label=right:$s$] (s) at (0,6.44) {};
      \node [corner,label=below:$l$] (a) at (-7, 0) {};
      \node [special,label=below:$b_1$] (a1) at (-5, 0) {};
      \node [vertex] (a2) at (-3, 0) {};
      \node [vertex] (b2) at (3, 0) {};
      \node [special,label=below:$b_d$] (b1) at (5, 0) {};
      \node [corner,label=below:$r$] (b) at (7, 0) {};
      \node [special,label=right:$c$] (c) at (0,4) {};
      \node [vertex,label=below:$x$] (u) at (0,-1.5) {};
      \draw[] (a1) -- (a2) -- (c) -- (b2) -- (b1);
      \draw[dashed] (a2) -- (b2);
      
      \draw[] (a2) -- (u) -- (b2);
      \draw[at edge] (u) -- (b1) -- (b);
      \draw[at edge] (a1) -- (c) -- (b1);
      \draw[at edge] (s) -- (c) -- (u) -- (a1) -- (a);
    \end{tikzpicture}
    \caption{$\dag$-AW, $x\not\sim l,r$}
  \end{subfigure}

  \begin{subfigure}[b]{0.3\textwidth}
    \centering
    \begin{tikzpicture}[scale=.35]
      \node [corner,label=right:$s$] (s) at (0,6.44) {};
      \node [corner,label=below:$l$] (a) at (-7, 0) {};
      \node [vertex,label=below:$b_1$] (a1) at (-5, 0) {};
      \node [vertex] (a2) at (-3, 0) {};
      \node [vertex] (b2) at (3, 0) {};
      \node [vertex,label=below:$b_d$] (b1) at (5, 0) {};
      \node [corner,label=below:$r$] (b) at (7, 0) {};
      \node [special,label=left:$c_1$] (c1) at (-1.5,4) {};
      \node [special,label=right:$c_2$] (c2) at (1.5,4) {};
      \node [vertex,label=below:$x$] (u) at (0, -1.5) {};
      \draw[] (a) -- (a1) -- (a2) -- (c1) -- (b2) -- (b1) -- (b);
      \draw[] (a1) -- (c1) -- (b1);
      \draw[] (a1) -- (c2) -- (b1);
      \draw[] (a2) -- (c2) -- (b2);
      \draw[] (a1) -- (u) -- (b1);
      \draw[] (a2) -- (u) -- (b2);
      \draw[dashed] (a2) -- (b2);
      \draw[at edge] (a) -- (c1) -- (s) -- (c2) -- (b) -- (u) -- (a);
      \draw[at edge] (u) -- (c1) -- (c2) -- (u);
    \end{tikzpicture}
    \caption{$\ddag$-AW, $x\sim l,r$}
  \end{subfigure}%
  \qquad
  \begin{subfigure}[b]{0.3\textwidth}
    \centering
    \begin{tikzpicture}[scale=.35]
      \node [corner,label=right:$s$] (s) at (0,6.44) {};
      \node [corner,label=below:$l$] (a) at (-7, 0) {};
      \node [vertex,label=below:$b_1$] (a1) at (-5, 0) {};
      \node [vertex] (a2) at (-3, 0) {};
      \node [vertex] (b2) at (3, 0) {};
      \node [vertex,label=below:$b_d$] (b1) at (5, 0) {};
      \node [corner,label=below:$r$] (b) at (7, 0) {};
      \node [special,label=left:$c_1$] (c1) at (-1.5,4) {};
      \node [special,label=right:$c_2$] (c2) at (1.5,4) {};
      \node [vertex,label=below:$x$] (u) at (0, -1.5) {};
      \draw[] (a1) -- (a2) -- (c1) -- (b2) -- (b1) -- (b);
      \draw[] (c1) -- (b1) -- (c2);
      \draw[] (a2) -- (c2) -- (b2);
      \draw[] (a1) -- (u) -- (b1);
      \draw[] (a2) -- (u) -- (b2);
      \draw[dashed] (a2) -- (b2);

      \draw[at edge] (a) -- (a1)  -- (u);
      \draw[at edge] (c1) -- (a1) -- (c2);
      \draw[at edge] (a) -- (c1) -- (s) -- (c2) -- (b) -- (u);
      \draw[at edge] (u) -- (c1) -- (c2) -- (u);
    \end{tikzpicture}
    \caption{$\ddag$-AW, $x\not\sim l$, and $x\sim r$}
  \end{subfigure}%
  \qquad
  \begin{subfigure}[b]{0.3\textwidth}
    \centering
    \begin{tikzpicture}[scale=.35]
      \node [corner,label=right:$s$] (s) at (0,6.44) {};
      \node [corner,label=below:$l$] (a) at (-7, 0) {};
      \node [special,label=below:$b_1$] (a1) at (-5, 0) {};
      \node [vertex] (a2) at (-3, 0) {};
      \node [vertex] (b2) at (3, 0) {};
      \node [special,label=below:$b_d$] (b1) at (5, 0) {};
      \node [corner,label=below:$r$] (b) at (7, 0) {};
      \node [special,label=left:$c_1$] (c1) at (-1.5,4) {};
      \node [special,label=right:$c_2$] (c2) at (1.5,4) {};
      \node [vertex,label=below:$x$] (u) at (0, -1.5) {};
      \draw[] (a1) -- (a2) -- (c1) -- (b2) -- (b1);
      \draw[] (a2) -- (c2) -- (b2);
      \draw[] (a1) -- (u) -- (b1);
      \draw[] (a2) -- (u) -- (b2);
      \draw[dashed] (a2) -- (b2);

      \draw[at edge] (a1) -- (c1) -- (b1);
      \draw[at edge] (a1) -- (c2) -- (b1);
      \draw[at edge] (a) -- (a1) -- (u) -- (b1) -- (b);
      \draw[at edge] (a) -- (c1) -- (s) -- (c2) -- (b);
      \draw[at edge] (u) -- (c1) -- (c2) -- (u);
    \end{tikzpicture}
    \caption{$\ddag$-AW, $x\not\sim l,r$}
  \end{subfigure}%

  \caption{Adjacency between a common neighbor $x$ of $B$ and $s$
    [Lemma~\ref{lem:common-neighbor-of-base}].}
  \label{fig:common-neighbor}
\end{figure*}
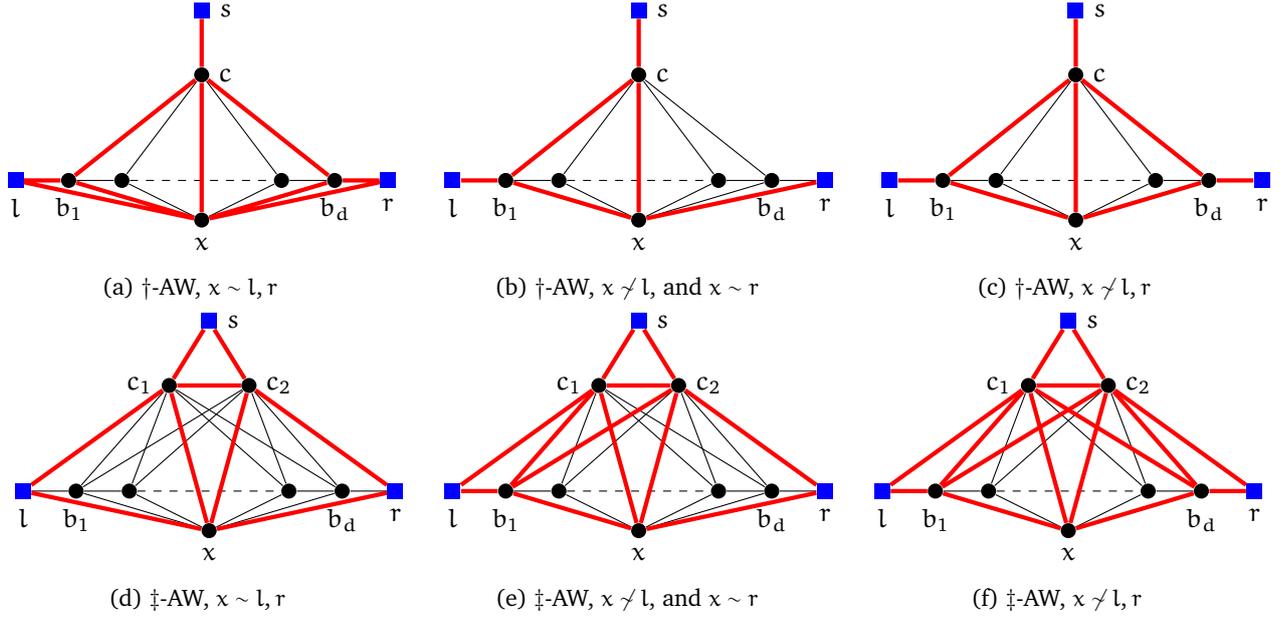

\begin{lemma}\label{lem:shallow}
  Let $W$ be an AW 
  in a prereduced graph, and $x$ be adjacent to the shallow terminal
  $s$.
  \begin{enumerate}
  \item[(1)] Then $x$ is also adjacent to the center(s) of $W$
    (different from $x$).
  \item[(2)] Classifying $x$ with respect to its adjacency to the base
    $B$, we have the following categories:
    \begin{description}
    \item [(full)] $x$ is adjacent to every base vertex.\\
      Then $x$ is also adjacent to every vertex in $N(s)\setminus
      \{x\}$.
    \item [(partial)] $x$ is adjacent to some, but not all base vertices.\\
      Then there is an AW whose shallow terminal is $s$, one center is
      $x$, and base is a proper sub-path of $B$.
    \item [(none)] $x$ is adjacent to no base vertex.\\
      Then $x$ is adjacent to neither base terminals, and thus
      replacing the shallow terminal of $W$ by $x$ makes another AW.
    \end{description}
  \end{enumerate}
\end{lemma}
\begin{figure*}[t]
  \centering
  \begin{subfigure}[b]{0.3\textwidth}
    \centering
    \begin{tikzpicture}[scale=.35]
      \node [corner,label=$x$] (x) at (-3,6.44) {};
      \node [corner,label=$s$] (s) at (0,6.44) {};
      \node [vertex] (a) at (-7, 0) {};
      \node [vertex] (a1) at (-5, 0) {};
      \node [corner] (a2) at (-3, 0) {};
      \node [vertex] (b2) at (3, 0) {};
      \node [vertex] (b1) at (5, 0) {};
      \node [vertex] (b) at (7, 0) {};
      \node [vertex,label=left:$c_1$] (c1) at (-1.5,4) {};
      \node [corner,label=right:$c_2$] (c2) at (1.5,4) {};
      \draw[] (a2) -- (a1) -- (a) -- (c1) -- (b2) -- (b1) -- (b);
      \draw[] (a1) -- (c1) -- (b1);
      \draw[] (a1) -- (c2) -- (b1);
      \draw[] (a2) -- (c1) -- (s);
      \draw[] (c1) -- (c2) -- (b);
      \draw[] (c2) -- (b2);
      \draw[dashed] (a2) -- (b2);
      
      \draw[at edge] (c2) -- (s) -- (x) -- (a2) -- (c2);
    \end{tikzpicture}
    \caption{$x \sim B$}
  \end{subfigure}
  \qquad
  \begin{subfigure}[b]{0.3\textwidth}
    \centering
    \begin{tikzpicture}[scale=.35]
      \node [corner,label=$x$] (x) at (-3,6.44) {};
      \node [corner,label=$s$] (s) at (0,6.44) {};
      \node [corner] (a) at (-7, 0) {};
      \node [corner] (a1) at (-5, 0) {};
      \node [vertex] (a2) at (-3, 0) {};
      \node [vertex] (b2) at (3, 0) {};
      \node [vertex] (b1) at (5, 0) {};
      \node [vertex] (b) at (7, 0) {};
      \node [vertex,label=left:$c_1$] (c1) at (-1.5,4) {};
      \node [corner,label=right:$c_2$] (c2) at (1.5,4) {};
      \draw[] (a2) -- (a1) -- (c1) -- (b2) -- (b1) -- (b);
      \draw[] (a) -- (c1) -- (b1);
      \draw[] (a1) -- (c2) -- (b1);
      \draw[] (a2) -- (c2) -- (b2);
      \draw[] (c1) -- (c2) -- (b);
      \draw[] (a2) -- (c1) -- (x);
      \draw[] (c1) -- (s);
      \draw[dashed] (a2) -- (b2);
      
      \draw[at edge] (c2) -- (s) -- (x) -- (a) -- (a1) -- (c2);
    \end{tikzpicture}
    \caption{$x \not\sim B$ but $x\sim \{l,r\}$}
  \end{subfigure}
  
  \medskip
  \begin{subfigure}[b]{0.3\textwidth}
    \centering
    \begin{tikzpicture}[scale=.35]
      \node [corner,label=$x$] (x) at (-3,6.44) {};
      \node [special,label=$s$] (s) at (0,6.44) {};
      \node [corner] (a) at (-7, 0) {};
      \node [special] (a1) at (-5, 0) {};
      \node [vertex] (a2) at (-3, 0) {};
      \node [vertex] (b2) at (3, 0) {};
      \node [special] (b1) at (5, 0) {};
      \node [corner] (b) at (7, 0) {};
      \node [special,label=right:$c$] (c) at (0,4) {};
      \draw[at edge] (a) -- (a1) -- (c) -- (b1) -- (b);
      \draw[at edge] (c) -- (s) -- (x);
      \draw[] (a1) -- (a2) -- (c) -- (b2) -- (b1);
      \draw[dashed] (a2) -- (b2);
    \end{tikzpicture}
    \caption{$x \sim s$ but $x \not\sim B$}
  \end{subfigure}%
  \qquad
  \begin{subfigure}[b]{0.3\textwidth}
    \centering
    \begin{tikzpicture}[scale=.35]
      \node [corner,label=$x$] (x) at (-3,6.44) {};
      \node [vertex,label=$s$] (s) at (0,6.44) {};
      \node [corner] (a) at (-7, 0) {};
      \node [vertex] (a1) at (-5, 0) {};
      \node [vertex] (a2) at (-3, 0) {};
      \node [vertex] (b2) at (3, 0) {};
      \node [vertex] (b1) at (5, 0) {};
      \node [corner] (b) at (7, 0) {};
      \node [vertex,label=left:$c_1$] (c1) at (-1.5,4) {};
      \node [vertex,label=right:$c_2$] (c2) at (1.5,4) {};
      \draw[at edge] (a) -- (c1) -- (c2) -- (b);
      \draw[at edge] (c1) -- (s) -- (x);
      \draw[at edge] (c2) -- (s);
      \draw[] (a) -- (a1) -- (a2) -- (c1) -- (b2) -- (b1) -- (b);
      \draw[] (a1) -- (c1) -- (b1);
      \draw[] (a1) -- (c2) -- (b1);
      \draw[] (a2) -- (c2) -- (b2);
      \draw[dashed] (a2) -- (b2);
    \end{tikzpicture}
    \caption{$x \sim s$ but $x \not\sim B$}
  \end{subfigure}%
  \qquad
  \begin{subfigure}[b]{0.3\textwidth}
    \centering
    \begin{tikzpicture}[scale=.35]
      \node [corner,label=$x$] (x) at (-3,6.44) {};
      \node [vertex,label=$s$] (s) at (0,6.44) {};
      \node [corner] (a) at (-7, 0) {};
      \node [vertex] (a1) at (-5, 0) {};
      \node [vertex] (a2) at (-3, 0) {};
      \node [vertex] (b2) at (3, 0) {};
      \node [vertex] (b1) at (5, 0) {};
      \node [corner] (b) at (7, 0) {};
      \node [vertex,label=left:$c_1$] (c1) at (-1.5,4) {};
      \node [special,label=right:$c_2$] (c2) at (1.5,4) {};
      \draw[at edge] (a) -- (c1) -- (c2) -- (b);
      \draw[at edge] (c1) -- (s) -- (x) -- (c1) -- (a1);
      \draw[at edge] (a) -- (a1) -- (c2) -- (s) -- (x);
      \draw[] (a1) -- (a2) -- (c1) -- (b2) -- (b1) -- (b);
      \draw[] (c1) -- (b1) -- (c2);
      \draw[] (a2) -- (c2) -- (b2);
      \draw[dashed] (a2) -- (b2);
    \end{tikzpicture}
    \caption{$x \sim c_1$ but $x \not\sim c_2,l$}
  \end{subfigure}%
  
  \caption{Adjacency between a neighbor $x$ of $s$ and centers
    [Lemma~\ref{lem:shallow}].}
  \label{fig:neighbor-of-s-is-neighbor-of-c-2}
\end{figure*}
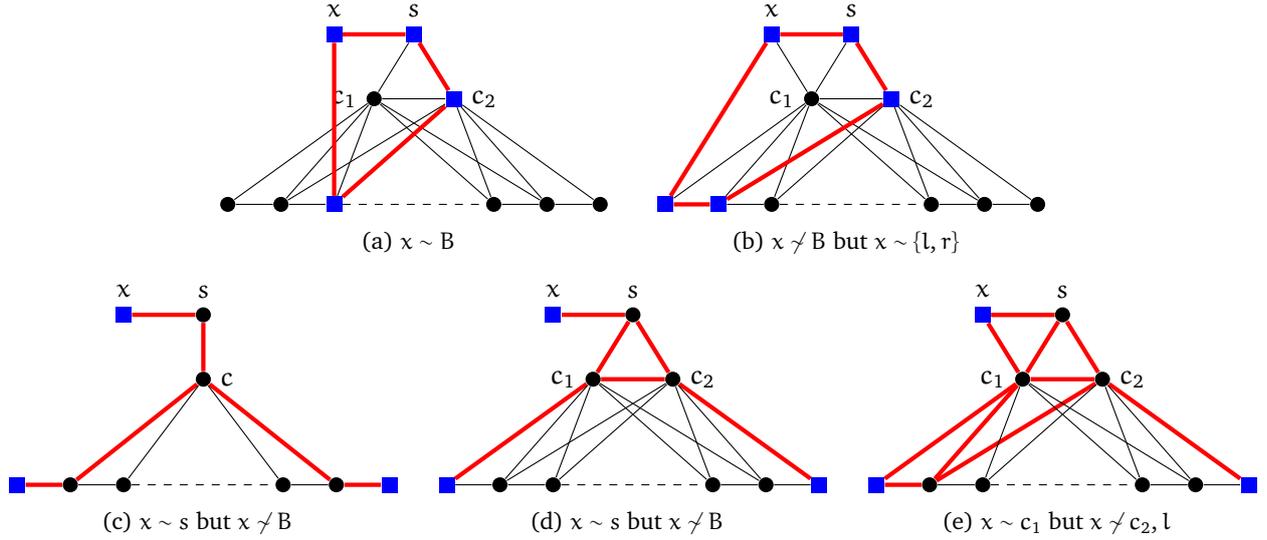
\begin{table*}[ht]
  \footnotesize
  \begin{center}
    \begin{tabular}{r|lccc}
      \hline
      & {} & {$q = p + 1$} & {$q = p + 2$} & {$q > p + 2$} \\
      \hline
      \multirow{6}{*}{$\dag$-AW} &
      $p = 0$  &  $4$-hole &  tent & $\ddag$-AW  \\
      & Fig.~\ref{fig:3.1} &  $(x c b_1 l x)^{*}$ &  $\{l, x, s, c, 
      b_2, b_1\}$ & $(s:x,c: l,b_1\dots b_{q-1},b_{q})^{**}$ 
      \\
      & & & & \\
      & $p = 1$  &  whipping top &  net &
      $\dag$-AW  \\
      & Fig.~\ref{fig:3}(b,c) &  $\{l,b_1,x,
      s,c,b_3,b_2\}$$^{***}$ &  $\{l,b_1,s,x,b_3,b_2\}$ &
      $(s:x:l,b_1\dots b_{q-1},b_{q})^{**}$ \\
      & & & & \\
      & $p > 1$  &  long claw$^{1}$ &  net& $\dag$-AW 
      \\
      & Fig.~\ref{fig:3}(d,e) &  $\{b_{p-2}, b_{p-1},b_p, s, x, 
      b_{p+2}, b_{p+1}\}$ & $\{b_{p-1},b_p,s,x, b_{q}, b_{q-1}\}$ 
      &  $(s:x:b_{p-1},b_p\dots b_{q-1},b_{q})$
      \\
      \hline
      \multirow{6}{*}{$\ddag$-AW} &
      $p = 0$  &  $4$-hole &  tent & $\ddag$-AW  \\
      & &  $(x c_2 b_1 l x)^{*}$ &  $\{l, x, s, c_2, 
      b_2, b_1\}$ & $(s:x,c_2: l,b_1\dots b_{q-1},b_{q})^{**}$ \\
      &&&&\\
      & $p = 1$  &  whipping top &  net &
      $\dag$-AW  \\
      & &  $\{l,b_1,x,s,c_2,b_3,b_2\}^{***}$ &  
      $\{l,b_1,s,x,b_3,b_2\}$ & $(s:x:l,b_1\dots b_{q-1},b_{q})^{**}$ \\
      &&&&\\          
      & $p > 1$  &  long claw &  net& $\dag$-AW \\
      & & $\{b_{p-2}, b_{p-1},b_p, s, x, 
      b_{p+2}, b_{p+1}\}$ & $\{b_{p-1},b_p,s,x, b_{q}, b_{q-1}\}$ 
      &  $(s:x:b_{p-1},b_p\dots b_{q-1},b_{q})$\\
      \hline
    \end{tabular}
  \end{center}
  
  \hspace*{20mm}{$*$} : The vertex $x$ is in  category ``none.''
  \\
  \hspace*{20mm}{$**$} : The vertex $x$ would be in  category ``full'' if 
  $q = d + 1$.
  \\
  \hspace*{20mm}{$***$} : A $4$-hole $(x b_p b_{p+1} b_{p+2} x)$ 
  would be introduced if $x \sim b_{p+2}$;
  
  \caption{Structures used in the proof of Lemma~\ref{lem:shallow}
    (category ``partial'' )}
  \label{fig:neighbor-of-shallow-terminal}
\end{table*}
\begin{figure*}[t]
  \centering
  \begin{subfigure}[b]{0.3\textwidth}
    \centering
    \begin{tikzpicture}[scale=.35]
      \node [special,label=above:$x$] (x) at (-3,4) {};
      \node [corner,label=above:$s$] (s) at (0,6.44) {};
      \node [corner,label=above:$l$] (a) at (-7, 0) {};
      \node [special,label=below:$b_1$] (a1) at (-5, 0) {};
      \node [corner] (a2) at (-3, 0) {};
      \node [vertex] (b2) at (3, 0) {};
      \node [vertex,label=above:$b_d$] (b1) at (5, 0) {};
      \node [vertex,label=above:$r$] (b) at (7, 0) {};
      \node [special,label=right:$c$] (c) at (0,4) {};
      \draw[] (c) -- (b1) -- (b);
      \draw[] (c) -- (b2) -- (b1);
      \draw[dashed] (a2) -- (b2);
      
      \draw[at edge] (a1) -- (a2) --(c) -- (s) -- (x) -- (a1);
      \draw[at edge] (x) -- (c) -- (a1) -- (a) -- (x);
    \end{tikzpicture}
    \caption{$N_B(x)=\{b_1\}$ and $x \sim l$.}
    \label{fig:3.1}
  \end{subfigure}%
  \qquad
  \begin{subfigure}[b]{.3\textwidth}
    \centering
    \begin{tikzpicture}[scale=.35]
      \node [vertex,label=above:$x$] (x) at (-3,4) {};
      \node [corner,label=above:$s$] (s) at (0,6.44) {};
      \node [corner,label=above:$l$] (a) at (-7, 0) {};
      \node [special,label=below:$b_1$] (a1) at (-5, 0) {};
      \node [vertex] (a2) at (-3, 0) {};
      \node [vertex] (b2) at (3, 0) {};
      \node [vertex,label=above:$b_d$] (b1) at (5, 0) {};
      \node [vertex,label=above:$r$] (b) at (7, 0) {};
      \node [vertex,label=right:$c$] (c) at (0,4) {};
      \node [corner,label=below:$b_3$] (a3) at (-1.5,0) {};
      \draw[] (c) -- (b1) -- (b);
      \draw[] (c) -- (b2) -- (b1);
      \draw[dashed] (a2) -- (b2);
      
      \draw[at edge] (a) -- (a1) -- (a2) -- (a3) -- (c) -- (a2);
      \draw[at edge] (x) -- (c) -- (a1) -- (x) -- (s) -- (c);
    \end{tikzpicture}
    \caption{$N_B(x)=\{b_1\}$ and $x \not\sim l$.}
    \label{fig:3.2}
  \end{subfigure}%
  \qquad
  \begin{subfigure}[b]{.3\textwidth}
    \centering
    \begin{tikzpicture}[scale=.35]
      \node [vertex,label=above:$x$] (x) at (-3,4) {};
      \node [corner,label=above:$s$] (s) at (0,6.44) {};
      \node [corner,label=above:$l$] (a) at (-7, 0) {};
      \node [vertex,label=below:$b_1$] (a1) at (-5, 0) {};
      \node [vertex] (a2) at (-3, 0) {};
      \node [vertex] (b2) at (3, 0) {};
      \node [vertex,label=above:$b_d$] (b1) at (5, 0) {};
      \node [vertex,label=above:$r$] (b) at (7, 0) {};
      \node [vertex,label=right:$c$] (c) at (0,4) {};
      \node [corner,label=below:$b_3$] (a3) at (-1.5,0) {};
      \draw[] (a1) -- (c) -- (b1) -- (b);
      \draw[] (c) -- (b2) -- (b1);
      \draw[dashed] (a2) -- (b2);
      \draw[] (a3) -- (c) -- (a2);
      \draw[] (x) -- (c) -- (s);
      
      \draw[at edge] (a) -- (a1) -- (a2) -- (a3);
      \draw[at edge] (x) -- (s);
      \draw[at edge] (a1) -- (x) -- (a2);
    \end{tikzpicture}
    \caption{$N_B(x)=\{b_1,b_2\}$.}
    \label{fig:3.3}
  \end{subfigure}%

  \begin{subfigure}[b]{0.4\textwidth}
    \centering
    \begin{tikzpicture}[scale=.35]
      \node [vertex,label=above:$x$] (x) at (-3,4) {};
      \node [corner,label=above:$s$] (s) at (0,6.44) {};
      \node [vertex,label=above:$l$] (a) at (-7, 0) {};
      \node [vertex,label=above:$b_1$] (a1) at (-5, 0) {};
      \node [vertex,label=above:$b_d$] (b1) at (5, 0) {};
      \node [vertex,label=above:$r$] (b) at (7, 0) {};
      \node [vertex,label=right:$c$] (c) at (0,4) {};
      \node [corner,label=below:$b_{i-2}$] (v1) at (-2.5,0) {};
      \node [vertex] (v2) at (-1,0) {};
      \node [special,label=below:$b_{i}$] (v3) at (0,0) {};
      \node [vertex] (v4) at (1,0) {};
      \node [corner,label=below:$b_{i+2}$] (v5) at (2.5,0) {};
      \draw[] (a) -- (a1) -- (c) -- (b1) -- (b);
      \draw[] (v1) -- (c) -- (v5);
      \draw[] (v2) -- (c) -- (v4);
      \draw[] (v3) -- (c);
      \draw[dashed] (a1) -- (v1);
      \draw[dashed] (b1) -- (v5);
      \draw[] (x) -- (c) -- (s);
      \draw[at edge] (v1) -- (v2) -- (v3) -- (v4) -- (v5);
      \draw[at edge] (s) -- (x) -- (v3);
    \end{tikzpicture}
    \caption{$N_B(x)=\{b_i\}$ ($1<i<d$).}
    \label{fig:3.4}
  \end{subfigure}%
  \quad
  \begin{subfigure}[b]{.4\textwidth}
    \centering
    \begin{tikzpicture}[scale=.35]
      \node [vertex,label=above:$x$] (x) at (-3,4) {};
      \node [corner,label=above:$s$] (s) at (0,6.44) {};
      \node [vertex,label=above:$l$] (a) at (-7, 0) {};
      \node [vertex,label=above:$b_1$] (a1) at (-5, 0) {};
      \node [vertex,label=above:$b_d$] (b1) at (5, 0) {};
      \node [vertex,label=above:$r$] (b) at (7, 0) {};
      \node [vertex,label=right:$c$] (c) at (0,4) {};
          
      \node [corner,label=below:$b_{i-1}$] (v1) at (-2.5,0) {};
      \node [vertex] (v2) at (-1,0) {};
      \node [vertex] (v4) at (1,0) {};
      \node [corner,label=below:$b_{j+1}$] (v5) at (2.5,0) {};
      \draw[] (a) -- (a1) -- (c) -- (b1) -- (b);
      \draw[] (v1) -- (c) -- (v5);
      \draw[] (v2) -- (c) -- (v4);
      \draw[dashed] (a1) -- (v1);
      \draw[dashed] (b1) -- (v5);
      \draw[] (x) -- (c) -- (s);
      
      \draw[at edge] (v1) -- (v2) -- (v4) -- (v5);
      \draw[at edge] (v2) -- (x) -- (v4);
      \draw[at edge] (s) -- (x);
    \end{tikzpicture}
    \caption{$N_B(x)=\{b_i, b_{i+1}\}$ ($1<i<d-1$).}
    \label{fig:3.5}
  \end{subfigure}%
  
  \caption{Vertex $x$ in category ``partial'' w.r.t. $W$
    [Lemma~\ref{lem:shallow}].}
  \label{fig:3}
\end{figure*}
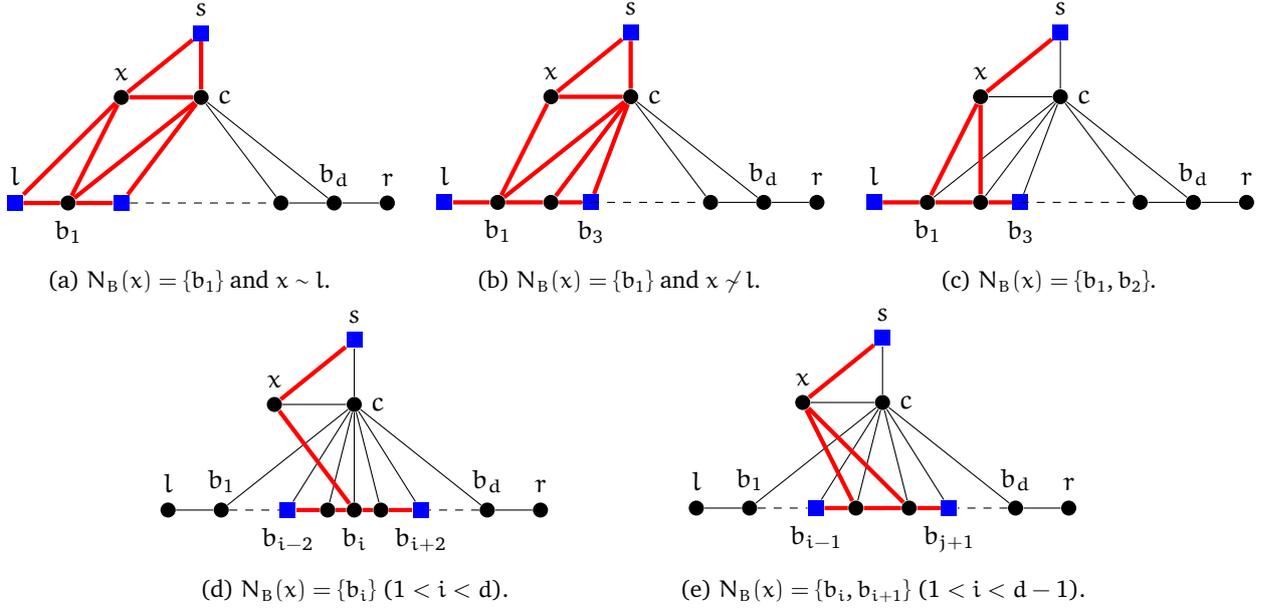
\begin{proof}
  Suppose to the contrary of statement (1), $x \not\sim c$ if $W$ is a
  $\dag$-AW or (without loss of generality) $x \not\sim c_2$ if $W$ is
  a $\ddag$-AW.  If $x \sim b_i$ for some $1 \leq i \leq d$ then there
  is a $4$-hole $(x s c b_i x)$ or $(x s c_2 b_i x)$ (See
  Figure~\ref{fig:neighbor-of-s-is-neighbor-of-c-2}(a)).  Hence we may
  assume $x \not\sim B$.  (See
  Figure~\ref{fig:neighbor-of-s-is-neighbor-of-c-2}(b,c,d,e).)  There
  is
  \begin{inparaitem}
  \item a $5$-hole $(x s c b_1 l x)$ or $(x s c b_d r x)$ if $W$ is a
    $\dag$-AW, and $x\sim l$ or $x\sim r$, respectively;
  \item a $5$-hole $(x s c_2 b_1 l x)$ or $4$-hole $(x s c_2 r x)$ if
    $W$ is a $\ddag$-AW, and $x\sim l$ or $x\sim r$, respectively;
  \item a long claw $\{x,s,c,b_1,l,b_d,r\}$ if $W$ is a $\dag$-AW and
    $x\not\sim l,r$;
  \item a net $\{x,s,l,c_1,r,c_2\}$ if $W$ is a $\ddag$-AW and $x
    \not\sim c_1,l,r$; or
  \item a whipping top $\{r,c_2,s,x,c_1,l,b_1\}$ centered at $c_2$ if
    $W$ is a $\ddag$-AW and $x\not\sim l,r$, but $x \sim c_1$.
  \end{inparaitem}
  Neither of these cases is possible, and thus statement (1) is
  proved.

  For statement (2), let us handle category ``none'' first.  Note that
  $x$, nonadjacent to $B$, cannot be a center of $W$.  If $x\sim l$,
  then there is a 4-hole $(x c b_1 l x)$ or $(x c_2 b_1 l x)$ when $W$
  is a $\dag$-AW or $\ddag$-AW, respectively.  A symmetrical argument
  will rule out $x\sim r$.  Now that $x$ is adjacent to the center(s)
  but neither base terminals nor base vertices of $W$, then
  $(x:c:l,B,r)$ (resp., $(x:c_1,c_2:l,B,r)$) makes another $\dag$-AW
  (resp., $\ddag$-AW).

  Assume now that $x$ is in category ``full.''  Suppose for
  contradiction that $x \not\sim v$ for some $v\in N(s)\setminus
  \{x\}$.  We have already proved in statement (1) that $v$ and $x$
  are adjacent to the center(s) of $W$ (different from them).  In
  particular, if one of $v$ and $x$ is a center, then they are
  adjacent. Therefore, we can assume that $v$ and $x$ are not centers.
  If $v\sim b_i$ for some $1\leq i \leq d$, then there is a $4$-hole
  $(x s v b_i x)$.  Otherwise, $v\not\sim B$, and it is in category
  ``none.''  Let $W'$ be the AW obtained by replacing $s$ in $W$ by
  $v$; then $x\sim v$ follows from
  Lemma~\ref{lem:common-neighbor-of-base}.

  Finally, assume that $x$ is in category ``partial,'' that is, $x
  \sim B$, but $x\not\sim b_i$ for some $1\le i \le d$. In this case,
  we construct the claimed AW as follows.  As the case $x \not\sim l$
  but $x \sim r$ is symmetric to $x \sim l$ but $x \not\sim r$, it is
  ignored in the following, i.e., we assume that $x \sim r$ only if $x
  \sim l$.  Let $p$ be the smallest index such that $x \sim b_p$, and
  $q$ be the smallest index such that $p < q \leq d+1$ and $x \not\sim
  b_q$ ($q$ exists by assumptions).  See
  Table~\ref{fig:neighbor-of-shallow-terminal} for the structures for
  $\dag$-AW and $\ddag$-AW respectively (see also
  Figure~\ref{fig:3}).\footnote{We omit the figure for $\ddag$-AWs:
    For a $\ddag$-AW $(s:c_1,c_2:l,B,r)$, we are only concerned with
    the relation between center $c_2$ and $B \cup \{l\}$, which is the
    same as the relation between $c$ and $B\cup \{l\}$ in a
    $\dag$-AW.}

  As the graph is prereduced and contains no small \fis, it is
  immediate from Table~\ref{fig:neighbor-of-shallow-terminal} that the
  case $q > p + 2$ holds; otherwise there always exists a small \fis.
  This completes the categorization of vertices in $N(s) \setminus T$
  and the proof.
\end{proof}
  
The proof of our main result of this section is an inductive
application of Lemma~\ref{lem:shallow}.  To avoid the repetition of
the essentially same argument for $\dag$-AWs and $\ddag$-AWs,
especially for the interaction between AWs, we use a generalized
notation to denote both.  We will uniformly use $c_1,c_2$ to denote
center(s) of an AW, and while the AW under discussion is a $\dag$-AW,
both $c_1$ and $c_2$ refer to the only center $c$.  As long as we do
not use the adjacency of $c_1$ and $l$, $c_2$ and $r$, or $c_1$ and
$c_2$ in any of the arguments, this unified (abused) notation will not
introduce inconsistencies.

\begin{theorem}
  \label{thm:shallow-is-module}
  Let $W$ be a $\dag$- or $\ddag$-AW in a prereduced graph $G$ with
  shallow terminal $s$ and base $B$.  Let $C=N(s)\cap N(B)$ and let
  $M$ be the vertex set of the component of $G-C$ containing $s$. Then
  $M$ is completely connected to $C$, and $G[C]$ is a clique.
\end{theorem}
\begin{proof}
  Denote by $W = (s:c_1,c_2:l,B,r)$, where $c_1 = c_2$ when $W$ is a
  $\dag$-AW.  Let $x$ and $y$ be any pair of vertices such that $x \in
  C$ and $y \in M$.  By definition, $G[M]$ is connected, and there is
  a chordless path $P = (v_0 \dots v_p)$ from $v_0=s$ to $v_p=y$ in
  $G[M]$.  We claim that no vertex in $P$ is adjacent to $B$.  It
  holds vacuously if $p = 1$ and then $y\sim s$; hence we assume
  $p>1$.  Suppose the contrary and let $q$ be the smallest index such
  that $v_q\sim B$. This means that every $v_i$ with $i<q$ is in
  category ``none'' of Lemma~\ref{lem:shallow}(2). Therefore, applying
  Lemma~\ref{lem:shallow}(1,2) on $v_{i}$ and AW
  $(v_{i-1}:c_1,c_2:l,B,r)$ inductively for $i = 1,\dots,q-1$, we
  conclude that there is an AW $W_i = (v_i:c_1,c_2:l,B,r)$ for each $i
  < q$.  One more application of Lemma~\ref{lem:shallow}(1) shows that
  $v_q$ is adjacent to the center(s) of $W_{q-1}$ as well.  If $v_q$
  is adjacent to all vertices of $B$, i.e., in the category ``full''
  with respect to every $W_i$, then Lemma~\ref{lem:shallow}(2) on
  $v_q$ and $W_{q-1}$ implies that $v_q$ is adjacent to $v_{q-2}\in
  N(v_{q-1})$, contradicting the assumption that $P$ is
  chordless. Otherwise (the category ``partial''), according to
  Lemma~\ref{lem:shallow}(2), there is another AW $W'=
  (v_{q-1}:c'_1,c'_2:l',B',r')$, where $B' \subset B$, and $v_q \in
  \{c'_1,c'_2\}$.  Now an application of Lemma~\ref{lem:shallow}(1) on
  $v_q$ and $W'$ shows that $v_q$ is adjacent to $v_{q-2}\in
  N(v_{q-1})$, again a contradiction.  From these contradictions we
  can conclude $P$ is disjoint from $N(B)$.  Applying
  Lemma~\ref{lem:shallow} inductively on $v_{i+1}$ and
  $W_i=(v_i:c_1,c_2:l,B,r)$, we get an AW with the same centers for
  every $0 \le i \le p$.

  As $x$ is adjacent to both $s$ and $B$, it cannot be in category
  ``none'' with respect to $W$.  We now separate the discussion based
  on whether $x$ is in the category ``full'' or ``partial.''  Suppose
  first that $x$ is in the category ``full''; as $x\in N(s)$,
  Lemma~\ref{lem:shallow}(1) implies that $x\sim c_1,c_2$. Then
  applying Lemma~\ref{lem:shallow}(2) inductively, where $i =
  1,\dots,p$, on vertex $x$ and $W_{i-1}$ we get that $x \sim v_i$ for
  every $i \leq p$; in particular, $x \sim v_p$ ($= y$).  Suppose now
  that $x$ is in in category ``partial.'' Then by
  Lemma~\ref{lem:shallow}(2), there is an AW $W'_0 =
  (v_0:c'_1,c'_2:l',B',r')$, where $B' \subset B$, and $x\in
  \{c'_1,c'_2\}$.  For any $0\le i\le p$, as $v_i\not\sim B$, we have
  $v_i\not\sim B'$ as well, i.e., $v_i$ is in category ``none'' with
  respect to $W'_0$. Therefore, by an inductive application of
  Lemma~\ref{lem:shallow}(2) on the vertex $v_i$ and AW
  $W'_{i-1}=(v_{i-1}:c'_1,c'_2:l',B',r')$ for $i = 1,\dots,p$, we
  conclude that there is an AW $W'_p = (v_p:c'_1,c'_2:l',B',r')$, from
  which $x\sim y$ follows immediately.

  Now we show the second assertion.  For any pair of vertices $x$ and
  $y$ in $C$, we apply Lemma~\ref{lem:shallow} on $x$ and $W$; by
  definition, $x\sim B$ and thus cannot be in category ``none.''  If
  $x$ is in category ``full'' with respect to $W$, then
  Lemma~\ref{lem:shallow}(2) and the fact $y\in N(s)$ imply that
  $x\sim y$.  Otherwise, if $x$ is in category ``partial'' with
  respect to $W$, then Lemma~\ref{lem:shallow}(2) implies that there
  is an AW $W' = (s:c'_1,c'_2:l',B',r')$ where $B' \subset B$ and
  $x\in\{c'_1,c'_2\}$. Therefore, by Lemma~\ref{lem:shallow}(1) on the
  vertex $y\in N(s)$ and $W'$, we get that $y\sim c'_1,c'_2$ and hence
  $x\sim y$.
\end{proof}

Now Theorem~\ref{thm:reduced-instances-ap} follows from
Theorem~\ref{thm:shallow-is-module}: the set $M$ containing $s$ is in
a module whose neighborhood is a clique, hence every vertex in $M$ is
simplicial.  We point out that the set $C$ is the minimal $M$-$B$
separator.

\section{Long holes}\label{sec:holes-in-modules}

This section proves Theorem~\ref{thm:reduced-instances-b} by showing that
the holes in a reduced graph are pairwise congenial. During the study
of vertices of a hole, their indices become very subtle.  To simplify
the presentation, we will frequently apply a common technique, that
is, to number the vertices of a hole starting from a vertex of special
interest for the property at hand.  Needless to say, indexing two
adjacent vertices in a hole will determine the indices of all the
vertices in the hole, as well as the ordering used to traverse the
hole.  All indices of vertices in a hole $H$ should be understood as
modulo $|H|$, e.g., $h_{0} = h_{|H|}$.

We start from two simple facts on the relations between vertices and
holes, from which we derive the relations between two holes, and
finally generalize them to multiple holes.

\begin{proposition}
  \label{lem:neighbors-are-consecutive}
  For any vertex $v$ and hole $H$ of a prereduced graph, $N_H[v]$ are
  consecutive in $H$.  Moreover, either $N_H[v] = H$ or $|N_H[v]| <
  |H| - 7$.
\end{proposition}
\begin{proof}
  Both assertions are trivially true when $v\not\sim H$, $N_H[v] = H$,
  or $v \in H$ (then $|N_H[v]| = 3$); it is hence assumed that none of
  them holds true.  We number the vertices of $H$ in a way that $h_0
  \sim v$ but $h_1 \not\sim v$.  Suppose the first assertion is not
  true, then we can find the following three vertices of $H$, whose
  existence is clear from assumptions:
  \begin{inparaenum}[(\itshape a\upshape)]
  \item ${p_1}$ is the smallest index such that $p_1 > 1$ and $h_{p_1}
    \sim v$;
  \item ${p_2}$ is the smallest index such that $p_2 > p_1$ and
    $h_{p_2} \not\sim v$; and
  \item ${p_3}$ is the smallest index such that ${p_2} < p_3 < |H|$
    and $h_{p_3} \sim v$, or $p_3 = |H|$ if $h_i\not\sim v$ for each
    $i > p_2$ (then $h_{p_3} = h_0$).
  \end{inparaenum}
  By the selection of $p_1$, $p_2$, and $p_3$, we have $v \sim h_i$
  for each $p_1 \leq i < p_2$ and $i=0,p_3$, but $v \not\sim h_j$ for
  $0 < j < p_1$ or $p_2 \leq j < p_3$ (relations between $v$ and $h_i$
  where ${p_3} < i < |H|$ are immaterial).
 
  Now we examine the distances between the three indices.  If $p_1 <
  4$ or $p_3 < p_2 + 4$, then there is a small hole, $(v h_{0} h_1
  \dots h_{p_1} v)$ or $(v h_{p_2 - 1} h_{p_2} \dots h_{p_3} v)$,
  respectively, of length at most $6$.  Thus $4 \leq p_1 < p_2 < p_2 +
  4 \leq p_3$.  Nonetheless, there is
  \begin{inparaitem}
  \item a long claw $\{h_0,v,h_{p_1},h_{p_1 - 2}, h_{p_1 - 1},
    h_{p_2}, h_{p_2 + 1}\}$ if $p_2 = p_1 + 1$;
  \item a net $\{h_0,v,h_{p_1 - 1},h_{p_1}, h_{p_2}, h_{p_2 - 1}\}$ if
    $p_2 = p_1 + 2$; or
  \item a long claw $\{h_{1},h_{0},v,h_{p_1-1},h_{p_1},h_{p_2},h_{p_2
      - 1}\}$ if $p_2 > p_1 + 2$.
  \end{inparaitem}
  None of these \fiss\ involves both $h_0$ and $h_{p_3}$ and thus they
  exist regardless of $h_{p_3} = h_0$ or not.  This contradiction
  ensures that $N_H[v]$ are consecutive, so are the vertices of $H
  \setminus N_H[v]$.  On the second assertion, note that there is a
  hole of length at most $10$ if $1 \leq |H \setminus N_H[v]| \leq 7$.
\end{proof}

Recall that $\widehat N(H)$ is the set of all common neighbors of the
hole $H$.  If a vertex $v\not\in\widehat N(H)$ is adjacent to more
than three vertices of $H$, then we can use $v$ as a shortcut for the
inner vertices of the path induced by $N_H[v]$ to obtain another hole
that is strictly shorter than $H$.  
\begin{corollary}\label{lem:no-more-than-3}
  Let $H$ be a shortest hole.  If $v \not\in \widehat N(H)$, then
  $N_H[v] \leq 3$.
\end{corollary}
Note that each hole $H$ in a prereduced graph contains at least $11$
vertices.  If $v \in \widehat N(H)$, then on any five consecutive
vertices of the hole $H$ and $v$, Proposition~\ref{lem:five-path}(1)
applies, which implies that $v$ is dominating in the closed
neighborhood of $H$.
\begin{corollary}\label{lem:common-neighbor-of-holes}
  Let $H$ be a hole in a prereduced graph.  If $v \in \widehat N(H)$,
  then $v$ is adjacent to all vertices in $N[H] \setminus \{v\}$.
\end{corollary}

So far we characterized neighbors of holes in a prereduced graph: Any
vertex $v$ is adjacent to a (possibly empty) set of consecutive
vertices of a hole $H$; if $v$ is adjacent to all vertices of $H$,
then it is also adjacent to every neighbor of $H$.  From these facts
we now derive the relations between holes.  Following is the most
crucial concept of the section:

\begin{definition}
  Two holes $H_1$ and $H_2$ are called \emph{congenial} (to each
  other) if each vertex of one hole is a neighbor of the other hole,
  that is, $H_1 \subseteq N[H_2]$ and $H_2 \subseteq N[H_1]$.
\end{definition}

We remark that every hole is congenial to itself by definition.  
The definition is partially motivated by:

\begin{proposition}\label{lem:congenial-holes-cover}
  Let $\cal H$ be a set of holes all congenial to $H$.  For each $v
  \in H$, every hole in $\cal H$ intersects $N[v]$.
\end{proposition}

Since a vertex in a hole cannot be a common neighbor of it,
Corollary~\ref{lem:common-neighbor-of-holes} and the definition of
congenial holes immediately imply:

\begin{corollary}\label{lem:congenial-not-dominating}
  For any pair of congenial holes $H_1$ and $H_2$ in a prereduced
  graph, $\widehat N(H_1) = \widehat N(H_2)$.  Moreover, no vertex of
  $H_1$ (resp., $H_2$) is a common neighbor of $H_2$ (resp., $H_1$).
\end{corollary}

We analyze next the relation between two non-congenial holes.  It
turns out that if not all vertices of a hole $H_1$ are adjacent to
another hole $H_2$, then, as shown in the following lemma, every
vertex of $H_1$ is adjacent to either all or none of the vertices of
$H_2$.

\begin{lemma}
  \label{lem:adjacent-holes}
  Let $H_1$ and $H_2$ be two adjacent holes in a prereduced graph 
  such that $H_1 \not\subseteq N[H_2]$.  Each neighbor of $H_2$ in
  $H_1$ is a common neighbor of $H_2$, i.e., $N_{H_1}[H_2] \subseteq
  \widehat N(H_2)$.  In particular, $H_1$ and $H_2$ are disjoint.
\end{lemma}
  \begin{proof}
    Let $u$ be any vertex in $N_{H_1}[H_2]$, which is nonempty by
    assumption, and let $P$ be the maximal path in $H_1$ with the
    property that $u \in P \subseteq N_{H_1}[H_2]$; denote by $p$ the
    number of vertices of $P$.  Note that some vertices of $P$ can
    belong to $H_2$ (in particular, $u$ can be in $H_2$).  Observe
    that $p<|H_1|$, as by assumption, $H_1$ is not contained in
    $N[H_2]$.  Numbering the vertices in $H_1$ such that $P = u_0
    \dots u_{p-1}$ (the ordering of $H_1$ is immaterial when $p = 1$
    and then $u_1$ can be either neighbor of $u_0$ in $H_1$), the
    selection of $P$ means $u_i \sim H_2$ for each $0 \leq i < p$, and
    $u_{-1},u_{p} \not\sim H_2$ (it is immaterial whether $u_{-1} =
    u_p$ or not).  In the following, we show that both ends of $P$
    belong to $\widehat N(H_2)$, which induces a clique
    (Proposition~\ref{lem:common-neighbors-is-clique}). Thus either
    $u_0=u_{p-1}$ (i.e., $p=1$) or $u_0$ and $u_{p-1}$ are adjacent
    (i.e., $p=2$); in either case, we have $u \in \{u_0,u_{p-1}\}
    \subseteq \widehat N(H_2)$.  This proves the first assertion, and
    the second assertion ensues, as otherwise their common vertices
    will be common neighbors of $H_2$, which is not possible.

    Note that $u_0\not\in H_2$, as otherwise $u_{-1}$ is also adjacent
    to $H_2$, contradicting the maximality of $P$.  Similarly,
    $u_{-1},u_{-2}\not\in H_2$.  If $u_0$ has a unique neighbor $v$ in
    $H_2$, then the subgraph induced by $u_{-1}$, $u_{0}$ and five
    consecutive $H_2$ vertices centered at $v$ is a long claw (see
    Figure~\ref{fig:non-adjacent-hole-1}).  Now we consider the case
    $2 \leq |N_{H_2}[u_0]| \leq |H_2| - 7$
    (Proposition~\ref{lem:neighbors-are-consecutive}), and number the
    vertices of $H_2$ such that $N_{H_2}[u_0] = \{v_1 ,v_2 \dots,
    v_q\}$.  Note that $|N_{H_2}[u_0]|\leq |H_2|-7$ implies that
    $v_0\neq v_{q+1}$. If $u_{-2}$ is adjacent to $v_0$, $v_1$, $v_q$,
    or $v_{q+1}$, then there is a hole $(u_{-2} u_{-1} u_0 v_1 v_0
    u_{-2})$, $(u_{-2} u_{-1} u_0 v_1 u_{-2})$, $(u_{-2} u_{-1} u_0
    v_q u_{-2})$, or $(u_{-2} u_{-1} u_0 v_q v_{q+1} u_{-2})$,
    respectively.  Otherwise, $u_{-2}\not\sim \{v_0, v_1, v_{q},
    v_{q+1}\}$, then there is a net
    $\{u_{-1},u_{0},v_0,v_1,v_{q+1},v_q\}$ when $|N_{H_2}(u_0)|
    = 2$, or long claw $\{u_{-2},u_{-1},u_{0},v_0,v_1,v_{q+1},v_q\}$
    when $|N_{H_2}(u_0)| > 2$ (see
    Figure~\ref{fig:non-adjacent-hole-2}).  This proves $u_0 \in
    \widehat N(H_2)$, and with a symmetrical argument we can also
    prove $u_{p-1} \in \widehat N(H_2)$.
  \end{proof}
\begin{figure*}[t]
  \centering
  \begin{subfigure}[b]{0.45\textwidth}
    \small
    \centering
    \begin{tikzpicture}[scale=.45]
      \node[special,label=267:$u_0$] (u0) at (45:3) {};
      \node[corner,label=267:$u_{\!\!-\!1}$] (u1) at (75:3) {};
      \node[special,label=left:$u_p$] (up) at (305:3) {};
      \node[special,label=267:$u_{1}$]  at (20:3) {};
      \node[label=right:$H_1$]  at (180:3) {};

      \node[vertex,xshift=3.7cm] (v0) at (130:3) {};
      \node[corner,xshift=3.7cm] (v1) at (170:3) {};
      \node[vertex,xshift=3.7cm]  at (150:3) {};
      \node[vertex,xshift=3.7cm]  at (110:3) {};
      \node[corner,xshift=3.7cm]  at (90:3) {};
      \node[label=left:$H_2$,xshift=3.7cm]  at (0:3) {};

      \draw[dashed,blue,thick] (u0) arc (45:-50:3);
      \draw[blue,thick] (u0) arc (45:20:3);
      \draw[dashed] (u0) arc (45:315:3);
      \draw[dashed] (11.2,0) arc (0:360:3);

      \draw[at edge] (u1) arc (75:45:3) -- (v0);
      \draw[at edge] (v1) arc (170:90:3);
    \end{tikzpicture}

    \caption{$|N_{H_2} (u_0)| = 1$.}
    \label{fig:non-adjacent-hole-1}
  \end{subfigure}
  \qquad
  \begin{subfigure}[b]{0.45\textwidth}
    \centering
    \begin{tikzpicture}[scale=.45]
      \node[special,label=267:$u_0$] (u0) at (20:3) {};
      \node[corner,label=267:$u_{\!-\!1}$] (u1) at (45:3) {};
      \node[corner,label=267:$u_{\!-\!2}$] (u4) at (75:3) {};
      \node[special,label=left:$u_p$] (up) at (305:3) {};
      \node[label=right:$H_1$]  at (180:3) {};

      \node[vertex,xshift=3.7cm] (v2) at (190:3) {};
      \node[vertex,xshift=3.7cm] (v3) at (140:3) {};
      \node[corner,xshift=3.7cm]  at (215:3) {};
      \node[vertex,xshift=3.7cm] (v4) at (172:3) {};
      \node[vertex,xshift=3.7cm] (v5) at (158:3) {};
      \node[corner,xshift=3.7cm]  at (115:3) {};
      \node[label=left:$H_2$,xshift=3.7cm]  at (0:3) {};

      \draw[] (v4) -- (u0) -- (v5);
      \draw[dashed,blue,thick] (u0) arc (20:-50:3);
      \draw[dashed] (u0) arc (20:315:3);
      \draw[dashed] (11.2,0) arc (0:360:3);

      \draw[at edge] (u4) arc (75:20:3) -- (v2) arc (190:215:3);
      \draw[at edge] (u0) -- (v3) arc (140:115:3);
    \end{tikzpicture}

    \caption{$|N_{H_2} (u_0)| > 1$.}
    \label{fig:non-adjacent-hole-2}
  \end{subfigure}

  \caption{ Adjacencies between two non-congenial holes
    ($u_{-1}\not\sim H_2$)}
  \label{fig:neighboring-holes}
\end{figure*}
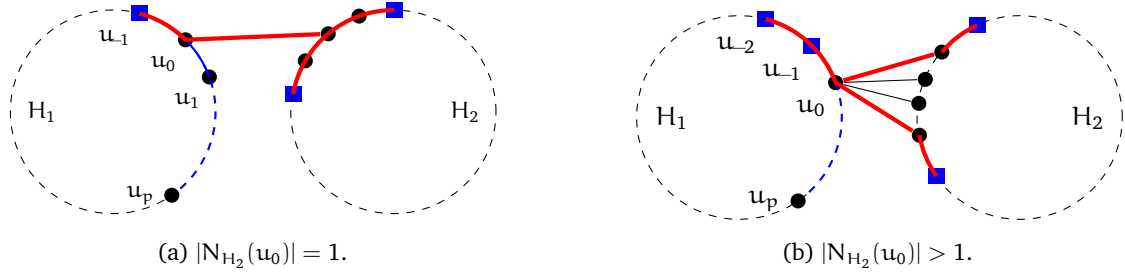

We are now ready to establish the transitivity of the congenial
relation.  The reflexivity and symmetry of this relation are clear
from definition; therefore congenial holes form an equivalence class.

\begin{lemma}
  \label{lem:congenial-equivalence}
  Let $H$, $H_1$, and $H_2$ be three holes in a prereduced graph $G$.  If
  both $H_1$ and $H_2$ are congenial to $H$, then $H_1$ and $H_2$ are
  congenial.
\end{lemma}
  \begin{proof} 
    According to Corollary~\ref{lem:congenial-not-dominating},
    $\widehat N(H_1) = \widehat N(H) = \widehat N(H_2)$.  If $H_1$ and
    $H_2$ are adjacent, then they have to be congenial, as otherwise
    Lemma~\ref{lem:adjacent-holes} implies that one of them contains a
    common neighbor of the other, hence a common neighbor of all three
    holes, which is impossible.  Assume hence that there is no edge
    between $H_1$ and $H_2$; in particular, they are disjoint.  Let
    $h$ be any vertex in $H$, and we number the vertices of $H_1$ and
    $H_2$ such that $N_{H_1}[h] = \{u_1,\dots,u_p\}$ and $N_{H_2}[h] =
    \{v_1,\dots,v_q\}$.
    Proposition~\ref{lem:neighbors-are-consecutive} implies that
    $u_0\neq u_{p+1}$ and $v_0\neq v_{p+1}$.  Note that $h$ is
    adjacent to some but not all vertices of both $H_1$ and $H_2$.
    There is
    \begin{inparaitem}
    \item a long claw $\{v_1,h,u_{-1},u_0,u_1,u_2,u_3\}$ when $p=1$;
    \item a net $\{v_1,h,u_0,u_1,u_3,u_2\}$ when $p=2$; or
    \item a long claw $\{v_0,v_1,u_0,u_1,h,u_p,u_{p+1}\}$ when $p\geq
      3$.
    \end{inparaitem}
  \end{proof}

  To prove Theorem~\ref{thm:reduced-instances-b}, we show that if
  there are two holes that are not congenial, then one of them is
  contained in a nontrivial module.  This is impossible in a reduced
  graph, where every nontrivial module induces a clique.  We construct
  this nontrivial module with the help of the following lemma, which
  shows that the common neighbors form a separator.
\begin{lemma}
  \label{lem:hole-is-module}
  Let $H$ be a hole that is the shortest among all the holes congenial
  to it in a prereduced graph $G$.  Then the set $\widehat N(H)$ of
    common neighbors of $H$ separates $N[H] \setminus \widehat N(H)$
  from $V(G) \setminus N[H]$.
\end{lemma}
  \begin{proof}
    Suppose to the contrary, $N[H] \setminus \widehat N(H)$ and $V(G)
    \setminus N[H]$ are still connected in $G - \widehat N(H)$, then
    there is a pair of adjacent vertices $u \in N[H] \setminus
    \widehat N(H)$ and $v \in V(G) \setminus N[H]$.  Note that $u
    \not\in H$, and we have two adjacent vertices only one of which is
    adjacent to part of the hole $H$.  Depending on the number of
    neighbors of $u$ in $H$, we have either a long claw (when
    $|N_H(u)| = 1$), a net (when $|N_H(u)| = 2$), or a $\dag$-AW of
    size $7$ (when $|N_H(u)| = 3$), none of which can exist in a
    prereduced graph.  On the other hand, if $|N_H(u)| > 3$ then we
    can use $u$ to find another hole $H'$ that is strictly shorter
    than $H$; it is surely congenial to $H$, which contradicts the
    assumption.
  \end{proof}

We are now ready to prove Theorem~\ref{thm:reduced-instances-b}:
\paragraph{Theorem 2.2 (restated).}
  {All holes in a reduced graph are congenial to each other.}
\begin{proof}
  Suppose, for contradiction, that not all holes are congenial to each
  other.  By Lemma~\ref{lem:congenial-equivalence}, being congenial is
  an equivalence relation.  Hence there are two equivalence classes of
  holes, from each of which we pick a shortest one; let them be $H_1$
  and $H_2$.  Assume without loss of generality that $H_2$ has a
  vertex $v$ not in $N[H_1]$.  Lemma~\ref{lem:hole-is-module} implies
  that $\widehat N(H_1)$ separates $N[H_1]\setminus \widehat N(H_1)$
  and $V(G)\setminus N[H_1]$.  Either $\widehat N(H_1)=\emptyset$ and
  then $G$ is disconnected where $N[H_1]$ induces a connected
  component ($v\not\in N[H_1]$); or $\widehat N(H_1)$ is the neighbor
  of $N[H_1]$ and they are completely connected
  (Corollary~\ref{lem:common-neighbor-of-holes}).  In either case,
  the set $N[H_1]\setminus \widehat N(H_1)$ is a nontrivial module
  that does not induce a clique.  Thus
  Reduction~\ref{rule:non-interval-modules} is applicable and the
  graph is not reduced.
\end{proof}

\section{Hole covers}
\label{sec:holes}
A set of vertices is called a \emph{hole cover} of a graph $G$ if it
intersects every hole in $G$, and the removal of any hole cover makes
the graph chordal.  {A hole cover is minimal if any proper subset of
  it is not a hole cover.  Any {interval deletion} set makes a hole
  cover of the input graph, and thus contains a minimal hole cover.}
The goal of this section is to prove
Theorem~\ref{thm:alg-reduced-with-holes}, that is, to provide a
polynomial bound on the number of minimal hole covers in a reduced
graph and give a polynomial time algorithm to find all of them.

To simplify the task, observe that no minimal hole cover contains a
vertex that is not in any hole.

\begin{proposition}\label{lem:hole-cover-in-reduced-graphs}
  Let $\cal H$ be the set of all holes in a reduced graph $G$, and
  $G_0$ be the subgraph induced by $\bigcup_{H \in \cal H}H$.  A set
  $HC$ of vertices is a minimal hole cover of $G$ if and only if it is
  a minimal hole cover of $G_0$.
\end{proposition}

In this section we will focus on the subgraph $G_0$ induced by the
union of all holes in the reduced graph $G$.  The subgraph $G_0$ has
the same set of holes as $G$, and they remain pairwise congenial.
Moreover, each vertex of $G_0$ is in the closed neighborhood of each
hole $H$ of $G_0$, which means $G_0$ is connected.  As we have said
earlier, circular-arc graphs form an important example of graphs of
which all holes are pairwise congenial.  Thinking of $G_0$ as a
circular arc graph gives the intuition behind most statements to
follow.  But since this fact is not directly used in this paper, we
are not giving a proof for it.

\begin{proposition}
  \label{lem:interval-subgraph}
  The subgraph $G_0 - HC$ is an interval graph for each hole cover
  $HC$ of $G_0$.
\end{proposition}
\begin{proof}
  Each vertex of $G_0$ belongs to some hole, and thus cannot be
  simplicial. Therefore, by Theorem~\ref{thm:reduced-instances-ap},
  $G_0$ contains no AW.  By definition, $G_0 - HC$ contains no hole;
  thus $G_0-HC$ is an interval graph.
\end{proof}

In what follows we prove a series of claims on how the neighborhood of
a vertex $v$ of a hole $H_1$ looks like in another hole $H_2$.  The
first statement is a paraphrase of
Corollary~\ref{lem:congenial-not-dominating}:
\begin{corollary}\label{lem:hole-not-common}
  No vertex $v$ of $G_0$ can be a common neighbor of any hole in
  $G_0$.
\end{corollary}

Therefore, by definition of congenial holes and
Proposition~\ref{lem:neighbors-are-consecutive}, we can assume that
for every $v\in V(G_0)$ and hole $H$, we have that $N_H[v]$ is a
proper nonempty subset of $H$ and its vertices induce a path in $H$.
Fixing any ordering of the vertices in $H$, we can denote two ends of
the path as $\mathtt{begin}_H(v)$ and $\mathtt{end}_H(v)$
respectively; when $N_H[v]$ contains both $h_0$ and $h_{|H|-1}$, we
number vertices of $N_H[v]$ as $\{h_{-p},\dots,h_0,\dots,h_q\}$ where
both $p$ and $q$ are nonnegative, and then $\mathtt{begin}_H(v) = -p$,
$\mathtt{end}_H(v) = q$.

\begin{proposition}
  \label{lem:neighbors-of-neighbors}
  Let $H$ be a hole of $G_0$.  For any pair of adjacent vertices $u,v$
  of $G_0$, their closed neighborhoods in $H$ satisfy the following
  properties.
  \begin{enumerate}[(1)]
  \item $N_H[u]\cap N_H[v]\ne \emptyset$ and $N_H[u]\cup N_H[v]\ne
    H$.
  \item If $v$ is adjacent to neither $h_{\mathtt{begin}_H(u)}$ nor
    $h_{\mathtt{end}_H(u)}$, then $N[v]\subset N[u]$.
  \end{enumerate}
\end{proposition}
\begin{proof}
  We number vertices of $H$ such that $N_{H}[u] =
  \{h_0,\dots,h_{\ell}\}$; the order can be either way if
  $|N_{H}[u]| = 1$, i.e., $\ell = 0$.

  Statement (1) holds trivially if either or both of $u$ and $v$
  belong to $H$ (Proposition~\ref{lem:neighbors-are-consecutive}).  Hence we
  assume $u,v\not\in H$.  Suppose first, for contradiction, $N_{H}[u]
  \cap N_{H}[v] = \emptyset$; we may assume
  $\{h_{\ell_1},\dots,h_{\ell_2}\} = N_{H}[v]$ where $\ell < \ell_1
  \leq \ell_2 < |H|$.  If $\ell_2 \geq |H| - 3$, then $(u v h_{\ell_2}
  \dots h_{|H|} u)$ is a hole of length at most $6$.  Otherwise, $(u v
  h_{\ell_1} \dots h_{\ell} u)$ is a hole not congenial to $H$: in
  particular, the vertex $h_{|H| - 2}$ in $H$ is nonadjacent to it.
  In either case, we end with a contradiction; hence $N_{H}[u]$ and
  $N_{H}[v]$ must intersect.  Suppose now, for contradiction,
  $N_{H}[u] \cup N_{H}[v] = H$.  Then $v$ is adjacent to every vertex
  in ($h_{\ell + 1} h_{\ell + 2} \cdots h_{|H| - 1}$).
  Proposition~\ref{lem:neighbors-are-consecutive} and
  Corollary~\ref{lem:hole-not-common} imply $6< \ell < |H| - 6$.  If
  $v\not\sim h_{\ell}$, then $(u h_{\ell} h_{\ell + 1} v u)$ is a
  $4$-hole.  A symmetric argument applies when $v\not\sim h_{0}$.  Now
  suppose $v$ is adjacent to both $h_0$ and $h_{\ell}$, then ($u
  h_{\ell} h_{\ell+1}\cdots h_{|H| - 1} h_0 u$) is a hole and $v$ is a
  common neighbor of it (contradicting
  Corollary~\ref{lem:hole-not-common}).  None of the cases is possible, and
  hence $N_H[u]\cup N_H[v]\ne H$.

  The condition of statement (2) means that $v\not\sim h_0$ and
  $v\not\sim h_\ell$.  According to statement (1), and since $N_H[v]$
  is consecutive in $H$, we must have $N_{H}[v] \subseteq \{h_1,
  \dots,h_{\ell-1}\}$.  Note that $N_H[v]$ is nonempty and thus
  $\ell\ge 2$.  If $u\in H$, then $N_H[v] = \{u\} = \{h_1\}$, and
  statement (2) follows from statement (1); here we use the fact that
  every $x\in N[v]$ is in the neighborhood of $H$ and that $u$ is the
  only neighbor of $v$ in $H$.  Assume now $u\not\in H$; the argument
  below holds regardless of whether $v\in H$ or not.
  Let $x$ be any vertex in $N[v]$ different from $u$, and we argue
  $x\sim u$.  By statement (1), $N_H[x]$ must intersect $\{h_1,
  \cdots, h_{\ell - 1}\}$.  If $x$ is nonadjacent to
  $\{h_0,h_{\ell}\}$, then $N_H[x]$ is also a subset of $\{h_1,
  \cdots, h_{\ell - 1}\}$, and $x\sim u$ follows from
  Proposition~\ref{lem:five-path}(3) (taking ($h_{-1} h_0\cdots h_{\ell}
  h_{\ell+1}$) as the path).  Otherwise, $x$ is adjacent to at
  least one of $\{h_0,h_{\ell}\}$, then $x\not\sim u$ will imply a
  $4$-hole ($h_0 u v x h_0$) or ($h_{\ell+1} u v x h_0$), which is
  impossible.  The proof is now completed.
\end{proof}
Noting that the closed neighborhoods of two consecutive vertices in a
hole are incomparable, Proposition~\ref{lem:neighbors-of-neighbors}(2) has
the following corollary.
\begin{corollary}
  \label{lem:neighbors-of-neighbors-in-holes}
  Let $H$ and $H_1$ be two holes in $G_0$.  For each pair of
  consecutive vertex $u_i, u_{i+1}\in H_1$, at least one end of
  $N_H[u_i]$ is in   $N_H[u_{i+1}]$.
\end{corollary}

The following lemmas characterize minimal hole covers of $G_0$.
\begin{lemma}\label{lem:hole-cover-is-clique}
  Any minimal hole cover of $G_0$ induces a clique.
\end{lemma}
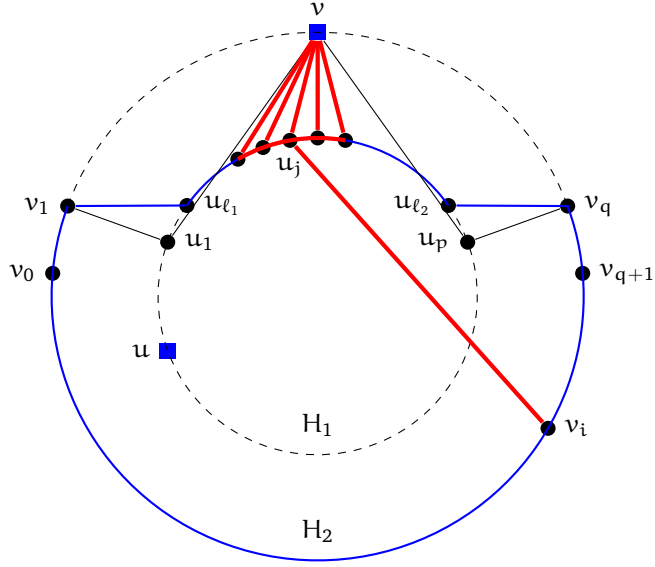
\begin{figure*}[t]
  \centering
  \begin{subfigure}[b]{0.75\textwidth}
    \centering
    \begin{tikzpicture}[scale=.7]
      \node[special,label=right:$u_1$] (u0) at (160:3) {};
      \node[special,label=left:$u_p$] (up) at (20:3) {};
      \node[vertex,label=below:$u_j$] (uj) at (100:3) {};
      \node[vertex] (u1) at (120:3) {};
      \node[vertex] (u2) at (110:3) {};
      \node[vertex] (u3) at (90:3) {};
      \node[vertex] (u4) at (80:3) {};
      \node[vertex,label=right:$u_{\ell_1}$] (u5) at (145:3) {};
      \node[vertex,label=left:$u_{\ell_2}$] (u6) at (35:3) {};
      \node[corner,label=left:$u$] (u) at (200:3) {};
      \node[label=above:$H_1$]  at (270:3) {};

      \node[corner,label=above:$v$] (v) at (90:5) {};
      \node[special,label=left:$v_1$] (v1) at (160:5) {};
      \node[special,label=left:$v_0$] (v0) at (175:5) {};
      \node[special,label=right:$v_q$] (vq) at (20:5) {};
      \node[special,label=right:$v_{q+1}$] (vqq) at (5:5) {};
      \node[vertex,label=right:$v_i$] (vi) at (330:5) {};
      \node[label=above:$H_2$]  at (270:5) {};

      \draw[blue,thick] (u6) arc (35:145:3) -- (v1) arc (160:380:5) --
      (u6);
      \draw[] (vq) -- (up) -- (v) -- (u0) -- (v1);
      \draw[dashed] (u5) arc (145:395:3);
      \draw[dashed] (v1) arc (160:20:5);

      \draw[at edge] (v) -- (uj) -- (vi);
      \foreach \s in {1,...,4}
        \draw[at edge] (u\s) -- (v);
      \draw[at edge] (u4) arc (80:120:3);
    \end{tikzpicture}
  \end{subfigure}

  \caption{Illustration for the proof of Lemma~\ref{lem:hole-cover-is-clique}.}
  \label{fig:hole-cover-is-clique}
\end{figure*}
\begin{proof}
  Suppose to the contrary, there is a minimal hole cover $HC$ that
  contains two nonadjacent vertices $u$ and $v$.  By the minimality of
  $HC$, there are two (unnecessarily disjoint) holes $H_1$ and $H_2$
  such that $HC \cap H_1 = \{u\}$ and $HC \cap H_2 = \{v\}$.  In
  particular, $u \not\in H_2$ and $v\not\in H_1$.  We number the
  vertices of $H_1$ such that $N_{H_1}[v] = \{u_1,u_2,\cdots, u_p\}$.
  The union of $N_{H_2}[u_1]$ and $N_{H_2}[u_p]$ is a consecutive set of
  vertices in $H_2$: they both contain $v$, and, by
  Proposition~\ref{lem:neighbors-are-consecutive}, are consecutive in $H_2$.
  We number the vertices of $H_2$ such that $u_1 \sim v_1$ and
  $N_{H_2}[u_1] \cup N_{H_2}[u_p] =\{v_1, \dots, v_q\}$.
    \begin{claim}\label{claim-0}
      At least one vertex of $H_2$ is adjacent to neither $u_1$ nor
      $u_p$.
    \end{claim}
    \begin{proof}
      The claim follows from Proposition~\ref{lem:neighbors-of-neighbors}(1)
      when $p = 2$; hence we may assume $p > 2$, which means $u_1
      \not\sim u_p$ (note that $u_0\neq u_{p+1}$).  Suppose
      $N_{H_2}[u_1] \cup N_{H_2}[u_p] = H_2$, then by
      Proposition~\ref{lem:neighbors-are-consecutive}, we have $7 <
      |N_{H_2}[u_1]| < |H_2| - 7$, which means at least one end of the
      path induced by $N_{H_2}[u_1]$ is not adjacent to $v$.  Without
      loss of generality, let it be $v_i$ where $i =
      \mathtt{end}_{H_2}[u_1]$; noting that by assumption $v_{i+1}\sim
      u_p$, there is either a $4$-hole $(v u_1 v_i u_p v)$ (if
      $v_i\sim u_p$) or a $5$-hole $(v u_1 v_i v_{i+1} u_p v)$ (if
      $v_i\not\sim u_p$).  
      \renewcommand{\qedsymbol}{$\lrcorner$}
\end{proof}
  
In what follows we show the existence of a hole in $G - HC$, which
contradicts the assumption that $HC$ is a hole cover and thus proves
this lemma.  Denote by $P_1 = (u_1 u_2 \dots u_p)$ and $P_2 = (v_q
v_{q+1} \dots v_0 v_1)$.  By definition $u \not\in P_1$; to show $v
\not\in P_2$ it suffices to rule out the possibility that $v\in
\{v_1,v_q\}$, as by the numbering of $H_2$, $v$ is in
$\{v_1,v_1,\dots,v_q\}$.  According to
Corollary~\ref{lem:neighbors-of-neighbors-in-holes}, the two neighbors
of $v$ in $H_2$ are adjacent to either $u_1$ or $u_p$; however,
Claim~\ref{claim-0} implies that $v_0$ and $v_{q+1}$ are inner
vertices of $P_2$ and hence are not adjacent to $u_1$ or $u_p$.  We
now argue that each inner vertex $v_i$ of $P_2$ is not adjacent to
$P_1$ (see the thick edges in Figure~\ref{fig:hole-cover-is-clique}).
Suppose to the contrary, $v_i$ is adjacent to $P_1$.  Noting that $v_i
\not\sim u_1$, $v_i \not\sim u_p$, and $u_1\ne u_{p+1}$,
Proposition~\ref{lem:five-path}(3) applies, and we can conclude
$v_i\sim v$, which is impossible.  (It is immaterial whether $v_i\in
H_1$ or not.)  Now we construct the hole in $G - HC$ as follows.
Claim~\ref{claim-0} implies that the length of $P_2$ is at least $2$.
If $u_1 \sim v_q$, then $(u_1 P_2 u_1)$ is such a hole.  Otherwise by
assumption we have $u_p \sim v_q$.  Let $\ell_1 = \max \{i : u_i \sim
v_1 \text{ and }0 \leq i \leq p\}$, and $\ell_2 = \min \{i : u_i \sim
v_q \text{ and }\ell_1 \leq i \leq p\}$.  Then $(u_{\ell_1} u_{\ell_2}
P_2 u_{\ell_1})$ will be such a hole (see the solid hole in
Figure~\ref{fig:hole-cover-is-clique}).
\end{proof}

\begin{lemma}\label{lem:non-neighbor-of-hole-cover}
  For any minimal hole cover $HC$ of $G_0$ and any shortest hole $H$,
  there is a vertex $v\in H$ such that $N_{G_0}[v] \not\sim HC$.
\end{lemma}
\begin{proof}
  We show this by construction.  By
  Corollaries~\ref{lem:no-more-than-3} and \ref{lem:hole-not-common},
  each vertex in $G_0$ has at most $3$ neighbors in $H$.  By
  Lemma~\ref{lem:hole-cover-is-clique}, $HC$ is a clique and hence $|H
  \cap HC| \le 2$.  We number the vertices of $H$ in a way that $h_0
  \in HC$ and $h_{1} \not\in HC$, and claim that $v = h_{5}$ is the
  asserted vertex.  Suppose to the contrary, $N_{G_0}[h_5]$ and $HC$
  are adjacent, then there is an $h_0$-$h_5$ path $P$ of length at
  most $3$ and all its inner vertices belong to $G_0$.  The case $P =
  (h_0 v h_5)$ is impossible, as by
  Proposition~\ref{lem:neighbors-are-consecutive} and
  Corollary~\ref{lem:no-more-than-3}, $v$ is adjacent to at most $3$
  consecutive vertices in $H$.  Now we may assume $P = (h_0 v_1 v_2
  h_5)$, and examine the neighbors of $v_1$ and $v_2$ in $H$.  By
  Corollary~\ref{lem:no-more-than-3}, we have $\mathtt{end}_H(v_1)
  \leq 2$ and $\mathtt{begin}_H(v_2) \geq 3$. This means that there is
  a hole $(v_1 h_i h_{i+1} \dots h_{j} v_2 v_1)$, where $i =
  \mathtt{end}_H(v_1)$ and $j = \mathtt{begin}_H(v_2)$, of length at
  least $4$ and at most $8$.
\end{proof}

We now relate minimal hole covers of $G_0$ to minimal separators in some
interval subgraphs.  In one direction of the proof, we need the
following claim.  Observe that in an interval representation of a
connected interval graph, the union of all the intervals also forms an
interval. Similarly, if there is a point $p$ in the real line such
that there are intervals not containing $p$ both to the left and to
the right of $p$, then the set of intervals containing $p$ is a clique
separator.

\begin{proposition}
  \label{lem:non-simplicial-neighborhood-is-separator}
  Let $v$ be a vertex in an interval graph $G$.  If $v$ is not
  adjacent to any simplicial vertex, then $N[v]$ is a separator of $G$.
\end{proposition}
\begin{proof}
  We consider an interval representation of $G$.  Without loss of
  generality, we assume that no two intervals have the same ends.
  Denote by $x$ the interval with the smallest right end, and by $y$
  the interval with the largest left end.  It is easy to see that $x$
  and $y$ are simplicial.  If $x\sim y$, then the graph is a complete
  graph (every interval contains the interval between the left end of
  $y$ and the right end of $x$); thus every vertex is adjacent to a
  simplicial vertex, and the assertion is vacuously true.  Therefore,
  we can assume $x \not\sim y$, and let $p$ be an arbitrary point in
  interval $v$.  By assumption $v$ is not adjacent to $x$ or $y$,
  which means that $x$ is to the left of $p$ and $y$ is to the right
  of $p$.  As every interval that contains $p$ is in $N[v]$, in the
  subgraph $G-N[v]$ that contains $x$ and $y$, no interval contains
  $p$; hence $x$ and $y$ are disconnected.  In other words, $N[v]$ is
  an $x$-$y$ separator.
\end{proof}

According to Lemma~\ref{lem:non-neighbor-of-hole-cover}, every minimal
hole cover satisfies the condition in the following lemma; hence the
lemma applies to all of them.  Note that $G_0 - N_{G_0}[v]$ is the
same as $G_0 - N[v]$.  

\begin{lemma}
  \label{lem:hole-cover-is-separator}
  Let $v$ be a vertex in a shortest hole $H$ of $G_0$, and $X$ induce
  a clique nonadjacent to $N_{G_0}[v]$.
  Set $X$ forms a minimal hole cover of $G_0$ if and only if $X$ is a
  minimal separator of $G_0 - N[v]$.
\end{lemma}
\begin{proof}
  It suffices to show that $X$ is a hole cover of $G_0$ if and only if
  it is a separator of $G_0 - N[v]$.

  $\Rightarrow$ Clearly, each component of $G_0 - N[v]$ contains a
  neighbor of $N[v]$. As $X$ is not adjacent to $N[v]$, the set $X$
  cannot fully contain a component of $G_0-N[v]$, which implies that
  the number of components of $G_0 - N[v] - X$ is no less than that of
  $G_0 - N[v]$.  Therefore, if $G_0 - N[v]$ is not connected, then
  neither is $G_0 - N[v] - X$, and $X$ makes a trivial separator for
  $G_0 - N[v]$.  In the following argument of this direction we may
  assume $G_0 - N[v]$ is connected, and it suffices to show that $G_0
  - N[v] - X$ is not connected.  By
  Proposition~\ref{lem:interval-subgraph}, $G_0-X$ is an interval
  subgraph; as $G_0$ itself contains no simplicial vertex, any vertex
  $x$ that is simplicial in $G_0-X$ must be a neighbor of $X$:
  otherwise $N_{G_0 - X}(x) = N_{G_0}(x)$ and cannot be a clique.  As
  $N[v]$ is not adjacent to $X$ by assumption, $v$ is not adjacent to
  any simplicial vertex of the interval graph $G_0-X$.  Therefore,
  according to
  Proposition~\ref{lem:non-simplicial-neighborhood-is-separator}, the
  removal of $N[v]$ disconnects $G_0-X$.  This finishes the proof of
  the ``only if'' direction.

    $\Leftarrow$ Let us start from a close scrutiny of $G_0 -
    N[v]$.  According to
    Proposition~\ref{lem:neighbors-are-consecutive}, the removal of
    $N[v]$ transforms each hole into a path of length at least $7$; in
    particular, let $P$ be the path induced by $H\setminus N_H[v]$. In
    the argument to follow, we show that ends of each such path are
    connected to the ends of $P$ respectively; the further removal of $X$
    separates each path into at most two sub-paths; hence if there is
    a hole disjoint from $X$, then the path left by it is able to
    connect every sub-path and thereby every vertex, which is
    impossible.

    We number the vertices of $H$ such that $v=v_0$.  Then $N_H[v] =
    \{v_{-1},v_0,v_1\}$ and the ends of $P$ are $v_{-2}$ and $v_2$.
    Let $H'$ be another hole of $G_0$, and $P'$ be the path induced by
    $H'\setminus N_{H'}[v]$.  We may number the vertices of $H'$ such
    that $N_{H'}[v] = \{h_1,\dots,h_p\}$.  As a result, the ends of
    $P'$ are $h_0$ and $h_{p+1}$.
    \begin{claim}
      The ends $h_0$ and $h_{p+1}$ of $P'$ are adjacent to
      $\{v_1,v_2\}$ and $\{v_{-1},v_{-2}\}$, respectively.  
    \end{claim}
    \begin{proof}
      By Corollary~\ref{lem:no-more-than-3},
      $N_H[h_1]\subset\{v_{-2},v_{-1}, v_0,v_1,v_2\}$, and according
      to Proposition~\ref{lem:neighbors-of-neighbors}(1), $h_0$ is
      adjacent to either $\{v_1,v_2\}$ or $\{v_{-1},v_{-2}\}$; a
      symmetric argument works for $h_{p+1}$.  Since the length of $H$
      is at least $11$, the sets $\{v_1,v_2\}$ and $\{v_{-1},v_{-2}\}$
      are disjoint.  It remains to show that the ends of $P'$ cannot
      be adjacent to both $\{v_1,v_2\}$ or both $\{v_{-1},v_{-2}\}$.
      Suppose, for contradiction, both $h_0$ and $h_{p+1}$ are
      adjacent to $\{v_1,v_2\}$.  We consider three cases.

      Suppose first that both $h_0$ and $h_{p+1}$ are adjacent to
      $v_1$.  According to
      Corollary~\ref{lem:neighbors-of-neighbors-in-holes} (applied on
      the adjacent vertices $v_0,v_1$ and the hole $H'$), at least one
      end of $N_{H'}[v_1]$ is in $N_{H'}[v_0]$, i.e.,
      $\{h_1,\dots,h_p\}$.  As a result, $N_{H'}[v_1]$ must contain
      all of $\{h_{p+1}, h_{p+2}, \dots, h_0\}$, and then
      $N_{H'}[v_0]\cup N_{H'}[v_1] = H'$.  This contradicts
      Proposition~\ref{lem:neighbors-of-neighbors}(1) and is
      impossible.

      Suppose now that $h_0$ and $h_{p+1}$ are both adjacent to $v_2$.
      Note that $v_2\not\in H'$; otherwise, since $v_0\not\sim v_2$,
      we must have $v_2 = h_{-1}$, and then ($v_0 h_1 h_0 v_2 h_{p+1}
      h_p v_0$) is a $6$-hole, which is impossible.  We can apply
      Proposition~\ref{lem:five-path}(1) on $v_2$, $v_0$, and path
      $(h_{-1} h_0 h_1\dots h_p h_{p+1} h_{p+2})$ to conclude $v_0\sim
      v_2$, which is impossible.

      In the remaining cases, $h_0$ and $h_{p+1}$ are adjacent to
      $v_1$ and $v_2$, respectively.  Without loss of generality, we
      consider $h_0\sim v_1$ and $h_{p+1}\sim v_2$.  Clearly $h_p \ne
      v_2$ as they have different adjacencies to $v_0$; likewise,
      $h_{p+1} \ne v_1$ and $h_{p+1} \ne v_2$.  We exclude $h_p =
      v_1$: then $p=1$ by $h_0\sim v_1$, and $v_0\sim h_0$ by
      Corollary~\ref{lem:neighbors-of-neighbors-in-holes}, which
      contradicts the numbering of $H'$.  Then $h_p\sim v_2$, as
      otherwise there is a hole $(v_2 v_1 v_0 h_p h_{p+1} v_2)$ or
      $(v_2 v_1 h_p h_{p+1} v_2)$; likewise, $h_p\sim v_1$.  It
      follows that, by Corollary~\ref{lem:no-more-than-3},
      $h_p\not\sim v_{-1}$, and $p>1$.  On the other hand,
      $h_0\not\sim v_{-1}$, as otherwise there is a hole $(h_0 v_{-1}
      v_0 v_1 h_0)$.  We can apply Proposition~\ref{lem:five-path}(3)
      on $v_1$, $v_{-1}$, and path $(h_{-1} h_0 h_1\dots h_p v_{p+1})$
      to conclude $v_{-1}\sim v_1$, which is impossible.

      A symmetric argument applies to $\{v_{-1},v_{-2}\}$.  
      \renewcommand{\qedsymbol}{$\lrcorner$}
    \end{proof}
    We may assume without loss of generality,
    $h_{p+1}\sim\{v_1,v_2\}$, and then $h_{0}\sim\{v_{-1},v_{-2}\}$.
    Let $\ell_1$ be the smallest index such that $\ell_1>p$ and
    $h_{\ell_1}\in N[v_2]$; for its existence, observe that $\ell_1 =
    p+1$ if $h_{p+1}\in N[v_2]$, otherwise by
    Corollary~\ref{lem:neighbors-of-neighbors-in-holes},
    $h_{\mathtt{end}_{H'}(v_1)}$ must be in $N[v_2]$.  For $p+1\le
    i\le \ell_1$, it holds that $h_i\sim v_1$, and since $v_1\in
    N[v]$, by assumption ($X\not\sim N[v]$), we have $h_i\not\in X$.
    Symmetrically, we can find the largest index $\ell_2$ such that
    $\ell_2\le |H|$ and $h_{\ell_2}\in N[v_{-2}]$.  For $\ell_2\le
    i\le |H|$, it holds that $h_i\sim v_{-1}$ and $h_i\not\in X$.  On
    the other hand, as $X$ is a clique, it contains at most two
    vertices of $P'$, which are in $\{h_{\ell_1+1}, \dots,
    h_{\ell_2-1}\}$.  Therefore, the removal of $X$ either leaves $P'$
    intact (when $X$ is disjoint from $P'$), or separates $P'$ into
    two sub-paths.  In the later case, the two sub-paths, containing
    $h_{\ell_1}$ and $h_{\ell_2}$ respectively, are connected to $v_2$
    and $v_{-2}$ respectively.

    To prove the ``if'' direction, we need to show that $X$ intersects
    every hole.  Suppose, for contradiction, that $X$ is disjoint from
    some hole $H_1$.  Then the path $P_1$ induced by $H_1\setminus
    N_{H_1}[v]$ remains a path of $G_0 - N[v] - X$.  Since $P_1$ is
    adjacent to both $v_2$ and $v_{-2}$, we conclude that $v_2$ and
    $v_{-2}$ are connected in $G_0 - N[v] - X$.  We have seen that for
    each hole $H'$, the vertices left from $H'$ after the removal of
    $N[v]$ and $X$ are connected to at least one of $v_2$ and
    $v_{-2}$.  Therefore, the subgraph $G_0 - N[v] - X$ is connected,
    contradicting the assumption that $X$ is a separator of $G_0 -
    N[v]$.  This finishes the proof of the ``if'' direction.
  \end{proof}

We are now ready to prove Theorem~\ref{thm:alg-reduced-with-holes}.
We remark that the quadratic bound can be improved to linear with more
careful analysis.
\paragraph{Theorem 2.3 (restated).}
  Every reduced graph of $n$ vertices contains at most $n^2$ minimal
  hole covers, and they can be enumerated in ${O}(n^3)$ time.
\begin{proof}
  Let $G_0$ be induced by the union of the holes of $G$.  On the one
  hand, according to Lemmas~\ref{lem:non-neighbor-of-hole-cover} and
  \ref{lem:hole-cover-is-separator}, each minimal hole cover of $G$
  corresponds to a minimal separator of the interval subgraph $G_0 -
  N[v]$ for some vertex $v$ of a shortest hole $H$.  On the other
  hand, there are at most $n$ minimal separators in $G_0 - N[v]$ for
  each vertex $v\in H$, which implies a quadratic bound for the total
  number of minimal hole covers of $G$.  To enumerate them, we try
  every vertex $v\in H$ and enumerate all minimal separators of
  $G_0-N[v]$.
  \end{proof}

\section{Caterpillar decompositions}
\label{sec:caterp-decomp}
This section proves Theorem~\ref{thm:alg-reduced-and-chordal} by
providing the claimed algorithm for {\sc interval deletion} on nice
graphs.  Recall that a nice graph is chordal and contains no small AW,
and every shallow terminal in a nice graph is simplicial; nice graphs
are hereditary.  Our algorithm finds an AW satisfying a certain
minimality condition, from which we can construct a set of ten
vertices that intersects some minimum interval deletion set.  Hence it
branches on deleting one of these ten vertices.  The set of all shallow
terminals, denoted by $ST(G)$, can be found in polynomial time as
follows.  For each triple of vertices, we check whether or not they
forms the terminals for an AW.  If yes, then one of them is
necessarily shallow.  The following lemma ensures that all shallow
terminals can be found as such.

\begin{proposition}
\label{lem:shallow-is-shallow}
  In a nice graph, all AWs with the same set of terminals have the
  same shallow terminal.
\end{proposition}
\begin{proof}
  Of any AT $\{x,y,z\}$, there must be a vertex, say,
  $x$, such that the shortest $y$-$z$ path in $G - N[x]$ has length at
  least $4$, as otherwise there is an AW of size at most $9$, which
  contradicts the definition of nice graphs.  Therefore, neither $y$
  nor $z$ can be the shallow terminal in an AW with terminals
  $\{x,y,z\}$.
\end{proof}
It should be noted that this does not rule out the possibility of a
vertex being a base terminal of an AW and the shallow terminal of
another AW.  If this happens, these AWs necessarily have at least one
different terminal.  Recall that by
Theorem~\ref{thm:reduced-instances-ap}, every vertex in $ST(G)$ is
simplicial in $G$.  For each $\dag$- or $\ddag$-AW, its shallow
terminal is in $ST(G)$ by definition, its base terminals might or
might not be in $ST(G)$, and none of the non-terminal vertices can be
in $ST(G)$ (as they are not simplicial).  From Lemma~\ref{lem:shallow}
we can derive
\begin{proposition}\label{lem:min-at}
  Let $s$ be a shallow terminal in a nice graph.  There is an AW of
  which every base vertex
  is adjacent to all vertices of $N(s)\setminus ST(G)$.
\end{proposition}
\begin{proof}
  Let $W$ be an AW with shallow terminal $s$ and shortest possible
  base.  Applying Lemma~\ref{lem:shallow} on any vertex $x\in
  N(s)\setminus ST(G)$ and $W$, it cannot be in category ``partial''
  by the minimality of $W$. Vertex $x$ cannot be in category ``none''
  either, otherwise $x$ is a shallow terminal, contradicting $x\in
  N(s)\setminus ST(G)$. Thus every vertex in $N(s)\setminus ST(G)$ is
  in category ``full.''
\end{proof}

Now that the graph is chordal, it makes sense to discuss its clique
tree, which shall be the main structure of this section.  No
generality will be lost by assuming $G$ is connected.  Since no inner
vertex of a shortest path can be simplicial, the removal of simplicial
vertices will not disconnect a connected graph; hence $G-ST(G)$ is a
connected interval graph.  This observation suggests a clique tree of
$G$ with a very nice structure.  A \emph{caterpillar (tree)} is a tree
that consists of a central path and all other vertices are leaves
connected to it.
\begin{proposition}
  \label{lem:caterpillar}
  In polynomial time we can construct a clique tree ${\cal T}$ for a
  connected nice graph $G$ such that
  \begin{itemize}
  \item ${\cal T}$ is a caterpillar;
  \item every shallow terminal of $G$ appears only in one leaf node of
    ${\cal T}$; and
  \item every other vertex in $G$ appears in some nodes of the central
    path of ${\cal T}$.
  \end{itemize}
\end{proposition}
\begin{proof}
  Let us inspect every maximal clique $K$ of $G$.  If $K$ contains
  some shallow terminal $s$, then $K$ must be $N[s]$: Being a clique,
  $K \subseteq N[s]$; this containment cannot be proper as $K$ is
  maximal and $N[s]$ induces a clique.  Otherwise, $K \cap ST(G) =
  \emptyset$, then $K$ is also a maximal clique of $G - ST(G)$.  On
  the other hand, a maximal clique $K'$ of $G - ST(G)$ has to be a
  maximal clique of $G$ as well; otherwise it is a proper subset of
  some maximal clique $K$ of $G$ that must contain a shallow terminal
  $s$, hence $K=N[s]$ and $K'\subseteq N[s]\setminus ST(G)$, which,
  however, according to Proposition~\ref{lem:min-at}, cannot be
  maximal in $G - ST(G)$.  Therefore, a maximal clique of $G$ is
  either $N[s]$ for some vertex $s\in ST(G)$ or a maximal clique of
  $G-ST(G)$. We construct the claimed clique tree as follows.  First
  use Theorem~\ref{thm:linear-tree} to make a clique path ${\cal T'}$
  for the interval subgraph $G - ST(G)$, then for each $s\in ST(G)$,
  attach $N[s]$ as a leave to ${\cal T'}$ at some maximal clique that
  properly contains $N[s]\setminus ST(G)$ (arbitrarily pick from
  multiple choices).
\end{proof}

Within a caterpillar decomposition, we number the nodes in the central
path as $K_0,K_1,\dots$.  By Proposition~\ref{lem:caterpillar} and the
definition of clique trees, each vertex not in $ST(G)$ is contained in
some consecutive nodes of the central path.  For each vertex $v\not\in
ST(G)$, we denote by $\mathtt{first}(v)$ and $\mathtt{last}(v)$ the
smallest and, respectively, largest indices of nodes that contain $v$.
In any $\dag$- or $\ddag$-AW, every vertex of the base is
non-simplicial, hence belongs to the central path of the caterpillar
decomposition.  By assumption, $d = |B|\ge 3$ and $b_1 \not\sim b_d$;
as a result, the nodes that contain $b_1$ and $b_d$ are disjoint.
When numbering the vertices of the base, we follow the convention that
$\mathtt{last}(b_1) < \mathtt{first}(b_d)$, i.e., base $B$ goes ``from
left to right.''  Given a numbering of the base, the base terminals
$l$ and $r$ can be distinguished from each other based on their
adjacency with $b_1$ and $b_d$. Similarly, in the case of a
$\ddag$-AW, the centers $c_1$ and $c_2$ can be distinguished from each
other, as they have different adjacency relations with $l$ and
$r$.\footnote{Note that we are not relying on the relation between
  $\mathtt{first}(c_1)$ and $\mathtt{first}(c_2)$ or that between
  $\mathtt{last}(c_1)$ and $\mathtt{last}(c_2)$, and they will not
  matter in the proofs to follow.}

By observing the adjacencies and nonadjacencies between vertices of an
AW and their possible positions in an interval representation of $G -
ST(G)$, the following is straightforward and hence stated here without
proof.  In order to avoid pointless repetition, we are again using the
same generalized notation for both $\dag$- and $\ddag$-AW as
stipulated in Section~\ref{sec:reduct-rules-branch}.

\begin{proposition}\label{lem:at-indices}
  Let $(s:c_1,c_2:l,B,r)$ be a $\dag$- or $\ddag$-AW in a nice graph
  $G$.  In a caterpillar decomposition of $G$,
  \begin{align}\label{eq:big-picture}
    \big( \mathtt{first}(b_1) \leq& \mathtt{last}(l) < \!\!\big)\;
    \mathtt{first}(c_2), \mathtt{first}(b_2) \leq \mathtt{last}(b_1)
    < \dots&
    \nonumber\\
    \leq& \mathtt{first}(b_{i}) \leq \mathtt{last}(b_{i-1}) <
    \mathtt{first}(b_{i+1})
    \leq \mathtt{last}(b_{i}) < \mathtt{first}(b_{i+2})& \nonumber\\
    \leq \dots <& \mathtt{first}(b_d) \leq \mathtt{last}(b_{d-1}),
    \mathtt{last}(c_1) \;\big( \!\!< \mathtt{first}(r) \leq
    \mathtt{last}(b_d) \big),&
  \end{align}
where relations in parentheses only hold when $l \not\in ST(G)$ and $r
\not\in ST(G)$, respectively.
\end{proposition}

Nodes that contain non-terminal vertices of an AW appear consecutively
in the central path of ${\cal T}(G)$. We would like to identify a
minimum set of consecutive nodes whose union contains all non-terminal
vertices of the AW.

\begin{definition}
  Let ${\cal T}$ be a caterpillar decomposition of a nice graph $G$.
  We define $\gemini[{p,q}] = \bigcup_{p \leq i \leq q} K_i$ for a
  pair of indices $p \leq q$, and $\gemini(W) =
  \gemini[{\mathtt{last}(b_1),\mathtt{first}(b_d)}]$ for an AW $W$.
  Set $\gemini(W)$ will be referred to as the \emph{container} of $W$,
  and we say it is \emph{minimal} if there exists no AW $W'$ such that
  $\gemini(W') \subset \gemini(W)$.
\end{definition}

Let us observe that every base vertex of $W$ appears in $\gemini(W)$
and no shorter subsequence of nodes contains every base vertex.
Moreover, the following proposition shows that the centers also appear
in $\gemini(W)$ (recall that $\widehat N(B)$ is the set of common
neighbors of $B$ and every center is in $\widehat N(B)$).
\begin{proposition}
  $K_{\mathtt{last}(b_1)}\cap K_{\mathtt{first}(b_d)} = \widehat
  N(B)$.
\end{proposition}
\begin{proof}
  By definition, a vertex of the left side is in $K_i$ for every
  $\mathtt{last}(b_1)\le i\le \mathtt{first}(b_d)$, and thus belongs
  to $\widehat N(B)$.  On the other hand, if a vertex $v$ does not
  belong to the left side, then either $\mathtt{first}(v) >
  \mathtt{last}(b_1)$ or $\mathtt{last}(v) < \mathtt{last}(b_d)$,
  which implies $v\not\sim b_1$ or $v\not\sim b_d$ respectively.  In
  either case, we have $v\not\in\widehat N(B)$.
\end{proof}
In Section~\ref{sec:holes-in-modules}, we considered holes of the
shortest length and observed that a vertex sees either all or at most
three vertices in such a hole.  Here for an AW whose container is
minimal and base consists of the inner vertices of a shortest $l$-$r$
path specified below, we can observe an analogous statement about the
number of base vertices a vertex can see.

\begin{definition}
  Let $W=(s:c_1,c_2:l,B,r)$ be an AW in a nice graph such that
  $\gemini(W)$ is minimal.  We say $B$ is a \emph{short} base if $(l B
  r)$ is a shortest $l$-$r$ path in the subgraph induced by
  $\big(\gemini(W) \setminus \widehat N(B) \big) \cup \{l, r\}$.
\end{definition}

The following lemma shows that if the base is not short, then we can
get an AW with a shorter base. In particular, this implies that a
vertex of $\gemini(W)\setminus \widehat N(B)$ can see at most $3$
consecutive vertices of the base.  
\begin{lemma}
\label{lem:shortbase}
  Let $W=(s;c_1,c_2;l,B,r)$ be an AW such that $\gemini(W)$ is
  minimal. Then there is an $W'$ such that $\gemini(W')=\gemini(W)$
  and $W'$ has a short base.
\end{lemma}
\begin{proof}
  We show that if $(l P r)$ is a chordless $l$-$r$ path in the
  subgraph induced by $\big(\gemini(W) \setminus \widehat N(B)\big)
  \cup \{l, r\}$, then we can replace the base $B$ of $W$ by $P$ to
  obtain another AW $W_P = (s:c_1,c_2:l,P,r)$.  Clearly the center(s)
  of $W$ belong to $\widehat N(B)$, thereby adjacent to every other
  vertex in $\gemini(W)$, and hence to $P$.  It is also easy to verify
  that no vertex in $\gemini(W)\setminus \widehat N(B)$ is adjacent to
  $s$: if such a vertex exists, then Lemma~\ref{lem:shallow}
  classifies it as ``partial'' with respect to $W$, hence there is
  another AW $W'$ such that $B'\subset B$ and $\gemini(W') \subset
  \gemini(W)$, which contradicts the minimality of
  $\gemini(W)$. Therefore, $W_p$ is indeed an AW.  Letting $b'_1$ and
  $b'_{d'}$ be the first and, respectively, last vertices of $P$, the
  selection of $P$ implies $\mathtt{last}(b'_1)\ge\mathtt{last}(b_1)$
  and $\mathtt{first}(b'_{d'})\le\mathtt{first}(b_{d})$, hence
  $\gemini(W_P) \subseteq \gemini(W)$; as the latter is already
  minimal, they must be equal. Therefore, if the base of $W$ is not
  short, then we can find another AW with the same container and
  shorter base.  Applying this argument repeatedly will eventually
  procure an AW with the same container and having a short base.
\end{proof}

With all pertinent definitions and observations, we are now ready to
present the main lemma of this section which justifies our branching
rule.  Without an upper bound on the number of vertices in an AW---in
particular, the length of its base can be arbitrarily long---trying
each vertex in it cannot be done in FPT time.  Thus we have to avoid
most but a (small) constant number of base vertices---those are close
to the base terminals---to procure the claimed algorithm.  To further
decrease the number of vertices we need to consider, observing that
the central path of the caterpillar decomposition has a linear
structure, we start from the \emph{leftmost} minimal container.  By
definition, minimal containers cannot properly contain each other, and
thus the one with smallest begin-index also has the smallest
end-index.  In particular, the leftmost minimal container is unique,
though it might be observed by more than one AWs, and can be
identified in polynomial time.  With this additional condition, if
another AW intersects $\gemini(W)$, it has to come ``from the right.''

Let $W$ be an AW of leftmost minimal container and having a short
base. We claim that there is a minimum interval deletion set that
breaks $W$ in a canonical way: it contains either one of a constant
number of specific vertices of $W$, or a specific minimum separator
(details are given below) breaking the base of $W$.  Therefore, by
branching into ten directions, we can guess one vertex of this interval
deletion set.\footnote{A slightly weaker version of
  Lem~\ref{lem:chordal-to-interval} is given in the appendix.  The
  proof of Lem~\ref{lem:chordal-to-interval}, trying to minimizing the
  number of branching directions, has to consider many cases and is
  ponderous.  In contrast, the proof of the weaker version uses only
  the fact that a vertex that is not a common neighbor of $B$ sees at
  most three vertices in it; hence the underlying ideas are easier to
  understand.}

For each $\mathtt{last}(b_1)\le i< \mathtt{first}(b_{d-1})$, let us
define $S_i = K_i\cap K_{i+1}$ to be the {\em $i$th separator.}  Note
that $S_i$ contains $\widehat N(B)$ as a proper subset.
\begin{lemma}\label{lem:chordal-to-interval}
  Let ${\cal T}$ be a caterpillar decomposition of a nice graph $G$,
  and $W=(s:c_1,c_2:l,B,r)$ be an AW in $G$ such that
  \begin{itemize}
  \item $\mathtt{first}(b_d)$ is the smallest among all AWs; 
  \item $\gemini(W)$ is minimal; and
  \item $B$ is a short base.
  \end{itemize}
  Let $\ell$ be the minimum index such that $\mathtt{last}(b_1)\le
  \ell< \mathtt{first}(b_{d-2})$ and the cardinality of $S_{\ell}$ is
  minimum among $\{ S_i : \mathtt{last}(b_1) \leq i <
  \mathtt{first}(b_{d-2}) \}$.  There is a minimum interval deletion
  set to $G$ that either contains one of the 9 vertices
  \[
  V_B = \{s,c_1,c_2, l, b_1, b_{d-2}, b_{d-1}, b_{d}, r\},
  \]
  or the whole set $X = S_{\ell} \setminus N$, where $N = \widehat
  N(B)$.
\end{lemma}

\begin{proof}
  We prove by construction.  Let $Q$ be any minimum interval deletion
  set; we may assume $Q \cap V_B = \emptyset$, and $X \not\subseteq
  Q$, as otherwise $Q$ satisfies the asserted condition and we are
  finished.  We claim $Q' = (Q \setminus V_{I}) \cup X$, where $V_{I}
  = \gemini[{\mathtt{last}(b_2), \mathtt{first}(b_{d-3})}] \setminus
  N$, is the desired interval deletion set, which fully contains $X$
  in particular.

  As $G$ is chordal, all \mfiss\ in $G$ are AWs.  To show that $Q'$
  makes an interval deletion set to $G$, it suffices to argue that if
  there exists an AW $W'$ avoiding $Q'$ then we can also find an AW
  $W''$, not necessarily the same as $W'$, avoiding $Q$.  Suppose $W'
  = (s':c'_1,c'_2:l',B',r')$ is the AW in $G - Q'$. By the
  construction of $Q'$, this AW must intersect $V_{I} \setminus X$;
  let $u \in W' \cap (V_{I} \setminus X)$.  Clearly, $u$ can neither
  be $s'$, as $u\not\in ST(G)$, nor $r'$, as otherwise according to
  Proposition~\ref{lem:at-indices}, $\mathtt{first}(b'_{d'}) <
  \mathtt{first}(u) < \mathtt{first}(b_{d})$, contradicting the
  selection of $W$.  The following claim rules out the possibility
  that $u\in \{c'_1,c'_2\}$.

  \begin{claim}\label{claim-1}
    For each vertex $v \in \gemini[0,\mathtt{first}(b_{d-2})] \setminus
    N$, we have $\mathtt{last}(v) < \mathtt{first}(b_{d})$, and $v
    \not\sim ST(G)$.
  \end{claim}
  \begin{proof}
    By definition, if $v$ is adjacent to $B$, then $v\sim b_i$ for
    some $i\le d-3$.  If $v\sim b_d$, then $B$ is not a short base, as
    there would be a a shorter (not necessarily chordless) $l$-$r$
    path $(l,\dots,b_i,v,b_d,r)$. Therefore, $v\not\sim b_d$ and it
    follows that $\mathtt{last}(v) < \mathtt{first}(b_{d})$.  Suppose
    to the contrary of the second assertion, $v$ is adjacent to the
    shallow terminal $x$ of some AW $W_1$.  We apply
    Lemma~\ref{lem:shallow}(2) on $v$ and $W_1$.  As $v\not\in ST(G)$,
    it has to be in categories ``full'' or ``partial.''  In either
    case, there exists an AW whose base is fully contained in
    $\gemini[{\mathtt{first}(v), \mathtt{last}(v)}]$, contradicting
    the selection of $W$.  
    \renewcommand{\qedsymbol}{$\lrcorner$}
  \end{proof}

  Therefore, either $u=l'$ or $u\in B'$.  Now we focus on the
  chordless path $l'B'r'$, which we shall refer to by $P'$, and how it
  reaches $u$ when going from $r'$ to $l'$.  Recall that every vertex of $B'$ appears in the
  central path of the caterpillar decomposition.
  Figure~\ref{fig:claim-2} depicts non-terminal vertices of $W$ in an
  interval representation of the interval subgraph $G-ST(G)$, where
  base terminals $l$ and $r$ are illustrated with dashed lines as they
  might belong to $ST(G)$.  The main observation here is: for any
  vertex $u$ in $V_I$, if another vertex $z\in N(u)\setminus N$ (the
  thick segment) reaches outside of $\gemini(W)$, then $z\sim b_d$,
  and $u$ and $z$ will make a short cut between $b_{d-4}$ and $b_d$,
  which is impossible.

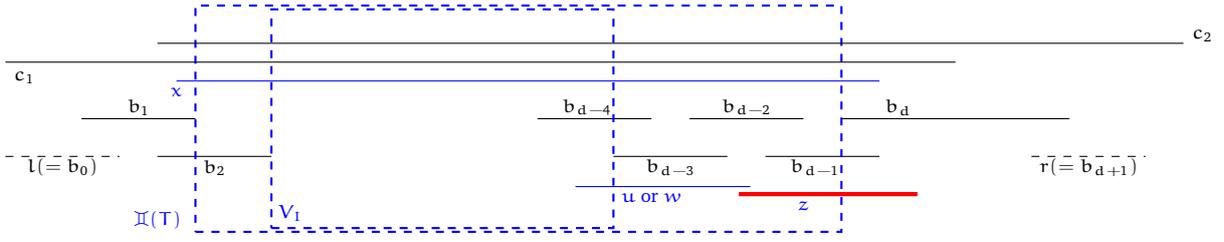
\begin{figure*}[t]
  \centering
\begin{tikzpicture}[scale=.5]
\scriptsize

\draw[dashed] (0,2) -- (3,2);  
\draw  (2,3) -- (5,3);  
\draw  (4,2) -- (7,2);  
\draw  (14,3) -- (17,3);  
\draw  (16,2) -- (19,2);  
\draw  (18,3) -- (21,3);  
\draw  (20,2) -- (23,2);  
\draw  (22,3) -- (28,3);  
\draw[dashed]  (27,2) -- (30,2);  

\node at (1.5,1.7) {$l (=b_{0})$};
\node at (3.5,3.3) {$b_{1}$};
\node at (5.5,1.7) {$b_{2}$};
\node at (17.5,1.7) {$b_{d-3}$};
\node at (21.3,1.7) {$b_{d-1}$};
\node at (28.5,1.7) {$r (=b_{d+1})$};
\node at (15.3,3.3) {$b_{d-4}$};
\node at (19.5,3.3) {$b_{d-2}$};
\node at (23.5,3.3) {$b_d$};

\draw[blue,thick,dashed] (7,0.1) -- (16,0.1) -- (16,5.9) -- (7,5.9) -- (7,0.1);  
\draw[blue,thick,dashed] (5,0) -- (22,0) -- (22,6) -- (5,6) -- (5,0);  
\node[blue] at (4,.3) {$\gemini(T)$};
\node[blue] at (7.5,.5) {$V_I$};

\draw  (0,4.5) -- (25,4.5);  
\draw  (4,5) -- (31,5);  
\node at (.5,4.1) {$c_{1}$};
\node at (31.5,5.2) {$c_{2}$};

\draw[blue] (4.5,4) -- (23,4);  
\draw[blue] (15,1.2) -- (19.6,1.2);  
\draw[ultra thick, red] (19.3,1) -- (24,1);  

\node[blue] at (4.5,3.7) {$x$};
\node[blue] at (17,.9) {$u$ or $w$};
\node[blue] at (21,.7) {$z$};

\end{tikzpicture}
  \caption{Interval representation of non-terminal vertices of a
    leftmost minimal AW.}
  \label{fig:claim-2}
\end{figure*}

  \begin{claim}
    $B'\cap N = \emptyset$.
  \end{claim}
  \begin{proof} 
    Suppose the contrary and let $x$ be a vertex in $B'\cap N$ (see
    Figure~\ref{fig:claim-2}).  Then $s\sim x$ follows from
    Lemma~\ref{lem:common-neighbor-of-base}.  We claim that every
    neighbor $z$ of $u$ is adjacent to $x$. Note that $z\not\sim x \in
    N$ implies either $\mathtt{last}(z) < \mathtt{last}(b_1)$ or
    $\mathtt{first}(z) > \mathtt{first}(b_d)$ and the latter is ruled
    out by the definition of $u$ and Claim~\ref{claim-1}. Let $B_1$ be
    the subset of the inner vertices of the $z$-$r$ path
    $(z,u,\dots,b_d,r)$. Then one of the following AW contradicts the
    minimality of the choice of $W$: $(s:c_1,x:z,B_1,r)$ (when
    $c_1\sim z$ and $x\sim r$), $(s:c_1:z,B_1,r)$ (when $c_1\not\sim
    z$), or $(s:x:z,B_1,r)$.

    As $x$ and $u$ are both in the chordless path $P'$ and $x$ is
    adjacent to every neighbor of $u$, vertex $u$ has to be one end of
    $P'$.  More specifically, $u=l'$ and $x = b'_1$.  A further
    consequence is that $u$ is the only vertex in $W'\cap V_I$: the
    argument above applies to any vertex $u'\in W'\cap V_I$, and thus
    $u'=l'=u$.  Now we show, for any vertex $w$ in $X \setminus Q$,
    which is nonempty by assumption, it has the same neighbors as $u$
    in $W'$, and hence $(s':c'_1,c'_2:w,B',r')$ is an AW in $G - Q$,
    contradicting the assumption that $Q$ is an interval deletion set
    to $G$.

    First, if a vertex is in $N$, then it is adjacent to both $w$ and
    $u$. Vertex $b'_1$ ($=x$) is in $N$. We claim that $c'_1$ is also
    in $N$ when $W'$ is a $\ddag$-AW. Otherwise, observe that
    $\mathtt{first}(c'_1) \le \mathtt{last}(u) < \mathtt{first}(b_d)$
    as $c'_1\sim u$ and $u$ satisfies the condition of
    Claim~\ref{claim-1}.  Let $B_1$ be the path $(b_1,\dots, b_i, u,
    c'_1)$, where $b_i$ is the first base vertex of $W$ adjacent to
    $u$. Now one of the following AW contradicts the minimality of
    $W$: $(s: x,c_2: l, B_1, b'_{d'})$ (when $x\sim l$ $c_2\sim
    b'_{d'}$), $(s: x: l, B_1, b'_{d'})$ (when $x\not\sim l$), or $(s:
    c_2: l, B_1, b'_{d'})$ (when $c_2\not\sim b'_{d'}$).

    Second, $c'_2\not\sim u$ ($=l'$) implies $c'_2\not\in N$. We claim
    that $c'_2\not\sim w$. Otherwise, observe that
    $\mathtt{first}(c'_2) \le \mathtt{last}(w) < \mathtt{first}(b_d)$
    as $c'_2\sim w$ and $w$ satisfies the condition of
    Claim~\ref{claim-1}.  Let $B_1$ be the path $(b_1,\dots, b_i, w,
    c'_2)$, where $b_i$ is the first base vertex of $W$ adjacent to
    $w$. Now one of the following AWs contradicts the minimality of
    $W$: $(s: x,c_2: l, B_1, b'_{d'})$ (when $x\sim l$ and $c_2\sim
    b'_{d'}$), $(s: x: l, B_1, b'_{d'})$ (when $x\not\sim l$), or $(s:
    c_2: l, B_1, b'_{d'})$ (when $c_2\not\sim b'_{d'}$).

    Finally, we claim that $w\not\sim b'_2$. Otherwise, observe that
    $\mathtt{first}(b'_2) \le \mathtt{last}(w) < \mathtt{first}(b_d) <
    \mathtt{first}(b'_3)$ as $b'_2\sim w$, $w$ satisfies the condition
    of Claim~\ref{claim-1}, and $b'_3\not\sim x$ ($=b'_1\in N$). Let
    $B_1$ be the path $(b_1,\dots, w, b'_2)$. Now one of the following
    AWs contradicts the minimality of $W$: $(s:x,c_2:l,B_1, b'_3)$
    (when $x\sim l$ and $c_2\sim b'_3$), $(s:c_2:l,B_1, b'_3)$ (when
    $c_2\not\sim b'_3$), or $(s:x:l,B_1, b'_3)$ (when $x\not\sim l$).
    Moreover, from $\mathtt{last}(w) < \mathtt{first}(b'_2) <
    \mathtt{first}(b'_i)$, it can be easily inferred that $w\not\sim
    b'_i$ for any $3\le i\le d'+1$.
    \renewcommand{\qedsymbol}{$\lrcorner$}    
  \end{proof}

  Now that $P'$ reaches $u$ not through $N$, next we show that the
  center $c'_2$ has to be in $N$ as it is adjacent to all base
  vertices $b_i$ of $W$ for $d-3\le i\le d$.

  \begin{claim}\label{claim-3}
    $c'_2 \in N$.
  \end{claim}
  \begin{proof}
    Suppose to the contrary, $c'_2 \not\in N$, then $c'_2$ cannot be
    adjacent to $b_1$.  From $c'_2\sim b'_1$ and $c'_2\not\sim b_1$ we
    can derive $\mathtt{last}(b'_1) \ge \mathtt{first}(c'_2) >
    \mathtt{last}(b_1)$.  On the other hand, as $\gemini(W)$ is
    minimal, it does not properly contain $\gemini(W')$, which implies
    $\mathtt{first}(b'_d) > \mathtt{first}(b_d)$.  Therefore,
    $\mathtt{last}(c'_2) \ge \mathtt{first}(b'_d) >
    \mathtt{first}(b_d)$.  (See Figure~\ref{fig:claim-2}.)  By the
    selection of $W$ and $B$, if a vertex $z$ is adjacent to both $u$
    and $b_d$, then $z\in N$, as otherwise there exists a path
    $(l,b_1,\dots,b_p,u,z,b_d,r)$, where $p\le {d-4}$, shorter than $l
    B r$.  In particular, the vertex next to $u$ in the path $l' B'
    r'$ is not adjacent to $b_d$; letting $u=b'_i$ where $i<d'$, it
    means $last(b'_{i+1}) <first(b_d)$.  From $c'_2\sim b'_{i+1}$ we
    can conclude $\mathtt{first}(c'_2) < \mathtt{first}(b_d)$, and
    then $c'_2$ and $b_d$ are adjacent, which further implies that $u$
    is not adjacent to $c'_2$.  In other words, $u$ has to be $l'$.
    Let $p$ be the index such that $b_p\not\sim c'_2$ and $b_{p+1}\sim
    c'_2$, which exists by assumption.  We now show $b_{d-2}\not\sim
    c'_2$, and $p\ge d-2$, by contradiction.
    \begin{inparaitem}
    \item If $s'$ is adjacent to every vertex in $N$, then $(s':
      c_1,c'_2: b_{p}, b_{p+1} \dots b_d, r)$ or $(s': c'_2: b_{p},
      b_{p+1} \dots b_d, r)$ would be an AW that contradicts the
      selection of $W$.
    \item If $s'$ is not adjacent to $x\in N$, then $(s: x,c_2: l,
      b_{1} \dots b_{p+1}, c'_2, s')$, $(s: x: l, b_{1} \dots b_{p+1},
      c'_2, s')$ or $(s: c_2: l, b_{1} \dots b_{p+1}, c'_2, s')$ would
      be an AW that contradicts the selection of $W$: noting that
      $\mathtt{first}(c'_2)\le \mathtt{last}(b_{d-2}) <
      \mathtt{first}(b_{d}) $ as $c'_2\sim b_{d-2}$.
    \end{inparaitem}
    However, $(s':c'_1,c'_2:b_p,b'_q,\dots,b'_{d'}, r')$,
    $(s':c'_1:b_p,b'_q,\dots,b'_{d'}, r')$, or
    $(s':c'_2:b_p,b'_q,\dots,b'_{d'}, r')$, where $q$ is the largest
    index such that $b'_q\sim b_p$, will be an AW in $G-Q$, which is
    impossible as $Q$ is an interval deletion set to $G$.
    \renewcommand{\qedsymbol}{$\lrcorner$}
  \end{proof}

  An immediate consequence of Claim~\ref{claim-3} is $c'_2\sim u$,
  hence $u\in B'$.  By Proposition~\ref{lem:at-indices},
  $\mathtt{first}(b'_1) \le \mathtt{first}(u) <
  \mathtt{first}(b_{d-2})$.  Then from Claim~\ref{claim-1} and the
  fact $l' \sim b'_1$, it can be inferred that $l' \not\in ST(G)$.
  Now $\mathtt{last}(l')$ is defined, and $\mathtt{last}(l') \le
  \mathtt{first}(c'_2) < \mathtt{last}(b_1)$; the selection of $W$
  implies $\mathtt{first}(b'_{d'}) \ge \mathtt{first}(b_{d})$.
  Therefore, the $l'$-$b'_{d'}$ path $l'B'$ has to go through $X$, and
  we end with a contradiction.

  This verifies that $Q'$ is an interval deletion set to $G$, and it
  remains to show that $Q'$ is minimum, from which the lemma follows.

  \begin{claim}
    $|Q'| \leq |Q|$.
  \end{claim}
  \begin{proof}
    It will suffice to show that $Q \cap V_I$ makes a $b_1$-$b_{d-2}$
    separator in $G - N$, and then the claim ensues as
    $$|Q'| = |Q\setminus V_I| + |X| \leq |Q\setminus V_I| + |Q\cap V_I| =
    |Q|.$$ Suppose to the contrary, there is a chordless
    $b_1$-$b_{d-2}$ path $P$.  We can extend $P$ into an $l$-$r$ path
    $P^+ = (l P b_{d-1} b_{d} r)$, which is disjoint from $Q$ and $N$.
    Within $P^+$ there is a chordless $l$-$r$ path $(l B_1 r)$.  By
    assumption, $\{s,c_1,c_2\}\cap Q=\emptyset$; every vertex in $B_1$
    satisfies the condition of Claim~\ref{new-claim-1}, and hence
    nonadjacent to $s$.  Moreover, $c_1,c_2\in N$, and therefore both
    $c_1$ and $c_2$ are adjacent to every vertex of $B_1$.  Thus,
    $(s:c_1,c_2:l,B_1,r)$ is an AW in $G-Q$, which is impossible.
    \renewcommand{\qedsymbol}{$\lrcorner$}
  \end{proof}
\end{proof}

To complete the proof of Theorem~\ref{thm:alg-reduced-and-chordal}, we
need one last piece of the jigsaw, i.e., to find the AW required by
Lemma~\ref{lem:chordal-to-interval}.

\paragraph{Theorem 2.4 (restated).}
  There is a $10^k\cdot n^{O(1)}$ time algorithm for \textsc{interval
    deletion} on nice graphs.
\begin{proof}
  Based on Lemma~\ref{lem:chordal-to-interval}, it suffices to show how
  to find such an AW, and then the standard branching will deliver the
  claimed algorithm.  For any triple of vertices $\{x,y,z\}$ and pair
  of indices $\{p,q\}$ for the nodes in the central path of the
  caterpillar decomposition, we can check whether or not there is an
  AW $W$ whose terminals are $\{x,y,z\}$ and non-terminal vertices are
  fully contained in $\gemini[p,q]$.  Therefore, in ${O}(n^6)$ time we
  are able to find the correct terminals and indices, from which the
  short base $B$ can also be easily constructed.  This finishes the
  construction of the AW required by
  Lemma~\ref{lem:chordal-to-interval}.
\end{proof}

\section{Concluding remarks}
\label{sec:remark}
We have classified \textsc{interval deletion} to be FPT by presenting
a $c^k \cdot n^{{O}(1)}$-time algorithm with $c=10$.  The constant $c$
might be improvable, and let us have a brief discussion on how to
achieve this.  The current constant $10$ comes from
Reduction~\ref{rule:small-forbidden-subgraph} and
Theorem~\ref{thm:alg-reduced-and-chordal}.  The constant in
Reduction~\ref{rule:small-forbidden-subgraph} is not tight, and it can
be replaced by $8$.  We choose the current number for the convenience
for later argument; for example, if we do not break AWs of size $9$ in
preprocessing, then we have to use a far more complicated proof for
Proposition~\ref{lem:shallow-is-shallow}.  In other words, the real
dominating step is to break ATs in nice graphs, where
we need to branch into $10$ cases.  As a nice graph exhibits a linear
structure, it might help to apply dynamic programming here.  To
further lower the constant $c$, we need to break small \fiss\ in a
better way than the brute-force in our algorithm.  So a natural
question is: Can it be $c=2$?

It is known that \textsc{chordal completion} can be solved in
polynomial time if the input graph is a circular-arc graph
\cite{kloks-98-fill-in-circular-arc} while \textsc{interval
  completion} remains NP-hard on chordal graphs
\cite{peng-06-interval-completion-on-choral-graphs}.  It would be
interesting to inquire the complexity of \textsc{interval deletion} on
chordal graphs and other graph classes.  At least, can it be solved in
polynomial time if the input graph is nice, which, if positively
answered, would suggest that all the troubles are small forbidden
subgraphs.  We leave open the parameterized complexity of
\textsc{interval edge deletion}, which instead asks for a set of $k$
edges whose removal makes an interval graph
\cite{goldberg-95-interval-edge-deletion,bodlaender-95-fpt-computational-biology}.
To adapt our approach to this problem, one needs a reasonable bound
for the number of {edge hole covers} for congenial holes.

As having been explored in \cite{narayanaswamy-13-d-cos-r}, we would
also like to ask which other problems can be formulated as or reduced
to \textsc{interval deletion} and then solved with our algorithm.
Both practical and theoretical consequences are worth further
investigation.

\appendix
\subsection*{Appendix. A simpler and weaker version of
  Lemma~\ref{lem:chordal-to-interval}}
\label{app:improve}
\paragraph{Lemma a.}
  Let ${\cal T}$ be a caterpillar decomposition of a nice graph $G$,
  and $W=(s:c_1,c_2:l,B,r)$ be an AW in $G$ such that
  \begin{itemize}
  \item $\mathtt{first}(b_d)$ is the smallest among all AWs; 
  \item $\gemini(W)$ is minimal; and
  \item $B$ is a short base.
  \end{itemize}
  Let $\ell$ be the minimum index such that $\mathtt{last}(b_2)\le
  \ell< \mathtt{first}(b_{d-5})$ and the cardinality of $S_{\ell}$ is
  minimum among $\{ S_i : \mathtt{last}(b_2) \leq i <
  \mathtt{first}(b_{d-5}) \}$.  There is a minimum interval deletion
  set to $G$ that either contains one of the 13 vertices
  \[
  V_B = \{s,c_1,c_2, l, b_1, b_2, b_{d-5}, b_{d-4}, b_{d-3}, b_{d-2},
  b_{d-1}, b_{d}, r\},
  \]
  or the whole set $X = S_{\ell} \setminus N$, where $N = \widehat
  N(B)$.
\begin{proof}
  We prove by construction.  Let $Q$ be any minimum interval deletion
  set; we may assume $Q \cap V_B = \emptyset$, and $X \not\subseteq
  Q$, as otherwise $Q$ satisfies the asserted condition and we are
  finished.  We claim $Q' = (Q \setminus V_{I}) \cup X$, where $V_I =
  \gemini[{\mathtt{last}(b_3), \mathtt{first}(b_{d-6})}] \setminus N$,
  is the desired interval deletion set, which fully contains $X$ in
  particular.  By definition of $V_I$, any vertex $z\in V_I$ is
  adjacent to some vertex $b_i$ for $4\le i\le d-7$, then as $B$ is
  short and $z\not\in N$, we have
  \begin{equation}
    \label{eq:u}
    \mathtt{first}(b_{2}) \le \mathtt{last}(b_1) < \mathtt{first}(z)\le
    \mathtt{last}(z)  < \mathtt{first}(b_{d-4}) \le \mathtt{last}(b_{d-5}).
  \end{equation}

  As $G$ is chordal, all \mfiss\ in $G$ are AWs.  To show that $Q'$
  makes an interval deletion set to $G$, it suffices to argue that if
  there exists an AW $W'$ avoiding $Q'$ then we can also find an AW,
  not necessarily the same as $W'$, avoiding $Q$.  Suppose $W' =
  (s':c'_1,c'_2:l',B',r')$ is an AW in $G - Q'$. By the construction
  of $Q'$, this AW must intersect $V_{I} \setminus X$; let $u \in W'
  \cap (V_{I} \setminus X)$.  Clearly, $u$ can neither be $s'$, as
  $u\not\in ST(G)$, nor $r'$, as otherwise according to
  Proposition~\ref{lem:at-indices}, $\mathtt{first}(b'_{d'}) <
  \mathtt{first}(u) < \mathtt{first}(b_{d})$, contradicting the
  selection of $W$.  The following claim further rules out the
  possibility that $u\in \{c'_1,c'_2\}$.

  \begin{claim}\label{new-claim-1}
    For each vertex $v \in \gemini[0,\mathtt{first}(b_{d-2})] \setminus
    N$, we have $\mathtt{last}(v) < \mathtt{first}(b_{d})$, and $v
    \not\sim ST(G)$.
  \end{claim}
  \begin{proof}
    By definition, if $v$ is adjacent to $B$, then $v\sim b_i$ for
    some $i\le d-3$.  If $v\sim b_d$, then $B$ is not a short base, as
    there would be a a shorter (not necessarily chordless) $l$-$r$
    path $(l,\dots,b_i,v,b_d,r)$. Therefore, $v\not\sim b_d$ and it
    follows that $\mathtt{last}(v) < \mathtt{first}(b_{d})$.  Suppose
    to the contrary of the second assertion, $v$ is adjacent to the
    shallow terminal $x$ of some AW $W_1$.  We apply
    Lemma~\ref{lem:shallow}(2) on $v$ and $W_1$.  As $v\not\in ST(G)$,
    it has to be in categories ``full'' or ``partial.''  In either
    case, there exists an AW whose base is fully contained in
    $\gemini[{\mathtt{first}(v), \mathtt{last}(v)}]$, contradicting
    the selection of $W$.  \renewcommand{\qedsymbol}{$\lrcorner$}
  \end{proof}

  Therefore, either $u=l'$ or $u\in B'$.  Now we focus on the
  chordless path $l'B'r'$, which we shall refer to by $P'$, and how it
  reaches $u$ when going from $r'$ to $l'$.  Recall that every vertex
  of $B'$ appears in the central path of the caterpillar
  decomposition.  

  \begin{claim}\label{new-claim-2}
    $B'\cap N = \emptyset$.
  \end{claim}
  \begin{proof} 
    Suppose the contrary and let $x$ be a vertex in $B'\cap N$.  By
    definition of $N$ and (\ref{eq:u}), we have $\mathtt{first}(x) <
    \mathtt{first}(u)\le \mathtt{last}(u)< \mathtt{last}(x)$.  Then
    every neighbor of $u$, which is not in $ST(G)$ according to
    Claim~\ref{new-claim-1}, is thus adjacent to $x$.  As $x$ and $u$ are
    both in the chordless path $P'$, vertex $u$ has to be one end of
    it.  More specifically, $u=l'$ and $x = b'_1$.  A further
    consequence is that $u$ is the only vertex in $W'\cap V_I$: the
    argument above applies to any vertex $u'\in W'\cap V_I$, and thus
    $u'=l'=u$.  

    Now we show, for any vertex $w$ in $X \setminus Q$, which is
    nonempty by assumption, it has the same neighbors as $u$ in $W'$,
    and hence $(s':c'_1,c'_2:w,B',r')$ is an AW in $G - Q$,
    contradicting the assumption that $Q$ is an interval deletion set
    to $G$.  Observe that any vertex in $N$ is adjacent to both $u$
    and $w$.  
    \renewcommand{\qedsymbol}{$\lrcorner$} 
    \begin{itemize}
    \item The assumption $w\not\in N$ implies $w\not\sim s$:
      otherwise, $w$ is adjacent to both $s$ and $B$ but not in $N$,
      and we can apply Lemma~\ref{lem:shallow} to $W$ and $w$, which is
      in category ``partial,'' to obtain an AW with strictly smaller
      container.
    \item By the selection of $W$, we have $\mathtt{last}(c'_i)\ge
      \mathtt{first}(b'_d)\ge \mathtt{first}(b_d)$ for both $i=1,2$.
      If $c'_i$, where $i=1$ or $2$, is adjacent to one of $u$ and
      $w$, then (\ref{eq:u}) implies $\mathtt{first}(c'_i)<
      \mathtt{last}(b_{d-5})$; as $B$ is short, $c'_i$ must be in $N$,
      and then adjacent to both $u$ and $w$.
    \item Vertex $b'_1$ ($=x$) is in $N$, hence adjacent to $w$.
    \item By definition, $b'_3\sim b'_2$ and $b'_3\not\sim b'_1$($\in
      N$) imply $\mathtt{last}(b'_2)\ge \mathtt{first}(b'_3) >
      \mathtt{last}(b'_1)\ge \mathtt{first}(b_d)$.  On the other hand,
      $b'_2\not\sim u$ implies $b'_2\not\in N$.  Then as $B$ is short,
      $\mathtt{first}(b'_2) > \mathtt{last}(b_{d-5})$.  Therefore,
      from (\ref{eq:u}) we can conclude for $2\le i\le d'+1$, it holds
      that $\mathtt{first}(b'_i)> \mathtt{last}(w)$ and thus
      $w\not\sim b'_i$. \qedhere
    \end{itemize}
  \end{proof}

  \begin{claim}\label{new-claim-3}
    $c'_2 \in N$.
  \end{claim}
  \begin{proof}
    As $u= b'_i$ for some $0\le i\le d'$,
    Proposition~\ref{lem:at-indices} and (\ref{eq:u}) imply
    $\mathtt{first}(b'_1)\left(\le \mathtt{last}(l')\right) \le
    \mathtt{last}(u) < \mathtt{last}(b_{d-5})$.  By
    Claim~\ref{new-claim-2}, $b'_1$ is not in $N$ and adjacent to at
    most $3$ vertices of $B$; thus $\mathtt{last}(b'_1) <
    \mathtt{first}(b_{d-2}) \le\mathtt{last}(b_{d-3})$.  On the other
    hand, by the selection of $W$, we have $\mathtt{last}(c'_2)\ge
    \mathtt{first}(b'_d)\ge \mathtt{first}(b_d)$.  Therefore, $c'_2$
    is adjacent to at least $4$ vertices of $B$ and is in $N$.
    \renewcommand{\qedsymbol}{$\lrcorner$}
  \end{proof}

  From Claim~\ref{new-claim-3} we can conclude $c'_2\sim u$ and then
  $u\in B'$.  By Proposition~\ref{lem:at-indices},
  $\mathtt{first}(b'_1) \le \mathtt{first}(u) <
  \mathtt{first}(b_{d-4})$.  Then from Claim~\ref{new-claim-1} and the
  fact $l' \sim b'_1$, it can be inferred that $l' \not\in ST(G)$.
  Now $\mathtt{last}(l')$ is defined, and $\mathtt{last}(l') <
  \mathtt{first}(c'_2)\le \mathtt{last}(b_1)$; the selection of $W$
  implies $\mathtt{first}(b'_{d'}) \ge \mathtt{first}(b_{d})$.
  Therefore, the $l'$-$b'_{d'}$ path $l'B'$ has to go through $X$, and
  we end with a contradiction.  This verifies that $Q'$ is an interval
  deletion set to $G$, and it remains to show that $Q'$ is minimum,
  from which the lemma follows.

  \begin{claim}
    $|Q'| \leq |Q|$.
  \end{claim}
  \begin{proof}
    It will suffice to show that $Q \cap V_I$ makes a $b_2$-$b_{d-5}$
    separator in $G - N$, and then the claim ensues as
    $$|Q'| = |Q\setminus V_I| + |X| \leq |Q\setminus V_I| + |Q\cap V_I| =
    |Q|.$$ Suppose to the contrary, there is a chordless
    $b_2$-$b_{d-5}$ path $P$.  We can extend $P$ into an $l$-$r$ path
    $P^+ = (l b_1 P b_{d-4} b_{d-3} b_{d-2} b_{d-1} b_{d} r)$, which
    is disjoint from $Q$ and $N$.  Within $P^+$ there is a chordless
    $l$-$r$ path $(l B_1 r)$.  By assumption, $\{s,c_1,c_2\}\cap
    Q=\emptyset$; every vertex in $B_1$ satisfies the condition of
    Claim~\ref{new-claim-1}, and hence nonadjacent to $s$.  Moreover,
    $c_1,c_2\in N$, and therefore both $c_1$ and $c_2$ are adjacent to
    every vertex of $B_1$.  Thus, $(s:c_1,c_2:l,B_1,r)$ is an AW in
    $G-Q$, which is impossible.
    \renewcommand{\qedsymbol}{$\lrcorner$}
  \end{proof}
\end{proof}
\end{document}